\documentclass{article}
\usepackage[main,final]{neurips_2026}  %

\usepackage[utf8]{inputenc}
\usepackage[T1]{fontenc}
\usepackage{amsmath,amssymb,amsfonts,amsthm}
\usepackage{mathtools}
\usepackage{booktabs}
\usepackage{array}
\usepackage{nicefrac}
\usepackage{microtype}
\usepackage{graphicx}
\usepackage{float}
\usepackage{algorithm}
\usepackage{algorithmic}
\usepackage{hyperref}
\usepackage{url}
\usepackage{cleveref}
\usepackage{xcolor}
\usepackage{subcaption}
\usepackage{multirow}
\usepackage{enumitem}

\newtheorem{theorem}{Theorem}
\newtheorem{proposition}[theorem]{Proposition}
\newtheorem{lemma}[theorem]{Lemma}
\newtheorem{corollary}[theorem]{Corollary}
\newtheorem{assumption}{Assumption}
\theoremstyle{definition}
\newtheorem{definition}[theorem]{Definition}
\newtheorem{decomposition}[theorem]{Decomposition}
\theoremstyle{remark}
\newtheorem{remark}{Remark}

\newcommand{\indep}{\perp\!\!\!\perp}
\newcommand{\nindep}{\not\!\perp\!\!\!\perp}
\newcommand{\E}{\mathbb{E}}

\newcommand{\Xb}{\mathbf{X}}

\newcommand{\hatX}{\hat{\mathbf{X}}}
\newcommand{\tCI}{t_{\text{CI}}}

\title{PAIR-CI: Calibrated Conditional Independence Testing for Causal Discovery with Incomplete Data}

\author{Thomas S.~Robinson \\
  Department of Methodology\\
  London School of Economics and Political Science\\
  London, UK \\
  \texttt{t.robinson7@lse.ac.uk} \\
\And Ranjit Lall \\
  Department of Politics and International Relations\\
  University of Oxford\\
  Oxford, UK \\
  \texttt{ranjit.lall@politics.ox.ac.uk}
}

\begin{document}

\maketitle

\begin{abstract}
    The standard constraint-based paradigm for causal discovery with incomplete data---impute first, test second---is frequently miscalibrated: any consistent conditional independence (CI) test rejects a true null with probability approaching 1 when imputation error induces spurious conditional dependence. We introduce \textit{PAIR-CI}, a nonparametric CI test that restores calibration by integrating multiple imputation directly into the inferential procedure via a paired permutation design. PAIR-CI compares cross-validated models that include and exclude the candidate variable while receiving the same imputed conditioning set, so that imputation error cancels in their loss difference rather than contaminating the test statistic. A variance estimator that nests a provably consistent cross-validation estimator inside Rubin's rules yields inference that is asymptotically valid in the many-imputation limit and conservative by a derived, sign-definite mechanism---to our knowledge, the first formal result combining these two inferential frameworks. In simulations, existing imputation-based CI tests exhibit false positive rates of 28--45\% when data are missing not at random (MNAR), whereas PAIR-CI's pooled rate is at most 3.6\%, exceeding the nominal 5\% only in some cells with five or more conditioning variables. Integrated into the PC algorithm, PAIR-CI is stable across functional forms and missingness mechanisms, whereas every parametric pipeline we consider, including missing-value PC, has a regime in which it fails badly. On nonlinear network benchmarks it reduces structural Hamming distance by 18\% (ALARM) and 11\% (HAILFINDER) relative to missing-value PC and by 31--48\% relative to test-wise deletion and imputation-based Fisher-$Z$ pipelines, while attaining the highest skeleton F1; when relationships are linear, a well-specified parametric pipeline remains preferable.
\end{abstract}

\section{Introduction}
\label{sec:intro}

Conditional independence (CI) testing---determining whether $Z \indep Y \mid \Xb$---is the inferential engine of constraint-based methods for causal discovery across the sciences. Widely used algorithms such as PC (Peter--Clark), FCI (Fast Causal Inference), and their variants reduce the problem of learning causal structure to a sequence of CI tests over variable pairs with progressively expanding conditioning sets \citep{spirtes2000causation}. This creates a fundamental problem when data are incomplete: missing values in the conditioning set $\Xb$ prevent direct evaluation of relationships between test variables. Standard remedies are problematic. Discarding incomplete observations---either globally (complete-case analysis) or per test (test-wise deletion; \citealt{tu2019causal})---sacrifices data, yields different effective samples across tests, and leads to selection bias in the common scenario that data are not missing completely at random (MCAR) \citep{lall2016multiple}. Imputing missing values before testing is valid under MCAR or missing at random (MAR) but can lead to miscalibration under MNAR and other forms of misspecified imputation: residual bias in the completed data may introduce spurious associations between test variables that are mistaken for genuine conditional dependence. We formalize this bias in Proposition~\ref{prop:impossibility}: miscalibration is not a small-sample artifact but an asymptotic inevitability whenever imputation error induces residual conditional dependence between test variables. The result is a causal graph populated with spurious edges that propagate errors through edge orientation, compromising the reliability of downstream inference.

To avoid this miscalibration, we propose \textit{PAIR-CI}, a nonparametric CI test that integrates multiple imputation directly into the inferential procedure via a paired permutation design. PAIR-CI compares cross-validated models of $Z$ that differ in access to the candidate variable but receive the same imputed conditioning set, ensuring that any distortion introduced by imputation affects both models identically and thus cancels in their loss difference rather than entering the test statistic. Under the regularity conditions of Section~\ref{sec:theory}, PAIR-CI has asymptotic size at most $\alpha$ for the internal null $Z \indep Y \mid \hatX$ in the limit of infinitely many imputations, with the conservatism governed by the share of the combined variance that draw-to-draw noise contributes. The test's paired structure extends this calibration guarantee to the scientific null $Z \indep Y \mid \Xb$ under MAR with a consistent imputer, up to a measured caching residual, and we measure rather than assume how far it extends under MNAR.

A key inferential challenge is that PAIR-CI's test statistic combines two distinct sources of dependence: overlap between training sets in $K$-fold cross-validation; and variation across $M$ multiply imputed datasets. Whereas standard variance estimators address only one source at a time, we simultaneously account for both by nesting the provably consistent within-imputation estimator developed by \citet{bayle2020cross} inside Rubin's rules \citep{rubin1987multiple}. The resulting combined estimator delivers inference that is asymptotically valid in the many-imputation limit---conservative, with the direction of the bias derived rather than observed---to our knowledge, the first formal result connecting cross-validation and multiple imputation.

To summarize, we make three principal contributions:
\begin{enumerate}[leftmargin=*,nosep]
    \item \textbf{A nonparametric CI test with calibration under misspecification.} We develop a classification-based CI test with incomplete $\Xb$ that integrates multiple imputation, cross-validation, and conditional permutation into a single inferential framework. PAIR-CI has asymptotic size at most $\alpha$ for the internal null under the conditions of Proposition~\ref{prop:calibration}, with power tending to 1 if $Y$ improves prediction (Proposition~\ref{prop:consistency}). The paired design cancels imputation error by feeding both models the same imputed conditioning set, empirically maintaining calibration for the scientific null without requiring correct specification of the imputation model (Remark~\ref{rem:paired_robust}); among the methods we consider, only PAIR-CI achieves an average false positive rate $\leq 4\%$ across data-generating processes (DGPs) and missingness mechanisms.
    \item \textbf{Unified variance estimation for cross-validation and multiple imputation.} Reliable inference requires accounting for both fold overlap in cross-validation and variation across multiply imputed datasets, yet existing combinations of the two rely on naive pooled variances. We show that embedding \citeauthor{bayle2020cross}'s \citeyearpar{bayle2020cross} cross-validation variance approximation within Rubin's rules yields an asymptotically valid, conservative estimator (Theorem~\ref{thm:bridge}) that improves power by 9--10 percentage points at intermediate effect sizes and enables $K=10$ folds without sacrificing calibration.
    \item \textbf{Stable causal discovery across functional forms and missingness mechanisms.} When integrated into the PC algorithm, PAIR-CI's graph recovery is stable across linear and nonlinear edges and across MAR, MNAR, and mixed mechanisms, whereas every parametric pipeline we consider, including missing-value PC \citep[MVPC;][]{tu2019causal}, has a regime in which it fails badly. Where relationships are nonlinear, PAIR-CI attains the lowest structural Hamming distance (SHD) and the highest skeleton F1 on the ALARM ($p=37$) and HAILFINDER ($p=56$) networks at 20\% missingness, with the highest recall on ALARM; where they are linear, MVPC is the best method overall, and a well-specified parametric pipeline is preferable when its assumptions hold.
\end{enumerate}

\section{Background and Related Work}
\label{sec:background}

\paragraph{Constraint-based causal discovery.} The PC algorithm \citep{spirtes2000causation} recovers causal structure from observational data as a completed partially directed acyclic graph (CPDAG)---a representation of all directed acyclic graphs (DAGs) consistent with observed CI relations---in two phases. The \emph{skeleton phase} begins with a complete undirected graph and removes edges by testing $X_i \indep X_j \mid \mathbf{S}$ for conditioning sets $\mathbf{S}$ of increasing size drawn from the adjacency of $X_i$. The \emph{orientation phase} then directs edges by identifying v-structures and applying Meek's \citeyearpar{meek1995causal} rules to propagate orientation constraints through the graph.

\paragraph{Conditional independence tests.} Traditional CI tests include Fisher's $Z$-test for Gaussian data and $\chi^2$ or $G$-tests for discrete data. In nonparametric settings, kernel-based tests such as KCI \citep{zhang2011kernel} and the Randomized Conditional Independence Test \citep[RCIT;][]{strobl2019approximate} assess CI by embedding variables into high-dimensional feature spaces, while the Generalized Covariance Measure \citep[GCM;][]{shah2020hardness} measures dependence between regression errors. Recently developed classifier-based tests compare the predictive performance of models that include and exclude the candidate variable \citep{sen2017model,watson2021testing,bellot2019conditional}.

\paragraph{Incomplete data in causal discovery.} Test-wise deletion restricts each CI test in the skeleton phase to complete observations across the variables of interest, inducing selection bias under MAR and MNAR. Missing-value PC \citep[MVPC;][]{tu2019causal} builds on this first stage: it learns a skeleton with test-wise-deletion Fisher's $Z$, detects the direct causes of each missingness indicator by test-wise-deletion CI tests on the indicators, flags as potentially extraneous any edge whose endpoints share a common adjacent variable or missingness indicator, and re-tests those edges on virtual data regenerated from the test-wise-deleted sample. MVPC is therefore the principled constraint-based competitor for missing data, of which plain test-wise deletion is the uncorrected first stage, and we include it as a baseline in Section~\ref{sec:experiments}. Parametric tests based on multiple imputation, most notably the FZ-Rubin pipeline proposed by \citet{foraita2020causal} and studied in depth by \citet{witte2022multiple}, pool Fisher's $Z$ across imputations using Rubin's rules \citep{rubin1987multiple}, yielding valid inference when the imputation model is correctly specified and test variables are jointly Gaussian. Score-based methods such as MissDAG \citep{gao2022missdag} maximize the observed-data likelihood by expectation--maximization under ignorable missingness and an identifiable additive-noise model, at the cost of those assumptions. Although often unknown in advance, the missingness mechanism can also be encoded explicitly via $m$-graph modeling \citep{gain2018structure,mohan2021graphical}.

\section{Method}
\label{sec:method}

\subsection{PAIR-CI Testing Procedure}
\label{sec:ci_test}

\paragraph{Problem setting.} PAIR-CI tests $H_0\colon Z \indep Y \mid \Xb$ for a response $Z$, a candidate variable $Y$, and a conditioning set $\Xb$ (within PC, the candidate separating set), which throughout Section~\ref{sec:method} is the only part of the data assumed incomplete: $R$ is the response indicator, with $R_{ij} = 1$ if and only if $X_{ij}$ is observed, $\Xb = (\Xb_{\mathrm{obs}}, \Xb_{\mathrm{miss}})$ is partitioned accordingly, and $\hatX$ denotes a completed conditioning set. The full model predicts $Z$ from $(\hatX, Y)$ and the partial model from $(\hatX, U)$, where $U$ is a placebo for $Y$ obtained by conditional permutation; both are scored on $Z$, and the test rejects when $Y$ carries predictive information about $Z$ beyond $\hatX$. We distinguish the \emph{internal null} $Z \indep Y \mid \hatX$, which concerns the completed data, from the \emph{scientific null} $Z \indep Y \mid \Xb$; Appendix~\ref{app:notation} collects the notation.

\begin{figure}[t]
    \centering
    \caption{\textbf{Schematic of PAIR-CI testing procedure}.}
    \label{fig:method}
    \includegraphics[width=0.85\textwidth]{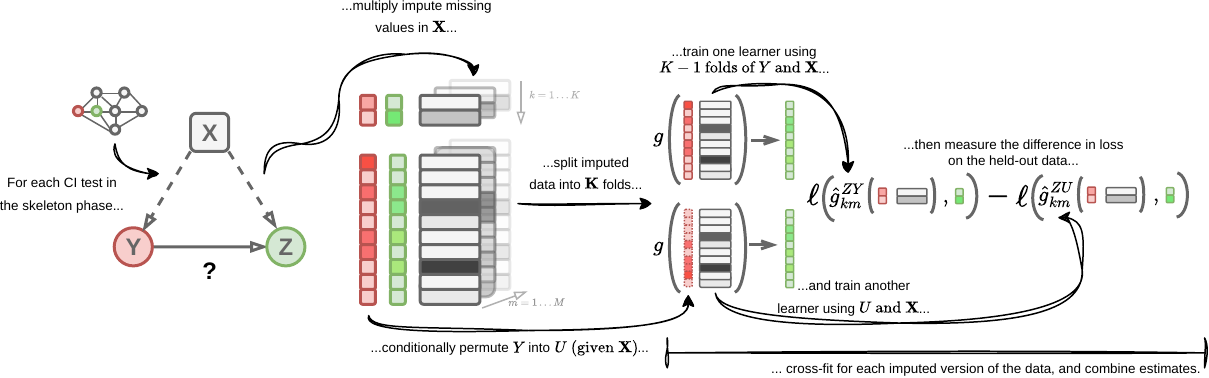}
\end{figure}

\paragraph{Step 1: Multiple imputation.} We generate $M$ completed datasets $\hat{D}^{(1)}, \ldots, \hat{D}^{(M)}$ by multiply imputing missing values in $\Xb$. In principle, $Z$ and $Y$ should be excluded from the imputation model in the per-query setting (i.e., when imputing separately for each CI test) to avoid spurious associations between $\hat{\Xb}$ and other test variables. Our default imputation method is MICE (Multiple Imputation by Chained Equations) with Bayesian ridge regression \citep{van2011mice}; any procedure satisfying the properness condition (A3) of Section~\ref{sec:theory} can be substituted, including flexible machine-learning imputers \citep{lall2022midas}.

\paragraph{Step 2: Cross-validated model comparison.} We partition each completed dataset $\hat{D}^{(m)}$ into $K$ folds. For each fold $k$, we train two models of $Z$ on the remaining $K-1$ folds: a full model $\hat{g}^{ZY}_{km}: (\hatX, Y) \to Z$, which includes both $Y$ and $\hatX$ as predictors; and a partial model $\hat{g}^{ZU}_{km}: (\hat{X}, U_k) \to Z$, which replaces $Y$ with a fold-specific placebo column $U_k$ constructed by conditionally resampling $Y$ from its $k_{\text{nn}}$ nearest neighbors in $\hat{X}$-space on the training fold $T_{km}$ \citep{runge2018conditional} (neighborhood size in Appendix~\ref{app:knn}). Conditional permutation preserves $Y$--$\Xb$ dependence while breaking the $Y$--$Z \mid \Xb$ link. The placebo ensures that both models receive the same number of input features, equalizing finite-sample regularization behavior. When the conditioning set is empty or neighborhood size $k_{nn} < 2$, we apply unconditional random permutation.

\paragraph{Step 3: Loss comparison.} We evaluate the full and partial models on held-out fold $k$ of imputed dataset $m$, computing the difference in out-of-sample loss:
\begin{equation}
    \hat{\mu}_{km} = \ell\!\left(\hat{g}^{ZU}_{km}\right) - \ell\!\left(\hat{g}^{ZY}_{km}\right)
    \label{eq:fold_diff}
\end{equation}
where $\ell$ is binary cross-entropy for discrete $Z$ ($\leq 20$ unique values) or mean squared error for continuous $Z$.

\paragraph{Step 4: Combined test statistic.} The overall loss difference is estimated by averaging across folds and imputations: \begin{equation}
    \hat{\mu} = (KM)^{-1} \sum_k \sum_m \hat{\mu}_{km}.
\end{equation}
To account for dependence arising from both fold overlap and variation across imputed datasets, we combine the within-imputation cross-validation variance $\bar{W}$ \citep[Theorem~4]{bayle2020cross} with the between-imputation variation $B$ using Rubin's rules \citep{rubin1987multiple}:
\begin{equation}
    T = \bar{W} + (1 + 1/M) B,
\end{equation}
where
\begin{equation}
    \bar{W} = M^{-1} \sum_m n^{-1} K^{-1} \sum_k \hat{\sigma}^2_{k,m} \quad  \text{and} \quad B = (M-1)^{-1} \sum_m (\hat{\mu}_m - \hat{\mu})^2.
    \label{eq:variance_components}
\end{equation}
Here $\hat{\sigma}^2_{k,m}$ is the Bessel-corrected sample variance, over the held-out observations $i \in F_k$ of imputed dataset $m$, of the per-observation loss differences $d_{ikm} = \ell(z_i, \hat{g}^{ZU}_{km}) - \ell(z_i, \hat{g}^{ZY}_{km})$ whose fold means are the $\hat{\mu}_{km}$ of Equation~\ref{eq:fold_diff}; averaged over folds and divided by $n$, it is \citeauthor{bayle2020cross}'s consistent estimator of the variance of the cross-validated mean $\hat{\mu}_m$ within imputation $m$, $\bar{W}$ averages this quantity over imputations, and $B$ is Rubin's between-imputation variance. The PAIR-CI test statistic $\tCI = \hat{\mu}/\sqrt{T}$ is compared against a $t$-distribution with Barnard--Rubin--adjusted degrees of freedom. We use a one-tailed test, since evidence against $H_0$ always requires the full model to outperform the partial model.

\paragraph{Choice of learner.} Appendix~\ref{app:learner_choice} describes the \textit{general} and \textit{fast} variants of PAIR-CI used in our experiments.

\subsection{PC Integration for Causal Discovery}
\label{sec:computation}

We use PAIR-CI as the oracle of a simplified PC skeleton search (Appendix~\ref{app:implementation}), with one modification for computational efficiency: given the high cost of re-imputing $\Xb$ for all $O(p^2)$ CI tests, we cache $M$ completed versions of $D$ upfront and reuse them across all tests, masking test variables from each query. Counterintuitively, including $Z$ and $Y$ in the cached imputation model improves robustness to MNAR by allowing their correlation with missing components of $\Xb$ to be absorbed into $\hatX$, thereby reducing residual imputation bias $\kappa$ (Appendix~\ref{app:kappa}). Adversarial stress tests across six synthetic DGPs with MNAR missingness confirm that caching maintains calibration (Appendix~\ref{app:adversarial_kappa}). For DAGs with $p \in \{5, 10, 20\}$, the cached and per-query strategies agree on 96.0--97.1\% of CI decisions (Appendix~\ref{app:kappa}). Full pseudocode for the modified PC algorithm is presented in Appendix~\ref{app:implementation}. When the test variables $Z$ or $Y$ are themselves incomplete, we impute their missing entries from the same cached MICE model rather than deleting those rows; Appendix~\ref{app:testvar} shows that this rule keeps the standalone false positive rate between 0.000 and 0.068 across 16 null configurations with 30\% missingness on $Z$, $Y$, both, or $\Xb$ and $Z$ jointly.

\section{Theoretical Foundations and Guarantees}
\label{sec:theory}

Full proofs are given in Appendix~\ref{app:proofs}.

\begin{proposition}[Miscalibration of impute-then-test]
    \label{prop:impossibility}
    Consider the null $H_0\colon Z \indep Y \mid \Xb$, where $\Xb = (\Xb_{\mathrm{obs}}, \Xb_{\mathrm{miss}})$ contains missing entries (MAR or MNAR), and let $\hatX = (\Xb_{\mathrm{obs}}, \hatX_{\mathrm{miss}})$ denote the completed conditioning set. Suppose that (a) the imputation model is asymptotically misspecified: $\|P(\hatX_{\mathrm{miss}} \mid \Xb_{\mathrm{obs}}) - P(\Xb_{\mathrm{miss}} \mid \Xb_{\mathrm{obs}})\|_{\textup{TV}} \geq \varepsilon$ for some $\varepsilon > 0$ and all sufficiently large $n$; and (b) imputation error induces spurious dependence ($Z \nindep Y \mid \hatX$ under the imputed distribution). Then any consistent CI test applied to the completed data $\hat{D}_n$ rejects $H_0$, at any fixed level $\alpha$, with probability tending to one. Remark~\ref{rem:conditions} (Appendix~\ref{app:proofs}) discusses when (a) and (b) hold; Appendix~\ref{app:worked_example} derives (b) in closed form.
\end{proposition}

\paragraph{Assumptions.} To characterize when PAIR-CI avoids the miscalibration established in Proposition~\ref{prop:impossibility}, we require subsets of the following conditions: (A1)~\emph{imputation consistency}: the imputation procedure converges to the true conditional distribution of the missing entries given all observed data; (A1$'$)~\emph{common-limit imputation}: the $M$ imputed datasets converge to a common, not necessarily correct, limit law (Assumption~\ref{ass:cross_imp}, Appendix~\ref{app:bridge}); (A2)~\emph{bounded loss}: $\ell \in [0, L_\ell]$ for some $L_\ell < \infty$; (A3)~\emph{proper imputation}: imputed values are drawn from a posterior predictive distribution that correctly incorporates parameter uncertainty \citep{rubin1987multiple}; (A4)~\emph{learner stability}: the conditional variance convergence, loss stability, and uniform integrability conditions of \citet[Theorems~3--4]{bayle2020cross}, which we assume for the paired loss difference (Appendix~\ref{app:bridge}); (A5)~\emph{learner consistency}: the learner converges to the Bayes-optimal prediction function for the DGP at hand, which \citet[Theorem~1]{scornet2015consistency} establish for random forests under an additive model with uniform design and Gaussian noise and which we assume more broadly; and (A6)~\emph{faithfulness}: every CI relation in the complete-data distribution corresponds to a $d$-separation in the true DAG. (A1) implies (A1$'$).

\begin{proposition}[Calibration]
    \label{prop:calibration}
    Under Assumptions (A1$'$) and (A2)--(A4) and the conditions of Theorem~\ref{thm:bridge}, including (N2), PAIR-CI has asymptotic size at most $\alpha$ for the internal null $Z \indep Y \mid \hatX$ in the limit of infinitely many imputations, $\lim_{M} \limsup_{n} P(\tCI > t_{\alpha,\nu}) \leq \alpha$, with equality when the draw-to-draw component of the variance of the loss difference is negligible relative to its sampling component.\footnote{At fixed $M$ the bound rests on the Barnard--Rubin reference distribution dominating Rubin's $t$-type limit law, which we examine empirically in Appendix~\ref{app:variance_ratio}.} Under Assumption~(A1), read as a statement about the cached imputer of Section~\ref{sec:computation}, the same bound holds for the scientific null $Z \indep Y \mid \Xb$ up to a caching residual that we measure rather than bound (Appendix~\ref{app:kappa}). The proof (Appendix~\ref{app:proofs}) shows by exchangeability that the population loss difference is zero under the internal null with an exact conditional-law placebo, and of order $\varepsilon_n$ under the local permutation, and then applies Theorem~\ref{thm:bridge}; Remark~\ref{rem:paired_robust} (Appendix~\ref{app:kappa}) explains how the paired design achieves robustness under MNAR.
\end{proposition}

\begin{proposition}[Consistency]
    \label{prop:consistency}
    Under Assumptions (A1)--(A2) and (A4)--(A5), PAIR-CI is consistent against alternatives in which $Y$ improves the optimal prediction of $Z$: $P(\tCI > t_{\alpha,\nu}) \to 1$ as $n \to \infty$.\footnote{Assumption (A5) holds, for instance, for Breiman's random forests under the additive model of \citet[Theorem~1]{scornet2015consistency}, with subsample size $a_n$ and $t_n$ terminal nodes satisfying $a_n, t_n \to \infty$ and $t_n (\log a_n)^9 / a_n \to 0$.}
\end{proposition}

\begin{theorem}[Valid, conservative inference under cross-validation and multiple imputation]
    \label{thm:bridge}
    Suppose that Assumptions (A1$'$) and (A2)--(A4) hold, together with the regularity conditions of Assumption~\ref{ass:regularity}, the first-order internal-null stability condition (N1), and the design conditions of Lemma~\ref{lem:sensitivity} (per-query imputation, the placebo coupling, and a local permutation whose discrepancy $\varepsilon_n$ vanishes), all stated in Appendix~\ref{app:bridge}, with the internal null read as $Z \indep Y \mid (\hatX, R)$. Let $\mu_n = \E[\hat\mu_m]$, let $\bar h^*$ be the limiting influence function of Assumption~\ref{ass:cross_imp}, and let $W_{\mathrm{obs}}$ be the observed part of an observation. As $n \to \infty$ with $K$ fixed:
    \begin{enumerate}[leftmargin=*,nosep]
        \item[(i)] for fixed $M$, $\sqrt{n}(\hat{\mu} - \mu_n) \xrightarrow{d} \mathcal{N}(0, \Sigma_M)$, with $\Sigma_M$ the sandwich variance of \citet[Theorem~1]{robins2000inference}, which decreases in $M$ to $\Sigma_\infty$ when the imputation sensitivity vanishes; under the internal null and as $M \to \infty$, $nT \xrightarrow{P} \Sigma_{\mathrm{Rubin}}$, where
        \begin{equation}
            \Sigma_{\mathrm{Rubin}} - \Sigma_\infty \;=\; 2\,\E\!\left[\mathrm{Var}\!\left(\bar{h}^*(W) \mid R,\, W_{\mathrm{obs}}\right)\right] \;\geq\; 0,
            \label{eq:overshoot}
        \end{equation}
        the scalar case of \citet[Theorem~2]{robins2000inference} once the sensitivity of the estimating function to the imputation parameter vanishes; for fixed $M$ and under the internal null, $\E[T] - \mathrm{Var}(\hat\mu)$ equals the same quantity divided by $n$ to first order, because $T$ charges the draw-to-draw noise at $(1 + 1/M)$ against the $1/M$ share at which it enters $\mathrm{Var}(\hat\mu)$ and its within-imputation component exceeds the sampling variance;
        \item[(ii)] if in addition $\sqrt{n}\,\mu_n \to 0$ under the internal null, which holds under (N2), $\sqrt{n}\,\varepsilon_n \to 0$, then $\lim_{M} \limsup_{n} P(\tCI > t_{\alpha,\nu}) \leq \alpha$, with equality when the draw-to-draw component of the variance is negligible relative to its sampling component.
    \end{enumerate}
    Appendix~\ref{app:bridge} gives the full statement, including the form of $\Sigma_M$ and the case of fixed $M$, where $nT$ has a random limit and the size bound rests on the Barnard--Rubin reference distribution, which we do not prove.
\end{theorem}

\paragraph{Proof sketch.} By \citet{bayle2020cross}, $\hat{\mu}_m$ is asymptotically linear with a consistent within-fold variance estimate (Lemma~\ref{lem:within_ral}), which places $\hat{\mu}$ within the estimating-equation framework of \citet{robins2000inference}: their Theorem~1 gives asymptotic normality, and their Theorem~2 characterises the large-sample bias of Rubin's variance estimator, in which every sign-ambiguous term is multiplied by the sensitivity of the estimating function to the imputation parameter. Under the internal null the population full and partial learners coincide, so the paired loss difference is antisymmetric under the swap of $Y$ and its placebo, which zeroes the score channel of that sensitivity (Lemma~\ref{lem:sensitivity}), while (N1) zeroes the channel through the trained learners; what remains is the non-negative term in Equation~\ref{eq:overshoot}, and Slutsky's lemma gives part~(ii).

\begin{corollary}[PC consistency]\label{cor:pc_consistency}
    Suppose there is a sequence of levels $\alpha_n \to 0$ along which the PAIR-CI oracle's type-I and type-II error probabilities both vanish for every conditioning set the PC algorithm visits. Then, under Assumptions \emph{(A1)--(A6)} and the conditions of Theorem~\ref{thm:bridge}, the PC algorithm \citep[Algorithms~1--2]{kalisch2007estimating} equipped with that oracle recovers the true CPDAG, $\hat{G} \xrightarrow{P} G_{\mathrm{CPDAG}}$, by \citet[Theorem~2]{kalisch2007estimating}. Propositions~\ref{prop:calibration} and~\ref{prop:consistency} give fixed-level size and, when $Y$ improves prediction, power tending to one; extending these to a shrinking sequence requires versions uniform in $\alpha$, which we do not establish (Appendix~\ref{app:proofs}). Only an upper bound on the type-I error is used, so the weakening of Theorem~\ref{thm:bridge} to conservative size does not propagate. (See Appendix~\ref{app:skeleton_fixed_alpha} for a fixed-level result and the search our experiments use.)
\end{corollary}

\paragraph{Scope, and assumptions in practice.} We distinguish what we prove from what we measure. Analytically, PAIR-CI is conservative for the internal null under (A1$'$), (A2)--(A4), Assumption~\ref{ass:regularity}, conditions (N1) and (N2), and the per-query design with an accurate local permutation, in the limit of infinitely many imputations; at fixed $M$ the bound is examined empirically. The extension to the scientific null under MAR requires a consistent imputer that conditions on $Y$ and $Z$ (A1), the cached design, under which Lemma~\ref{lem:sensitivity} holds only up to a residual we measure (Appendix~\ref{app:kappa}). Under MNAR, where A1 is false by construction for self-masking mechanisms, the paired design attenuates but does not eliminate the consequences, and we measure the residual rather than assume it away: across the 48 adversarial cells of Appendix~\ref{app:adversarial_kappa}, the estimated product $\hat\kappa_{\mathrm{imp}}^2 |\hat\kappa_Y| |\hat\kappa_Z|$ predicts the false positive rate (Spearman correlation 0.89; Figure~\ref{fig:kappa_fpr_trace}), and the benchmark MNAR settings are calibrated because $|\mathrm{corr}(\delta, Y)| \leq 0.011$ in every cell (Appendix~\ref{app:imputer_quality}); under MNAR, power as well as size rests on this empirical tier. Each assumption has an implementation counterpart, a diagnostic, and a failure behavior (Table~\ref{tab:assumptions_practice}, Appendix~\ref{app:assumptions_practice}); in particular, calibration requires learner stability (A2--A4), not learner accuracy, so a weak learner costs power, not size (Appendix~\ref{app:learner_robustness}).

\paragraph{Theoretical contribution.} The ingredients of our analysis are established results, and we claim no new asymptotic machinery; the contribution lies in their composition. Four pieces are, to our knowledge, new: Proposition~\ref{prop:impossibility}, which makes the failure of impute-then-test structural; the paired-placebo exchangeability argument of Proposition~\ref{prop:calibration}, which needs only boundedness and stability of the learner and so admits black-box learners; the variance bridge of Lemma~\ref{lem:within_ral} and Theorem~\ref{thm:bridge}; and the conservatism mechanism of Theorem~\ref{thm:bridge}(i) (Appendix~\ref{app:variance_ratio}).

\section{Experiments}
\label{sec:experiments}

A smaller, approximately linear benchmark (Sachs, $p = 11$) is discussed in Appendix~\ref{app:sachs}.\footnote{All code and data required to reproduce our results are included in the supplementary material and available at \url{https://github.com/tsrobinson/pair-ci}. The implementation is compatible with the \texttt{causal-learn} library in Python.}

\subsection{Standalone Performance}
\label{sec:exp_ci}

\paragraph{Setup.} We test $H_0$ under three DGPs: \emph{linear Gaussian} ($Z$ linear in $\Xb$ plus noise); \emph{post-nonlinear} ($Z = \sigma(\Xb^\top \beta + \varepsilon)$); and \emph{latent confounder} ($Z = \sin(X_1) + X_2^2 + 0.5L + \varepsilon_z$, $Y = \text{signal} \cdot L + \cos(X_3) + \varepsilon_y$, $L \sim \mathcal{N}(0,1)$ unobserved). Signal strength---the coefficient governing $Y$'s influence on $Z$ (or $L$)---varies in $\{0, 0.3, 0.6, 1.0\}$, with $n \in \{500, 1{,}000, 2{,}000, 5{,}000\}$ and $|\Xb| \in \{2, 5, 10\}$ (the cell $n = 500$, $|\Xb| = 10$ is not run), 100 replicates per configuration. Three missingness regimes are considered: complete data, MAR, and MNAR. For each incomplete variable $X_j$, entry $X_{ij}$ is masked independently with probability
\begin{equation}
    \pi_j(v_i) = \sigma\bigl(3\,(v_i - q_{1-\rho}(v))\bigr), \qquad \sigma(u) = (1 + e^{-u})^{-1},
    \label{eq:mechanism}
\end{equation}
where $q_{1-\rho}(v)$ is the empirical $(1-\rho)$-quantile of the driver $v$, so that the marginal rate is calibrated to the target $\rho$ and missingness concentrates in the driver's upper tail. Under MAR, $v = W$, a designated fully observed covariate outside the incomplete set; under MNAR, $v = X_j$ itself (self-masking); the mixed mechanism of the graph experiments alternates between the two. Two of the $|\Xb|$ conditioning variables are incomplete at a per-column rate $\rho = 0.3$ (one at $|\Xb| = 2$ under MAR, the other being the driver), so that 6--15\% of conditioning-set entries are missing overall (30\% under MNAR at $|\Xb| = 2$); Appendix~\ref{app:adversarial_kappa} covers other MNAR geometries. The imputation model throughout is MICE with Bayesian ridge regression and posterior draws over all variables, so that, as in the pipelines of \citet{witte2022multiple}, the test variables enter as auxiliary predictors. We compare PAIR-CI against Fisher's $Z$ with single imputation (FZ-single) and with Rubin-pooled multiple imputation (FZ-Rubin), GCM with random-forest regressions, and KCI in its default and a tuned configuration (exact configurations, the audit of alternatives, and the rationale for excluding RCIT are in Appendix~\ref{app:audit}), each applied both to complete cases and to imputed data.

\paragraph{Results.} Table~\ref{tab:calibration} reports rejection rates under $H_0$ pooled over DGPs, sample sizes, and conditioning-set sizes (Panel~A) and PAIR-CI's rates by sample size (Panel~B). PAIR-CI's false positive rate is at most 0.036 under every missingness condition, and only one DGP $\times$ mechanism cell exceeds the nominal level (0.070, linear Gaussian under MNAR; Table~\ref{tab:calibration_by_dgp}), an order of magnitude below the 41--59\% maxima of the imputed pipelines. All imputation-based baselines fail under MNAR, with DGP-averaged false positive rates of 28--45\%. Under MAR with a proper imputer, Rubin's rules nearly restore calibration for Fisher's $Z$ (0.058), as Rubin's theory predicts, whereas the flexible pipelines do not (GCM 0.254, KCI 0.329). The failure we target is therefore a property of the impute-then-test pipeline, not of missingness handling as such: complete-case Fisher's $Z$ is calibrated under every mechanism (0.049--0.050), because selecting rows on the conditioning variables leaves the conditional law of $(Z, Y)$ given $\Xb$ intact, at the cost of roughly half the rows and, inside PC, the selection effects of Section~\ref{sec:exp_graph}; this relies on the standalone design confining missingness to $\Xb$, since with missingness on the test variables themselves deletion is again selection (Appendix~\ref{app:testvar}). GCM and KCI remain miscalibrated under deletion (0.215 and 0.147 on complete data) because their complete-data behavior is the binding problem: GCM's rate rises with $n$ and $|\Xb|$ under every forest configuration we audit, the signature of the regularity condition of \citet{shah2020hardness} failing for forest rates, and tuning halves KCI's rate to 0.072 without making it nominal everywhere (Appendix~\ref{app:audit}).

\paragraph{Calibration and power by sample size.} Panel~B of Table~\ref{tab:calibration} shows the false positive rate rising toward the nominal level with $n$ (0.007 to 0.026 on complete data; 0.018 to 0.034--0.049 under MNAR) without exceeding 0.05 in any mechanism $\times$ $n$ cell, the small-$n$ conservatism being consistent with Theorem~\ref{thm:bridge} (Appendix~\ref{app:variance_ratio}), while power at signal 0.3 rises correspondingly (under MNAR, 0.28, 0.54, and 0.75 at $n = 500$, $1{,}000$, and $2{,}000$) and exceeds 80\% at signal $\geq 0.6$ and $n \geq 2{,}000$ under every DGP (Appendix~\ref{app:power_plots}). By $|\Xb|$ the rate rises with the conditioning dimension, to at most 0.067 at the mechanism $\times$ $|\Xb|$ grain and 0.077 at the mechanism $\times$ $n$ $\times$ $|\Xb|$ grain (Table~\ref{tab:calibration_full}), so it is mildly above nominal only in cells with $|\Xb| \geq 5$. For sample-size planning, 80\% power at signal $0.6$ and $|\Xb| = 5$ requires $n \approx 500$, $1{,}000$, or $2{,}000$ for the linear Gaussian, latent confounder, and post-nonlinear DGPs (Appendix~\ref{app:power_planning}); power comparisons with miscalibrated baselines are uninformative.

\begin{table}[t]
    \centering
    \caption{\textbf{Standalone calibration, overall and by sample size.} Panel~A: rejection rate under $H_0$ at $\alpha = 0.05$, averaged over three DGPs, $n \in \{500, \ldots, 5{,}000\}$ and $|\Xb| \in \{2, 5, 10\}$ (100 replicates per cell; KCI 50, at $n = 1{,}000$); the Complete column is complete-case; $^\dagger$tuned KCI, complete-data audit grid (Appendix~\ref{app:audit}; default 0.144). Panel~B: PAIR-CI by mechanism and $n$ (the $n = 500$ row pools $|\Xb| \in \{2, 5\}$; the learner subsamples to 2{,}000 rows). Values above $\alpha$ in \textcolor{red}{red}; disaggregations in Tables~\ref{tab:calibration_by_dgp} and~\ref{tab:calibration_full}. Binomial Monte Carlo SEs are at most 0.023 (0.003 for PAIR-CI in Panel~A).}
    \label{tab:calibration}
    \footnotesize
    \textbf{Panel A: pooled}\\[2pt]
    \begin{tabular}{llccc}
    \toprule
    Test & Variant & Complete & MAR & MNAR \\
    \midrule
    PAIR-CI        & integrated (MICE, $M{=}5$)   & 0.018 & 0.021 & \textbf{0.036} \\
    \midrule
    Fisher-$Z$     & complete-case                & 0.050 & 0.049 & 0.050 \\
                   & single imputation            & -- & \textcolor{red}{0.103} & \textcolor{red}{0.350} \\
                   & Rubin's rules                & -- & \textcolor{red}{0.058} & \textcolor{red}{0.282} \\
    \midrule
    GCM (RF)       & complete-case                & \textcolor{red}{0.215} & \textcolor{red}{0.192} & \textcolor{red}{0.191} \\
                   & imputed                      & -- & \textcolor{red}{0.254} & \textcolor{red}{0.415} \\
    \midrule
    KCI (default)  & complete-case                & \textcolor{red}{0.147} & \textcolor{red}{0.136} & \textcolor{red}{0.158} \\
                   & imputed                      & -- & \textcolor{red}{0.329} & \textcolor{red}{0.447} \\
    KCI (tuned)$^\dagger$ & complete-case & \textcolor{red}{0.072} & -- & -- \\
    \bottomrule
\end{tabular}
\\[5pt]
    \textbf{Panel B: PAIR-CI by sample size}\\[2pt]
    \begin{tabular}{llcccc}
    \toprule
     & Mechanism & $n{=}500$ & $n{=}1{,}000$ & $n{=}2{,}000$ & $n{=}5{,}000$ \\
    \midrule
    FPR (signal 0)       & Complete  & 0.007 & 0.016 & 0.022 & 0.026 \\
                         & MAR       & 0.000 & 0.021 & 0.027 & 0.029 \\
                         & MNAR      & 0.018 & 0.037 & 0.049 & 0.034 \\
    \midrule
    Power (signal 0.3)   & MNAR      & 0.282 & 0.542 & 0.748 & 0.692 \\
    \bottomrule
\end{tabular}

\end{table}

\subsection{Synthetic Graph Recovery}
\label{sec:exp_graph}

\paragraph{Setup.} We generate 10 random Erd\H{o}s--R\'enyi DAGs at three graph sizes: $p=10$ (edge probability 0.25, 3 incomplete variables), $p=20$ (edge probability 0.2, 6 incomplete variables), and $p=30$ (edge probability 0.15, 8 incomplete variables). We sample $n = 1{,}000$ observations with linear Gaussian and nonlinear edges, inject missingness under MAR, MNAR, and mixed mechanisms at a per-column rate of 30\% (about 9--11\% of all entries at $p = 10$), and obtain 20 replicates per graph for each condition. We compare six constraint-based methods---PAIR-CI (with imputation caching), complete-case PC, test-wise deletion, Fisher's $Z$-based PC applied to multiply imputed data with majority-vote edge aggregation (FZ-vote), FZ-Rubin, and MVPC---together with the score-based MissDAG. MVPC uses its authors' released implementation in \texttt{causal-learn} with the \texttt{MV\_Crtn\_Fisher\_Z} permutation-based correction, the only one released; because that correction is a whole-algorithm procedure rather than a per-query CI test, MVPC runs in its authors' own PC implementation on datasets bit-identical to the other methods' (on HAILFINDER, matched in distribution), whereas all other constraint-based methods share our PC harness. KCI is excluded on cost alone ($O(n^3)$ per test). We deploy the general variant of PAIR-CI ($M = 5$, $K = 10$) for $p = 10$ and the fast variant ($M = 5$, $K = 5$) for $p \geq 20$,\footnote{For $K>10$, further power gains require a stronger learner or larger $n$ rather than additional imputations (Appendix~\ref{app:m_sensitivity}).} and report total SHD against the true DAG (Table~\ref{tab:scaling}), skeleton SHD, skeleton F1, and the median number of output edges (precision, recall, runtimes, and robustness to the missingness rate and $n$ are in Appendix~\ref{app:robustness}).

\paragraph{Results.} Table~\ref{tab:scaling} summarizes performance across scales. At $p = 10$, PAIR-CI and MVPC tie at SHD~6 with linear edges and SHD~11 with nonlinear edges, ahead of the remaining constraint-based methods (6--7 and 12--13); PAIR-CI's advantage over test-wise deletion at $p=10$ is significant but small under paired, graph-clustered bootstrap intervals (Appendix~\ref{app:pr_frontier}). At $p \geq 20$ the ordering depends on the functional form. With linear edges, MVPC is the best method overall (SHD 28 and 49 at $p = 20$ and $30$, against PAIR-CI's 31 and 53), and PAIR-CI leads test-wise deletion from $p = 20$ (31 and 53 against 32 and 58). With nonlinear edges, PAIR-CI attains the lowest SHD at every scale, its gap over the Fisher-$Z$ pipelines widening from 1 unit at $p = 10$ to 15 at $p = 30$ (58 against 73--75); MVPC (60 at $p = 30$) is close throughout on this mildly nonlinear DGP, slightly ahead on skeleton F1 under missingness (0.64 against 0.61) and behind on SHD. PAIR-CI operates at a sparse, high-precision point (precision 1.00 at $p = 10$, 0.86--0.90 at $p \geq 20$) with recall below test-wise deletion's and MVPC's (0.47 against 0.57 and 0.54 at $p = 30$ with nonlinear edges), so that on these grids its lower SHD is not matched by a higher F1; Appendix~\ref{app:pr_frontier} sweeps each $\alpha$-tunable baseline to PAIR-CI's edge count to test whether that sparsity drives the SHD gains, and on the networks below PAIR-CI attains the highest F1 outright. Complete-case analysis and MissDAG owe their SHD to collapsing recall (F1 0.53 and 0.38 against 0.61 for PAIR-CI at $p = 30$ with nonlinear edges).

\begin{table}[t]
    \centering
    \caption{\textbf{Synthetic graph recovery at scale}. Panel A: total SHD {\scriptsize (skeleton SHD)} / skeleton F1; Panel B: skeleton precision / recall {\scriptsize (median output edges)}; medians across missingness conditions of per-condition medians ($n = 1{,}000$; 10 Erd\H{o}s--R\'enyi DAGs $\times$ 20 missingness draws per graph per condition); skeleton SHD counts missing and extra adjacencies; total SHD, computed against the true DAG, also counts mismatched edge marks, including undirected output edges that are undirected in the true CPDAG. Lowest constraint-based SHD in bold (ties within 0.5); the same metrics per missingness condition in Appendix~\ref{app:robustness}. $^\ddagger$MVPC runs in its authors' released PC implementation. $^\dagger$MissDAG outputs a DAG rather than a CPDAG, so its total SHD is not comparable (gray).}
    \label{tab:scaling}
    \resizebox{0.92\textwidth}{!}{\begin{tabular}{lcccccc}
    \toprule
                  & \multicolumn{3}{c}{Linear} & \multicolumn{3}{c}{Nonlinear} \\
    \cmidrule(lr){2-4} \cmidrule(lr){5-7}
    Method        & $p{=}10$ & $p{=}20$ & $p{=}30$ & $p{=}10$ & $p{=}20$ & $p{=}30$ \\
    \midrule
    PAIR-CI                                       & \textbf{6}\,{\scriptsize(1)} / 0.95 & 31\,{\scriptsize(14)} / 0.78 & 53\,{\scriptsize(28)} / 0.76 & \textbf{11}\,{\scriptsize(7)} / 0.65 & \textbf{35}\,{\scriptsize(23)} / 0.63 & \textbf{58}\,{\scriptsize(38)} / 0.61 \\
    Complete case                                 & 6\,{\scriptsize(2)} / 0.90 & 33\,{\scriptsize(17)} / 0.74 & 61\,{\scriptsize(36)} / 0.66 & 12\,{\scriptsize(7)} / 0.70 & 40\,{\scriptsize(27)} / 0.61 & 68\,{\scriptsize(50)} / 0.53 \\
    Test-wise                                     & 6\,{\scriptsize(2)} / 0.90 & 32\,{\scriptsize(15)} / 0.79 & 58\,{\scriptsize(29)} / 0.76 & 12\,{\scriptsize(7)} / 0.67 & 40\,{\scriptsize(25)} / 0.63 & 68\,{\scriptsize(44)} / 0.62 \\
    FZ-vote                                       & 7\,{\scriptsize(3)} / 0.87 & 33\,{\scriptsize(15)} / 0.79 & 60\,{\scriptsize(29)} / 0.77 & 13\,{\scriptsize(8)} / 0.67 & 43\,{\scriptsize(27)} / 0.63 & 75\,{\scriptsize(48)} / 0.62 \\
    FZ-Rubin (\citeauthor{witte2022multiple})     & 7\,{\scriptsize(2)} / 0.89 & 32\,{\scriptsize(15)} / 0.79 & 60\,{\scriptsize(29)} / 0.77 & 12\,{\scriptsize(8)} / 0.67 & 41\,{\scriptsize(27)} / 0.63 & 73\,{\scriptsize(47)} / 0.62 \\
    MVPC$^\ddagger$ (\citeauthor{tu2019causal})   & \textbf{6}\,{\scriptsize(1)} / 0.94 & \textbf{28}\,{\scriptsize(12)} / 0.82 & \textbf{49}\,{\scriptsize(23)} / 0.80 & \textbf{11}\,{\scriptsize(7)} / 0.70 & 36\,{\scriptsize(23)} / 0.64 & 60\,{\scriptsize(39)} / 0.64 \\
    \midrule
    \multicolumn{7}{l}{\emph{Score-based (DAG output):}} \\
    MissDAG$^\dagger$ (\citeauthor{gao2022missdag}) & \textcolor{lightgray}{5\,{\scriptsize(5)} / 0.62}$^\dagger$ & \textcolor{lightgray}{24\,{\scriptsize(23)} / 0.58}$^\dagger$ & \textcolor{lightgray}{36\,{\scriptsize(34)} / 0.64}$^\dagger$ & \textcolor{lightgray}{10\,{\scriptsize(10)} / 0.38}$^\dagger$ & \textcolor{lightgray}{32\,{\scriptsize(32)} / 0.38}$^\dagger$ & \textcolor{lightgray}{54\,{\scriptsize(53)} / 0.38}$^\dagger$ \\
    \bottomrule
\end{tabular}
}\\[3pt]
    {\small Panel B: skeleton precision / recall {\scriptsize (median output edges)}}\\[1pt]
    \resizebox{0.92\textwidth}{!}{\begin{tabular}{lcccccc}
    \toprule
                  & \multicolumn{3}{c}{Linear} & \multicolumn{3}{c}{Nonlinear} \\
    \cmidrule(lr){2-4} \cmidrule(lr){5-7}
    Method        & $p{=}10$ & $p{=}20$ & $p{=}30$ & $p{=}10$ & $p{=}20$ & $p{=}30$ \\
    \midrule
    PAIR-CI                                       & 1.00/0.95\,{\scriptsize(9)} & 0.87/0.71\,{\scriptsize(30)} & 0.86/0.68\,{\scriptsize(51)} & 1.00/0.52\,{\scriptsize(7)} & 0.90/0.48\,{\scriptsize(22)} & 0.89/0.47\,{\scriptsize(34)} \\
    Complete case                                 & 0.87/1.00\,{\scriptsize(10)} & 0.85/0.66\,{\scriptsize(29)} & 0.79/0.56\,{\scriptsize(44)} & 0.79/0.65\,{\scriptsize(11)} & 0.76/0.52\,{\scriptsize(26)} & 0.68/0.43\,{\scriptsize(38)} \\
    Test-wise                                     & 0.86/1.00\,{\scriptsize(10)} & 0.83/0.76\,{\scriptsize(34)} & 0.80/0.73\,{\scriptsize(58)} & 0.75/0.63\,{\scriptsize(10)} & 0.71/0.58\,{\scriptsize(32)} & 0.70/0.57\,{\scriptsize(52)} \\
    FZ-vote                                       & 0.82/1.00\,{\scriptsize(11)} & 0.81/0.78\,{\scriptsize(36)} & 0.79/0.75\,{\scriptsize(61)} & 0.70/0.65\,{\scriptsize(11)} & 0.67/0.60\,{\scriptsize(36)} & 0.64/0.60\,{\scriptsize(60)} \\
    FZ-Rubin (\citeauthor{witte2022multiple})     & 0.83/1.00\,{\scriptsize(11)} & 0.81/0.78\,{\scriptsize(36)} & 0.79/0.75\,{\scriptsize(61)} & 0.72/0.64\,{\scriptsize(11)} & 0.68/0.60\,{\scriptsize(35)} & 0.65/0.59\,{\scriptsize(58)} \\
    MVPC$^\ddagger$ (\citeauthor{tu2019causal})   & 1.00/0.96\,{\scriptsize(9)} & 0.94/0.72\,{\scriptsize(28)} & 0.94/0.69\,{\scriptsize(47)} & 0.82/0.60\,{\scriptsize(9)} & 0.81/0.55\,{\scriptsize(27)} & 0.80/0.54\,{\scriptsize(44)} \\
    \midrule
    \multicolumn{7}{l}{\emph{Score-based (DAG output):}} \\
    MissDAG$^\dagger$ (\citeauthor{gao2022missdag}) & \textcolor{lightgray}{1.00/0.45\,{\scriptsize(4)}}$^\dagger$ & \textcolor{lightgray}{0.95/0.42\,{\scriptsize(17)}}$^\dagger$ & \textcolor{lightgray}{0.97/0.48\,{\scriptsize(32)}}$^\dagger$ & \textcolor{lightgray}{1.00/0.24\,{\scriptsize(3)}}$^\dagger$ & \textcolor{lightgray}{0.89/0.25\,{\scriptsize(12)}}$^\dagger$ & \textcolor{lightgray}{0.83/0.26\,{\scriptsize(21)}}$^\dagger$ \\
    \bottomrule
\end{tabular}
}
\end{table}

\subsection{Scaling to Real-World Network Topologies}
\label{sec:exp_realworld}

\paragraph{ALARM medical diagnostic network.}
We next turn to the ALARM network \citep{beinlich1989alarm} ($p=37$, 46 edges), injecting 10--40\% missingness per column into 10 non-root variables (about 4--11\% of all entries) under MAR, MNAR, and mixed mechanisms (20 replicates), with linear Gaussian and with nonlinear structural equations; Table~\ref{tab:realworld} reports the 20\% MNAR cells, and MAR and mixed are similar (Appendix~\ref{app:alarm}). With linear edges, MVPC is the best method (SHD 18.5--20 across mechanisms, skeleton F1 0.93--0.95), ahead of PAIR-CI (22.5--24, F1 0.89--0.91), and both are well ahead of test-wise deletion (33--34), the Fisher-$Z$ pipelines (30.5--36), and complete-case analysis (34--40). With nonlinear edges the ordering reverses: MVPC inherits its Fisher-$Z$ core's limitations (its precision falls from 0.95 to 0.71--0.76, against PAIR-CI's 0.94), and PAIR-CI attains SHD 31.5--33.5 against 40--42 for MVPC, 48--52 for test-wise deletion, 48--51 for complete-case analysis, and 53--57.5 for the Fisher-$Z$ pipelines, with skeleton F1 0.79--0.80 against 0.49--0.69 for every constraint-based baseline and recall (0.67--0.70) at least as high as every baseline's.

\paragraph{HAILFINDER weather forecasting network.}
In the largest benchmark, the HAILFINDER network \citep{abramson1996hailfinder} ($p=56$, 66 edges), we simulate data using nonlinear structural equations and induce 20--40\% missingness per column in 15 non-root variables (about 7--11\% of all entries) (20 replicates). At 20\% missingness PAIR-CI attains the lowest SHD under every mechanism (62.5--65.5, against 72--83 for MVPC and 103.5--128 for test-wise deletion and the Fisher-$Z$ pipelines). Table~\ref{tab:realworld}'s edge counts show that the gain is not bought by pruning: the latter output 88.5--108.5 edges, mostly false (precision 0.38--0.43), whereas PAIR-CI outputs 42 with precision 0.73 and the highest skeleton F1 (0.59); complete-case analysis reaches SHD~66 only by returning an empty skeleton. At 40\% missingness MVPC and PAIR-CI converge (Appendix~\ref{app:hailfinder}).

\begin{table}[t]
    \centering
    \caption{\textbf{Real-world graph recovery at 20\% MNAR missingness} per incomplete column, about 7\% of all entries (20 replicates): median total SHD, skeleton precision (P), recall (R), F1, and median output edges (E) on ALARM ($p{=}37$, 46 edges; linear and nonlinear structural equations) and HAILFINDER ($p{=}56$, 66 edges; nonlinear); MAR and mixed in Appendices~\ref{app:alarm} and~\ref{app:hailfinder}. Best constraint-based SHD and F1 in bold. $^\ddagger$MVPC in its authors' released PC implementation (one HAILFINDER replicate exceeded the 1-hour cap; median over 19). $^\dagger$Median F1 $= 0.000$: an empty recovered skeleton.}
    \label{tab:realworld}
    \resizebox{\textwidth}{!}{\begin{tabular}{lccccccccccccccc}
    \toprule
     & \multicolumn{5}{c}{ALARM (linear)} & \multicolumn{5}{c}{ALARM (nonlinear)} & \multicolumn{5}{c}{HAILFINDER} \\
    \cmidrule(lr){2-6} \cmidrule(lr){7-11} \cmidrule(lr){12-16}
    Method & SHD & P & R & F1 & E & SHD & P & R & F1 & E & SHD & P & R & F1 & E \\
    \midrule
    PAIR-CI                                       & 22.5 & 0.91 & 0.88 & 0.89 & 44.5 & \textbf{33} & 0.94 & 0.67 & \textbf{0.80} & 33.5 & \textbf{64} & 0.73 & 0.48 & \textbf{0.59} & 42.0 \\
    Complete case                                 & 40 & 0.78 & 0.76 & 0.78 & 45.5 & 51 & 0.59 & 0.43 & 0.49 & 35 & 66 & 0.00 & 0.00 & \textcolor{lightgray}{0.00}$^\dagger$ & 0 \\
    Test-wise                                     & 34 & 0.80 & 0.91 & 0.84 & 52.5 & 48 & 0.61 & 0.63 & 0.62 & 47.5 & 103.5 & 0.43 & 0.58 & 0.49 & 88.5 \\
    FZ-vote                                       & 36 & 0.76 & 0.93 & 0.84 & 56 & 55.5 & 0.54 & 0.63 & 0.58 & 56 & 122 & 0.38 & 0.61 & 0.47 & 108.5 \\
    FZ-Rubin (\citeauthor{witte2022multiple})     & 35 & 0.77 & 0.93 & 0.84 & 55 & 53 & 0.55 & 0.62 & 0.58 & 51 & 120 & 0.39 & 0.61 & 0.48 & 107.5 \\
    MVPC$^\ddagger$ (\citeauthor{tu2019causal})   & \textbf{20} & 0.95 & 0.91 & \textbf{0.93} & 44 & 40 & 0.75 & 0.59 & 0.65 & 36.5 & 72 & 0.59 & 0.55 & 0.57 & 63 \\
    \bottomrule
\end{tabular}
}
\end{table}

\section{Discussion and Conclusion}
\label{sec:discussion}

Our experimental evidence suggests that PAIR-CI addresses a key gap in constraint-based causal discovery, a cornerstone of automated scientific inference. Whereas existing approaches are miscalibrated under common forms of misspecified imputation, PAIR-CI keeps false positive rates at or near the nominal level across missingness mechanisms, edge types, and graph sizes, exceeding it in some cells with $|\Xb| \geq 5$. The primary cost is reduced power relative to well-specified parametric tests when their assumptions hold, that is, when relationships are approximately linear and missingness is MCAR or MAR; part of this cost is consistent with the conservatism derived in Theorem~\ref{thm:bridge}, under which Rubin's combined variance overstates the variance of $\hat\mu$, so the test under-rejects at low signal strengths and small $n$ (Appendix~\ref{app:variance_ratio}). Under linear conditions a correctly specified parametric pipeline, now best represented by MVPC, outperforms PAIR-CI (Section~\ref{sec:exp_graph}; Appendix~\ref{app:sachs}), and we recommend it there; where parametric assumptions are violated or the missingness mechanism is unverifiable, PAIR-CI's calibration advantage outweighs its loss in power, and the gap widens with scale. Appendix~\ref{app:untestable} discusses the untestable assumptions on which these conclusions rest.

Two scope restrictions deserve mention. First, PAIR-CI is agnostic to the discovery algorithm and could in principle serve as the CI oracle within FCI; our recovery guarantee (Corollary~\ref{cor:pc_consistency}) and experiments cover PC under causal sufficiency, and evaluating the FCI combination is future work. Second, part (i) of Theorem~\ref{thm:bridge} holds for general cross-validated statistics over multiply imputed datasets, but its conservatism rests on the paired design and we illustrate it only for CI testing; extending it to cross-validated model selection with incomplete data is a natural direction. PAIR-CI complements score-based approaches such as MissDAG, which achieves lower SHD when the assumptions of the linear Gaussian variant we run are satisfied, by recovering more structure under nonlinearity, and, like all constraint-based methods, requires domain validation before deployment in high-stakes contexts: the recovered graph represents a set of CI relations consistent with the data, not a confirmed causal structure.

\begin{ack}
We thank Jouni Kuha for helpful discussions. Computations used the Winston HPC cluster of the Department of Methodology, London School of Economics. The authors received no third-party funding or support for this work and declare no competing interests.
\end{ack}

\bibliographystyle{plainnat}
\bibliography{references}

\newpage
\appendix

\part*{\Large Appendices}

\section{Proofs of Propositions 1--3 and Corollary 5}
\label{app:proofs}

\subsection{Notation and Setup}
Let $\hat{D}^{(m)}$ for $m = 1, \ldots, M$ denote multiply imputed versions of $D$, and let $\{F_1, \ldots, F_K\}$ be a partition of $\{1, \ldots, n\}$ into $K$ folds of approximately equal size $n_k = n/K$. For imputed dataset $m$ and fold $k$, denote the test set as $F_k$ and training set as $T_{km} = \hat{D}^{(m)} \setminus F_k$.

The fold-level loss difference is
\begin{equation}
\hat{\mu}_{km} = \frac{1}{n_k} \sum_{i \in F_k} 
\left[ 
\ell\!\left(z_i,\, \hat{g}^{ZU}_{km}(\hat{x}_i, u_{ik})\right) - 
\ell\!\left(z_i,\, \hat{g}^{ZY}_{km}(\hat{x}_i, y_i)\right) 
\right]
\label{eq:app_fold_diff}
\end{equation}
where $u_{ik}$ is the locally resampled value of $Y$ assigned to observation $i$ in fold $k$. The per-imputation and overall means are
\begin{equation}
    \hat{\mu}_m = \frac{1}{K} \sum_{k=1}^{K} \hat{\mu}_{km}, \qquad 
    \hat{\mu} = \frac{1}{M} \sum_{m=1}^{M} \hat{\mu}_m.
    \label{eq:app_means}
\end{equation}

\subsection{Proof of Proposition~\ref{prop:impossibility}
(Miscalibration of Impute-then-Test)}

\begin{proof}
\textbf{Step 1: Misspecified imputation induces spurious conditional dependence.}
Condition~(a) arises in two regimes.

\emph{(i) MNAR.} Under MNAR the response mechanism depends on the missing entries themselves, so ignorability fails: $P(\Xb_{\mathrm{miss}} \mid \Xb_{\mathrm{obs}}, R) \neq P(\Xb_{\mathrm{miss}} \mid \Xb_{\mathrm{obs}})$. An imputer for column $j$ fitted on the rows where that column is observed converges to $P(X_j \mid \Xb_{-j}, R_j = 1)$, whereas $P(X_j \mid \Xb_{-j})$ is the mixture $\pi P(X_j \mid \Xb_{-j}, R_j = 1) + (1 - \pi) P(X_j \mid \Xb_{-j}, R_j = 0)$ with $\pi = P(R_j = 1 \mid \Xb_{-j})$, so the two differ by $(1 - \pi)$ times the total-variation distance between the observed-row and missing-row conditionals. Whenever missingness depends non-trivially on $X_j$, that distance is positive at the level of the joint law of $(\Xb_{\mathrm{obs}}, \Xb_{\mathrm{miss}})$, say $\varepsilon_0$, and condition~(a) holds generically with any $\varepsilon < \varepsilon_0$ for all sufficiently large $n$. Self-masking, where $R_j$ depends on $X_j$ itself, is one canonical instance of MNAR and the instance used in our experiments (Equation~\ref{eq:mechanism}); Appendix~\ref{app:adversarial_kappa} considers MNAR geometries beyond self-masking.

\emph{(ii) MAR with a misspecified imputer.} Under MAR, the target $P(\Xb_{\mathrm{miss}} \mid \Xb_{\mathrm{obs}})$ is identifiable, but if the imputer's model class fails to contain the true conditional---for example, when Bayesian ridge imputation is used to impute data with nonlinear conditional means---the imputer converges to an in-class approximation at strictly positive TV distance from the truth, again establishing condition~(a).

\textbf{Step 2: Consistent tests detect the spurious association.} The completed data $\hat{D}_n$ converge in distribution to draws from $\hat{P}$, under which the conditional dependence $Z \nindep Y \mid \hatX$ holds by condition~(b). For \emph{single imputation}, a consistent test applied to data from $\hat{P}$ detects this dependence, and the rejection probability converges to $1$. For \emph{multiple imputation with Rubin's rules}, the within-imputation estimates $\hat{\theta}_m$ each converge to the same non-zero population quantity (since all imputed datasets converge to the same incorrect conditional $Q$), implying that the pooled estimate $\bar{\theta} \to \theta^* \neq 0$. As the between-imputation variance $B$ captures only imputation sampling variation---whose contribution to the total variance of $\bar{\mu}$ is $O(n^{-1})$ and vanishes asymptotically---the combined test statistic diverges.
\end{proof}

\begin{remark}[When the conditions of Proposition~\ref{prop:impossibility} hold]
    \label{rem:conditions}
    Condition~(a) holds both under MNAR, where standard imputation procedures are asymptotically misspecified because the ignorability assumption---$P(\Xb_{\mathrm{miss}} \mid \Xb_{\mathrm{obs}}, R) = P(\Xb_{\mathrm{miss}} \mid \Xb_{\mathrm{obs}})$---fails by definition; and under MAR when the imputation model is misspecified in functional form (for example, a Bayesian ridge imputation strategy in nonlinear conditions). Condition~(b) is often satisfied in causal discovery applications: $\Xb$ is chosen as a candidate separating set for $Z$ and $Y$, so when $\Xb_{\mathrm{miss}}$ lies on an active path between $Z$ and $Y$, imputation error fails to block this path and induces spurious dependence. Proposition~\ref{prop:impossibility} is structural, serving to motivate PAIR-CI by isolating the failure mode that its paired design is built to neutralize.
\end{remark}

\subsection{Proof of Proposition~\ref{prop:calibration} (Calibration)}
\begin{proof}
The proof proceeds in two steps: calibration for the internal null
$Z \indep Y \mid \hatX$ under Assumptions~(A1$'$) and (A2)--(A4); and extension to the scientific null $Z \indep Y \mid \Xb$ via imputation consistency~(A1).

\paragraph{Exchangeability under the internal null.}
Consider imputed dataset $m$ and fold $k$, and treat the sample $\hatX^{(m)}$ as the conditioning set. Local resampling of $Y$ from its $k_{\textup{nn}}$ nearest neighbors in $\hatX$-space (Appendix~\ref{app:knn}) produces a placebo $U$ whose conditional law given $\hatX$ approximates that of $Y$ given $\hatX$, with a total-variation discrepancy $\varepsilon_n$ between the joint law of the placebo sample and the exactly exchangeable one; for the binned version of the test, the analogous distance between the law of the binned sample and its conditional-independence projection bounds the excess type-I error \citep[Lemma~1]{kim2022local}; we assume that $\varepsilon_n$ vanishes and, for the size bound, condition (N2) of Appendix~\ref{app:bridge}; the argument below is stated for the exact coupling and holds up to $\varepsilon_n$ otherwise. Under the internal null $Z \indep Y \mid \hatX$, the joint distributions of $(\hatX, Z, Y)$ and $(\hatX, Z, U)$ coincide: $Y$ and $U$ are conditionally exchangeable given $\hatX$, and neither carries information about $Z$ beyond $\hatX$. The pairs $(\hatX, Y) \mapsto Z$ and $(\hatX, U) \mapsto Z$ are thus statistically indistinguishable, which entails that the distributions of $\hat{g}^{ZY}_{km}$ and $\hat{g}^{ZU}_{km}$ and their fold-level losses are identical. It follows that
\[
    \E[\hat{\mu}_{km}] = \E\bigl[\ell(Z, \hat{g}^{ZU}_{km}) 
    - \ell(Z, \hat{g}^{ZY}_{km})\bigr] = 0,
\]
which follows purely from exchangeability---no universal consistency or Bayes-optimality assumption for the learner is required. Since averaging over folds and imputations yields $\E[\hat{\mu}] = 0$, the population loss difference $\mu$ targeted by Theorem~\ref{thm:bridge} is also 0 under the internal null.

\paragraph{Asymptotic calibration.}
By Theorem~\ref{thm:bridge}, under Assumptions~(A1$'$) and (A2)--(A4) and the conditions stated there, and provided the local permutation is accurate enough that $\sqrt{n}\,\E[\hat\mu] \to 0$ under the internal null (condition (N2)), $nT$ converges in probability, in the limit of infinitely many imputations, to a limit $\Sigma_{\mathrm{Rubin}}$ at least as large as the limiting variance $\Sigma_\infty$ of $\sqrt{n}\hat\mu$. By Slutsky's lemma $\tCI = \hat{\mu}/\sqrt{T} \xrightarrow{d} \mathcal{N}(0, \Sigma_\infty/\Sigma_{\mathrm{Rubin}})$, whose upper-tail probabilities are no larger than those of a standard normal variable at every non-negative threshold, and since $t_{\alpha,\nu} \geq z_\alpha$ for every $\nu > 0$ and $\alpha < 1/2$, $\lim_M \limsup_n \Pr(\tCI > t_{\alpha,\nu}) \leq \alpha$, with equality when $\sigma^2_{\mathrm{draw}}$ is negligible relative to $\sigma^2_{\mathrm{obs}}$. For fixed $M$ the limit law of $\tCI$ is Rubin's $t$-type pivot and the bound rests on the Barnard--Rubin \citep{barnard1999small} reference distribution dominating it, which we examine empirically in Appendix~\ref{app:variance_ratio}; there, replacing the $t_\nu$ critical value by the standard normal changes rejection rates by at most 0.012, so the conservatism enters through $T$ rather than through the critical value.

\paragraph{From internal null to scientific null.}
The two nulls coincide when $\hatX$ recovers $\Xb$---the content of Assumption~(A1)---and (approximately) when imputation quality is adequate. When imputed values are systematically biased, Proposition~\ref{prop:impossibility} shows that the nulls diverge and calibration fails. PAIR-CI's paired design attenuates but cannot eliminate this effect (Appendix~\ref{app:imputation_degradation}).
\end{proof}

\subsection{Proof of Proposition~\ref{prop:consistency} (Consistency)}
\begin{proof}
Let $R^* = R^*(Z \mid \Xb, Y)$ and $R^*_0 = R^*(Z \mid \Xb)$ denote the Bayes risks of the full and partial prediction tasks, respectively, and $\Delta^* = R^*_0 - R^*$. We consider alternatives $H_1: Z \nindep Y \mid \Xb$ with $\Delta^* > 0$, so that $R^* < R^*_0$. For squared error $\Delta^* = \E[(\E[Z \mid \Xb, Y] - \E[Z \mid \Xb])^2]$, positive exactly when $Y$ shifts the conditional mean of $Z$, so a dependence confined to higher moments need not be detected; for cross-entropy with binary $Z$, $\Delta^*$ is the conditional mutual information of $Z$ and $Y$ given $\Xb$, positive under every conditional dependence.

By Assumption~(A5), the expected losses of the trained full and partial models converge to $R^*$ and $R^*_0$, respectively. Under Assumption~(A1), the completed-data law converges to the complete-data law, so
\[
    \E[\hat{\mu}_m] \to R^*_0 - R^* \equiv \Delta^* > 0 \quad \text{for all } m.
\]
Since the variance of $\hat{\mu}_m$ is $O(1/n)$ by Lemma~\ref{lem:within_ral}, and $T = O(1/n)$, the test statistic satisfies
\[
    t_{\mathrm{CI}} = \frac{\hat{\mu}}{\sqrt{T}} =
    \frac{\Delta^* + o_P(1)}{O(n^{-1/2})} \to \infty
    \quad \text{as } n \to \infty.
\]
Thus, $\Pr(t_{\mathrm{CI}} > t_{\alpha,\nu}) \to 1$.
\end{proof}

\subsection{Proof of Corollary~\ref{cor:pc_consistency} (PC Consistency)}
\begin{proof}
\citet[Theorem~2]{kalisch2007estimating} prove that the PC algorithm recovers the true CPDAG with probability tending to one along a sequence of levels $\alpha_n \to 0$ chosen so that the per-test type-I and type-II error probabilities of Fisher's $z$ test vanish uniformly over the tests the algorithm performs, under faithfulness and their conditions on the Gaussian model; the argument uses the test only through those per-test bounds and a union bound over the finitely many tests at fixed $p$, so it transfers to any oracle supplying such bounds. Proposition~\ref{prop:calibration} bounds the PAIR-CI oracle's type-I error at each fixed level and Proposition~\ref{prop:consistency} sends its type-II error to zero against alternatives with $\Delta^* > 0$, but neither is stated uniformly in $\alpha$, so the existence of the sequence $\alpha_n$ is the oracle property on which the corollary is conditioned. A test with size at most a fixed $\alpha$ cannot deliver the conclusion on its own, since a size bound does not force the probability of a type-I error among the tests to vanish, and, when a true null is among the tests, for a test of exact size $\alpha$ that probability is bounded away from zero; this corrects the submitted version, which invoked correct asymptotic size at a fixed level. The weakening of Theorem~\ref{thm:bridge} from exact to conservative size does not propagate, because only an upper bound on the type-I error enters, and Proposition~\ref{prop:consistency}'s divergence argument tolerates a constant-factor inflation of the variance ($T = O_P(n^{-1})$ under the alternative by Lemma~\ref{lem:within_ral} and (R1)--(R2)).
\end{proof}

\subsection{Skeleton False-Edge Control at a Fixed Level}
\label{app:skeleton_fixed_alpha}

Corollary~\ref{cor:pc_consistency} needs a sequence of levels $\alpha_n \to 0$. At a fixed level the size bound does not make spurious adjacencies vanish, so it cannot by itself deliver consistent recovery; the following corollary states what a fixed level does deliver for the skeleton of the standard PC adjacency search.

Corollaries~\ref{cor:pc_consistency} and~\ref{cor:skeleton_fixed_alpha} concern the standard PC adjacency search. The simplified search in our experiments (Algorithm~\ref{alg:pc} of Appendix~\ref{app:implementation}) draws conditioning sets from one endpoint's neighborhood and stops at the first level at which no edge is removed, so they do not cover it; PAIR-CI and the in-house baselines share this search, while MVPC runs the PC-stable search of \texttt{causal-learn} followed by its correction step.

\begin{corollary}[Skeleton false-edge control at a fixed level]\label{cor:skeleton_fixed_alpha}
    Fix $\alpha \in (0, 1/2)$. Let $G$ be the true DAG on the $p$ observed variables, Markov and, by (A6), faithful to the complete-data law, with $N_0$ non-adjacent pairs, and let $\hat G^{\mathrm{skel}}_{n,M}$ be the skeleton returned by the sample version of the PC adjacency search of \citet[Algorithm~1 and Section~2.2.2]{kalisch2007estimating}, in which the test for the ordered pair $(i, j)$ and set $\mathbf S$ is PAIR-CI at level $\alpha$ with response $X_i$, candidate $X_j$ and conditioning set $\Xb_{\mathbf S}$, the edge being deleted when PAIR-CI does not reject. Suppose that for every ordered triple $(i, j, \mathbf S)$ with $\mathbf S \subseteq V \setminus \{i, j\}$: \textup{(a)}~if $i$ and $j$ are adjacent in $G$, then $P(\tCI > t_{\alpha,\nu}) \to 1$ as $n \to \infty$ at every fixed $M$; \textup{(b)}~if $\mathbf S$ $d$-separates $i$ and $j$ in $G$, then $\limsup_M \limsup_n P(\tCI > t_{\alpha,\nu}) \leq \alpha$. Then $P(\hat G^{\mathrm{skel}}_{n,M} \supseteq \mathrm{skel}(G)) \to 1$ at every fixed $M$ and, for every pair $\{i, j\}$ non-adjacent in $G$,
    \[
        \limsup_{M} \limsup_{n} P\bigl(\{i, j\} \in \hat G^{\mathrm{skel}}_{n,M}\bigr) \leq \alpha,
        \quad\text{hence}\quad
        \limsup_{M} \limsup_{n} \E\bigl|\hat G^{\mathrm{skel}}_{n,M} \setminus \mathrm{skel}(G)\bigr| \leq \alpha N_0 .
    \]
\end{corollary}

We write $\limsup_M$ because the limit in $M$ of $\limsup_n P(\cdot)$ need not exist; the bound of Proposition~\ref{prop:calibration}, stated with $\lim_M$, implies~(b) where it applies. Neither condition follows from results proved here in the generality stated; the corollary isolates what the adjacency search adds once both are granted. For a triple whose endpoints are adjacent, Proposition~\ref{prop:consistency} gives~(a) at each fixed $M$ under (A1), (A2), (A4) and (A5) when the dependence is visible to the loss, which for squared error means $\E[X_i \mid \Xb_{\mathbf S}, X_j] \neq \E[X_i \mid \Xb_{\mathbf S}]$ on a set of positive probability: (A6) makes every such triple a conditional-dependence alternative, not necessarily a mean-dependence one. Proposition~\ref{prop:calibration} proves~(b) for the internal null; for the scientific null $X_i \indep X_j \mid \Xb_{\mathbf S}$ it gives the bound, under (A1) and the cached design, only up to a caching residual that we measure rather than bound, and with missingness confined to the conditioning set, whereas inside PC the test variables may themselves be incomplete (Appendix~\ref{app:testvar}). Condition~(b) is needed for one triple per non-adjacent pair only, the triple $(a, b, \mathrm{Pa}(a))$ of Step~1, and so only for $|\mathbf S|$ up to the maximum in-degree of $G$ (Step~3 below).

\paragraph{Proof.}
Throughout, $V = \{1, \ldots, p\}$; $\mathrm{adj}(G, k)$, $\mathrm{Pa}(k)$ and $\mathrm{De}(k)$ are the neighbors, parents and descendants of $k$ in the true DAG $G$; $\mathrm{skel}(G)$ is its set of unordered adjacent pairs; and $\mathrm{adj}(C, k)$ is the neighborhood of $k$ in the current graph $C$ of the search. We assume the causal sufficiency implicit in ``the true DAG on the observed variables'', the Markov property of the complete-data law with respect to $G$, and faithfulness (A6).

\paragraph{The oracle as a decision table.}
Let $\mathcal T = \{(i, j, \mathbf S) : i \neq j,\ \mathbf S \subseteq V \setminus \{i, j\}\}$, a finite set of at most $p(p-1)2^{p-2}$ ordered triples. For each $n$ and $M$ and each $(i, j, \mathbf S) \in \mathcal T$, let $D_{n,M}(i, j, \mathbf S) \in \{0, 1\}$ indicate that PAIR-CI at level $\alpha$, with response $X_i$, candidate $X_j$, conditioning set $\Xb_{\mathbf S}$, $n$ observations and $M$ imputations, rejects. We define all entries on one probability space together with the data, regarding the pseudo-random stream of the implementation as auxiliary randomness attached to each query. Our implementation draws that randomness from one sequential generator, so the stream a query sees depends on the queries before it; idealising the generator as independent uniform draws, the decision returned at a query has, given the data and the history, the law of the corresponding entry, and since the search queries each ordered triple at most once, its output has the same law as when it reads the table. The search consults the entry of a triple when it queries that triple, so its output is a deterministic function of the table, and the argument below uses only the marginal law of each entry, never the joint law of several. This is the device of \citet[proof of Theorem~1]{kalisch2007estimating}, for whom the sample algorithm coincides with the population algorithm on the event that every entry agrees with $d$-separation. We replace that event by the weaker event $E_{n,M}$ of Step~2, on which only the entries of adjacent pairs are required to be correct.

\paragraph{Step 1: a separating parent set.}
Let $\{i, j\} \notin \mathrm{skel}(G)$. If $j \notin \mathrm{De}(i)$, then $\mathrm{Pa}(i)$ $d$-separates $i$ and $j$. To see this, consider any path from $i$ to $j$ and let $v$ be the node after $i$. If $v \in \mathrm{Pa}(i)$, then $v$ is a non-collider on the path, because its edge to $i$ points away from $v$, and $v$ lies in the conditioning set, so the path is blocked at $v$. Otherwise the path leaves $i$ along a directed edge, and since it cannot reach $j \notin \mathrm{De}(i)$ as a directed path it contains a first collider $c$, with $c \in \mathrm{De}(i)$. By acyclicity neither $c$ nor any descendant of $c$ lies in $\mathrm{Pa}(i)$, so the path is blocked at $c$. If instead $j \in \mathrm{De}(i)$, then $i \notin \mathrm{De}(j)$ by acyclicity, and the same argument shows that $\mathrm{Pa}(j)$ $d$-separates them. Write $(a, b, \mathbf S^*)$ for $(i, j, \mathrm{Pa}(i))$ in the first case and for $(j, i, \mathrm{Pa}(j))$ in the second. Then $\mathbf S^* = \mathrm{Pa}(a) \subseteq \mathrm{adj}(G, a) \setminus \{b\}$, $|\mathbf S^*| \leq q_{\mathrm{in}}$, the maximum in-degree of $G$, and by the Markov property $X_a \indep X_b \mid \Xb_{\mathbf S^*}$. The triple $(a, b, \mathbf S^*)$ is therefore a scientific null and lies in the scope of condition~(b).

\paragraph{Step 2: no true edge is removed.}
Let
\[
    E_{n,M} = \bigcap_{\{i, j\} \in \mathrm{skel}(G)}\ \bigcap_{\mathbf S \subseteq V \setminus \{i, j\}} \bigl\{ D_{n,M}(i, j, \mathbf S) = 1 \bigr\} \cap \bigl\{ D_{n,M}(j, i, \mathbf S) = 1 \bigr\}.
\]
Every triple in this intersection has adjacent endpoints, so condition~(a) applies to it, and each of the finitely many events has probability tending to one at fixed $M$. The union bound gives $P(E_{n,M}^c) \to 0$ as $n \to \infty$ for every fixed $M$. An edge is deleted only when a queried entry equals $0$, so on $E_{n,M}$ no edge of $\mathrm{skel}(G)$ is ever deleted. This gives the first assertion, $\hat G^{\mathrm{skel}}_{n,M} \supseteq \mathrm{skel}(G)$, and it gives $\mathrm{adj}(C, k) \supseteq \mathrm{adj}(G, k) \supseteq \mathrm{Pa}(k)$ for every $k$ at every stage of the search. This inclusion does not depend on the order in which pairs and subsets are visited or on the in-place updating of $C$, which only ever removes edges outside $\mathrm{skel}(G)$.

\paragraph{Step 3: on $E_{n,M}$, a surviving false edge has been tested at $\mathbf S^*$.}
Fix $\{i, j\} \notin \mathrm{skel}(G)$, with $(a, b, \mathbf S^*)$ as in Step~1 and $d = |\mathbf S^*|$. We claim that on $E_{n,M}$,
\[
    \{a, b\} \in \hat G^{\mathrm{skel}}_{n,M} \quad\Longrightarrow\quad D_{n,M}(a, b, \mathbf S^*) = 1 .
\]
Suppose the edge $a$--$b$ is in the output. It was then never deleted, so it is present at every level $\ell$ of the search. On $E_{n,M}$ we have $\mathrm{adj}(C, a) \setminus \{b\} \supseteq \mathrm{Pa}(a) = \mathbf S^*$ throughout, because $b \notin \mathrm{Pa}(a)$, and hence $|\mathrm{adj}(C, a) \setminus \{b\}| \geq d$ at every stage. The termination test of \citet[Algorithm~1, line~18]{kalisch2007estimating} stops after level $\ell$ only if every ordered adjacent pair $(i', j')$ satisfies $|\mathrm{adj}(C, i') \setminus \{j'\}| < \ell$. This condition fails for $(a, b)$ after every level $\ell < d$, so the search reaches level $d$. At level $d$ the ordered pair $(a, b)$ is adjacent with $|\mathrm{adj}(C, a) \setminus \{b\}| \geq d$, so it is selected (line~8), and the inner loop (lines~9--16) queries the entries $D_{n,M}(a, b, \mathbf S)$ for the $d$-subsets $\mathbf S$ of $\mathrm{adj}(C, a) \setminus \{b\}$, in some order, stopping only when an entry equals $0$, in which case the edge is deleted, or when all $d$-subsets have been queried. Since $\mathbf S^*$ is one of these subsets and the edge was not deleted, every queried entry equals $1$, and in particular $D_{n,M}(a, b, \mathbf S^*) = 1$. Because the pair is selected with $a$ as its first coordinate, we need only the entry with response $X_a$ and candidate $X_b$, and only for $|\mathbf S^*| \leq q_{\mathrm{in}}$. If the search were capped at conditioning sets of size $d_{\max}$, the argument would go through provided $q_{\mathrm{in}} \leq d_{\max}$.

\paragraph{Step 4: conclusion.}
By Steps~2 and~3, for every $\{i, j\} \notin \mathrm{skel}(G)$ and every $n$ and $M$,
\[
    P\bigl(\{i, j\} \in \hat G^{\mathrm{skel}}_{n,M}\bigr) \leq P(E_{n,M}^c) + P\bigl(D_{n,M}(a, b, \mathbf S^*) = 1\bigr).
\]
Taking $\limsup_n$ at fixed $M$ removes the first term. The second term is the type-I error of PAIR-CI for the scientific null $X_a \indep X_b \mid \Xb_{\mathbf S^*}$, so taking the limit superior in $M$ gives $\limsup_M \limsup_n P(\{i, j\} \in \hat G^{\mathrm{skel}}_{n,M}) \leq \alpha$ by condition~(b). The order of limits is that of Theorem~\ref{thm:bridge}(ii), and we do not interchange it. Summing over the $N_0$ non-adjacent pairs, $\E|\hat G^{\mathrm{skel}}_{n,M} \setminus \mathrm{skel}(G)| = \sum_{\{i, j\} \notin \mathrm{skel}(G)} P(\{i, j\} \in \hat G^{\mathrm{skel}}_{n,M})$ has $\limsup_M \limsup_n$ at most $\alpha N_0$. The same two bounds give $\limsup_M \limsup_n P(\hat G^{\mathrm{skel}}_{n,M} \neq \mathrm{skel}(G)) \leq \alpha N_0$, which is informative only for $\alpha < 1/N_0$; showing that a fixed level actually fails to deliver consistency would need a lower bound on the probability that every test of a non-adjacent pair rejects, which we do not prove. Only an upper bound on the type-I error enters, so the conservatism of Theorem~\ref{thm:bridge} does not propagate. \qed

\begin{remark}[Scope of the argument]
    \label{rem:skeleton_scope}
    (i)~Our proof uses two properties of the adjacency search. On $E_{n,M}$ the candidate conditioning sets for the ordered pair $(a, b)$ include $\mathrm{Pa}(a)$, and the search does not stop before level $|\mathrm{Pa}(a)|$ while $a$--$b$ is present. Both hold for the order-independent PC-stable search of \citet[Algorithm~4.1]{colombo2014order}: its adjacency sets are frozen at the start of each level but, on $E_{n,M}$, still contain the parents, and its termination test, which stops once every pair $(i', j')$ adjacent in $C$ satisfies $|\mathrm{adj}(C_\ell, i') \setminus \{j'\}| \leq \ell$, where $C_\ell$ is the graph frozen at the start of level~$\ell$, cannot fire after a level $\ell < |\mathrm{Pa}(a)|$; the corollary applies to it verbatim. The first property fails for a search that draws conditioning sets from the neighborhood of one fixed endpoint of each pair, the second for a search that stops at the first level at which no edge is removed. With a perfect $d$-separation oracle, the DAG $u \to w \leftarrow v$, $w \to t$, $v \to t$ has $u$ and $t$ separated only by $\{v, w\}$, which the neighborhood of $u$ does not contain once $u$--$v$ is removed at level~$0$, so a search conditioning on subsets of $u$'s neighborhood keeps $u$--$t$; and the DAG $c \to a \leftarrow d$, $c \to b \leftarrow d$, $c \to d$ has $a$ and $b$ separated only by $\{c, d\}$, nothing is removed at level~$0$, and a search with the second rule returns the complete graph.
    (ii)~The argument is the asymptotic form of the edge-specific $p$-value bound of \citet[Section~5.2]{strobl2019estimating}, who bound the $p$-value of an adjacency by the maximum of the $p$-values over the conditioning sets the search examines, under a zero type-II-error assumption that $E_{n,M}$ replaces asymptotically.
    (iii)~Condition~(b) is stated for the scientific null. Proposition~\ref{prop:calibration} proves the size bound for the internal null $X_a \indep X_b \mid (\hatX_{\mathbf S^*}, R)$ under the per-query design, which the scientific null does not imply when $\Xb_{\mathbf S^*}$ is incomplete, and states it for the scientific null only under (A1), for the cached design, and up to a residual that we measure rather than bound; (b) is therefore a hypothesis of the corollary, not a consequence of Proposition~\ref{prop:calibration}, and the same holds for~(a) and Proposition~\ref{prop:consistency} when the dependence between adjacent variables is not visible to the loss. Conditions~(a) and~(b) are hypotheses on the finitely many test decisions that the search can consult, and we do not use them uniformly in $\alpha$, which is what distinguishes the corollary from Corollary~\ref{cor:pc_consistency}.
\end{remark}

\begin{table}[h]
    \centering
    \caption{\textbf{Standalone calibration by DGP and missingness mechanism.} False positive rates at signal $= 0$ averaged over sample sizes $n \in \{500, 1{,}000, 2{,}000, 5{,}000\}$ and $|\Xb| \in \{2, 5, 10\}$ (100 replicates per cell). Complete-data results (the control) are reported in Table~\ref{tab:calibration}. Values exceeding the nominal $\alpha = 0.05$ are shown in red.}    \label{tab:calibration_by_dgp}
    \begin{tabular}{l*{6}{c}}
    \toprule
                  & \multicolumn{2}{c}{Linear Gaussian} & \multicolumn{2}{c}{Post-nonlinear} & \multicolumn{2}{c}{Latent confounder} \\
    \cmidrule(lr){2-3} \cmidrule(lr){4-5} \cmidrule(lr){6-7}
    Method        & MAR & MNAR & MAR & MNAR & MAR & MNAR \\
    \midrule
    PAIR-CI        & 0.041 & \textcolor{red}{0.070} & 0.021 & 0.038 & 0.001 & 0.000 \\
    FZ-single      & \textcolor{red}{0.105} & \textcolor{red}{0.488} & \textcolor{red}{0.121} & \textcolor{red}{0.465} & \textcolor{red}{0.083} & \textcolor{red}{0.097} \\
    FZ-Rubin       & \textcolor{red}{0.055} & \textcolor{red}{0.408} & \textcolor{red}{0.060} & \textcolor{red}{0.388} & \textcolor{red}{0.060} & 0.049 \\
    GCM (Imputed)  & \textcolor{red}{0.364} & \textcolor{red}{0.587} & \textcolor{red}{0.321} & \textcolor{red}{0.565} & \textcolor{red}{0.077} & \textcolor{red}{0.094} \\
    KCI (Imputed)  & \textcolor{red}{0.260} & \textcolor{red}{0.513} & \textcolor{red}{0.413} & \textcolor{red}{0.507} & \textcolor{red}{0.313} & \textcolor{red}{0.320} \\
    \bottomrule
\end{tabular}
    
\end{table}

\section{Bridge Theorem: Cross-Validated Inference on Multiply Imputed Data}
\label{app:bridge}

The variance estimator combines the provably consistent within-imputation estimator of \citet[Theorem~4]{bayle2020cross} with Rubin's rules for pooling across imputed datasets. We link the cross-validation Central Limit Theorem (CLT) of \citet{bayle2020cross} to the large-sample multiple imputation theory of \citet{robins2000inference}, which characterizes inference under Rubin's rules when the within-imputation estimator is regular and asymptotically linear (RAL).

\subsection{Within-Imputation Regularity}

\begin{lemma}[Within-imputation RAL property]
\label{lem:within_ral}
Let $\hat{D}^{(m)}$ denote an imputed dataset, treated as a complete i.i.d.\ sample. Under Assumptions~(A2) and~(A4), the cross-validated loss difference satisfies:
\begin{enumerate}
    \item \textbf{Asymptotic linearity.} There exists an influence function $\bar{h}_m$ 
    with $\E[\bar{h}_m(W)] = 0$ and $\sigma^2_m = \mathrm{Var}(\bar{h}_m(W)) > 0$ 
    such that
    \begin{equation}
        \hat{\mu}_m - \mu_m = \frac{1}{n} \sum_{i=1}^n
        \bar{h}_m(W_i^{(m)}) + o_P(n^{-1/2}).
        \label{eq:within_al}
    \end{equation}
    \item \textbf{CLT.} $\sqrt{n}(\hat{\mu}_m - \mu_m) / \sigma_m \xrightarrow{d} \mathcal{N}(0, 1)$ \citep[Theorems~1 and~3]{bayle2020cross}.
    \item \textbf{Consistent variance.} The within-fold variance estimator 
    $\hat{\sigma}^2_{m,\mathrm{in}}$ of \citet[Theorem~4]{bayle2020cross} satisfies 
    $\hat{\sigma}^2_{m,\mathrm{in}} / \sigma^2_m \xrightarrow{P} 1$.
\end{enumerate}
\end{lemma}

\subsection{Cross-Imputation Regularity and the Bridge}

\begin{assumption}[Cross-imputation regularity]
\label{ass:cross_imp}
The influence functions $\bar{h}_m$ from Equation~\ref{eq:within_al} satisfy:
\begin{enumerate}
    \item[(R1)] \textbf{Uniform convergence.} There exists a limiting influence function $\bar{h}^*$ such that
    $\sup_w |\bar{h}_m(w) - \bar{h}^*(w)| \xrightarrow{P} 0$ as $n \to \infty$, for all $m$.
    \item[(R2)] \textbf{Convergence to a common variance limit.} There exists
    $\sigma^2_* > 0$ such that $\sigma^2_m \xrightarrow{P} \sigma^2_*$ for all $m$, with $\mathrm{Var}_m(\sigma^2_m) = O(n^{-1})$.
\end{enumerate}
We take (R1)--(R2) as the content of common-limit imputation~(A1$'$): the $M$ chained-equation imputers share the model and the data and differ only in their posterior draws, so their influence functions and variances share a limit. The $O(n^{-1})$ rate in (R2) reflects parameter uncertainty in the imputation model \citep[cf.][]{robins2000inference} and can be assessed empirically by verifying that $\mathrm{Var}_m(\hat{\sigma}^2_{m,\mathrm{in}})$ is small across imputations.
\end{assumption}

We write $\psi$ for the parameter of the imputation model (the chained-equation coefficients and posterior scale of the Bayesian-ridge imputer), $\psi^*$ for its limit, $\hatX(\psi)$ for the completed conditioning set it produces, and $S_{\mathrm{mis}}$ for the imputation score, the derivative with respect to $\psi$ of the assumed log conditional density of the imputed values given the observed data. \citet{robins2000inference} treat imputation estimators that solve stacked estimating equations over the $M$ completed datasets; their Theorem~1 gives asymptotic normality with a sandwich variance, and their Theorem~2 gives the large-sample bias of Rubin's variance estimator, as a probability limit and in the case of infinitely many imputations, both under their regularity conditions S1--S4 (smoothness and invertibility of the expected estimating function in the imputation and analysis parameters, $L^2$ bounds on the imputation score and on the derivative of the estimating function, stochastic equicontinuity in the imputation parameter, and moments of order slightly above two) together with their in-text conditions on the imputation-parameter estimate. Their Section~5 extends Theorem~1 and Corollary~1 to conditional imputation models fitted by arbitrary estimating equations, which we take to cover chained-equations imputers such as MICE, with the caveat that MICE's conditionals need not be compatible with any joint law; we use the same extension for Theorem~2 and state it as an extension rather than a citation. In Theorem~2's bias formula every sign-ambiguous term is multiplied by a sensitivity coefficient, which \citeauthor{robins2000inference} denote $\kappa$ and identify, under their assumption that the analysis estimating function does not itself depend on $\psi$, with the derivative of the expected estimating function with respect to $\psi$. Because our estimating function depends on $\psi$ through the trained learners, the coefficient that their expansion consumes is, by the product rule, the full derivative
\[
    \dot\lambda \;=\; \kappa_{\mathrm{score}} + \dot\lambda_{\mathrm{train}}, \qquad \kappa_{\mathrm{score}} = \E\bigl[\bar{h}^*(W)\, S_{\mathrm{mis}}^{\top}\bigr], \qquad \dot\lambda_{\mathrm{train}} = \E\bigl[\partial_\psi \bar h^*_\psi(W)\bigr]\big|_{\psi^*},
\]
and we state the reading of their Theorems~1 and~2 with $\dot\lambda$ in place of $\kappa$ as an extension of their argument rather than as a citation. To avoid a clash with Definition~\ref{def:kappa_factors}, whose $\kappa$ is the imputation bias of the loss difference, we write $\kappa_{\mathrm{score}}$ for the score channel throughout, and $\mathcal{U}$ for \citeauthor{robins2000inference}'s estimating function, here $\mathcal{U} = \bar h^*(W) - \beta$. When $\dot\lambda = 0$ the bias reduces to a non-negative term, so the paired design's contribution is to pin the sign of the bias.

\begin{assumption}[Regularity for the influence-function representation]
\label{ass:regularity}
(a)~\emph{Imputation-map regularity}: $\psi \mapsto \hatX(\psi)$ is continuously differentiable with an $L^2$ derivative; satisfied by the MICE/Bayesian-ridge imputer we deploy. (b)~\emph{Lipschitz loss}: $\ell$ is bounded and Lipschitz in its prediction argument, a sharpening of (A2) satisfied by cross-entropy with clipped probabilities and by squared error with $Z$ truncated at a fixed bound; the deployed regression loss is untruncated, and the theory is stated for the truncated variant, a truncation that changes no decision in our experiments. (c)~\emph{Population-tree non-degeneracy} and (d)~\emph{split margin}: for random forests, the population split criteria have unique optimizers with a margin, and (d$'$) this uniqueness persists under the alternative. These are genuine assumptions. \citet{scornet2015consistency} state their split-convergence result (their Lemma~3) in terms of the distance from the empirical cut to the set of optimal population cuts, which does not require uniqueness; their proof of that lemma nevertheless assumes a unique best cut at each level and defers the general case, so a margin condition of the kind we impose is implicit there too. What their formulation does not give is stability in $\psi$, which is what we need. (e)~\emph{Placebo-coupling regularity}: the conditional law of the placebo given $\hatX(\psi)$ is continuously differentiable in $\psi$ in total variation; verifiable for Gaussian designs, where that conditional law is Gaussian with parameters smooth in $\psi$.
\end{assumption}

\begin{remark}[What is verified and what is assumed]
\label{rem:verified}
Under Assumption~\ref{ass:regularity}(a)--(e), Conditions S1--S4 of \citet{robins2000inference} hold for the influence-function representation of Lemma~\ref{lem:within_ral}: S2 and S4 follow from the Gaussian-linear imputation family and the bounded loss; S3 follows from a bracketing argument applied to the $\psi$-indexed class of loss differences, whose $L^2$ modulus is controlled through split-boundary-crossing probabilities for the evaluation channel and through population-split stability under (d) for the training channel; S1 under the alternative uses (d$'$). For the linear-learner variant of Appendix~\ref{app:learner_robustness}, S1--S4 hold with no split-regularity assumption, and for purely random (Mondrian) forests the margin condition (d) is unnecessary. The stability condition (A4), in the form of \citet[Theorems~3 and~4]{bayle2020cross}, is assumed for the paired loss difference and is not derived from Assumption~\ref{ass:regularity}. Assumption~\ref{ass:cross_imp}(R2) requires $\sigma^2_* > 0$, that is, that the population full and partial learners, which coincide as functions under the internal null, retain a non-degenerate dependence on their second argument; for a learner whose population limit discards an uninformative predictor, $\bar h^* \equiv 0$ and the influence-function representation degenerates, the non-degenerate regime being the one in which the trained models differ in the limit. We assume (R2) rather than derive it, and note that a fully rigorous treatment of the degenerate case would require a triangular-array version of \citeauthor{robins2000inference}'s results with $\sigma^2_m \to 0$.
\end{remark}

\begin{lemma}[Sensitivity of the paired loss difference]
\label{lem:sensitivity}
Assume (A1$'$), (A2), Assumption~\ref{ass:regularity}(a)--(b), $S_{\mathrm{mis}} \in L^2$, the per-query imputation design, and that the placebo has the exact conditional-law coupling $U \sim P_{Y \mid \hatX, R}$ with $U \indep (Y, Z) \mid (\hatX, R)$ in both the training and the evaluation folds, the internal null being read as $Y \indep Z \mid (\hatX, R)$. Then
\[
    \kappa_{\mathrm{score}} = \E\bigl[\bar{h}^*(W)\, S_{\mathrm{mis}}^{\top}\bigr] = 0 .
\]
With the local permutation of Section~\ref{sec:ci_test} in place of the exact coupling, $|\kappa_{\mathrm{score}}| \leq 4 L_\ell\, \|S_{\mathrm{mis}}\|_{L^2}\, \varepsilon_n$, where $\varepsilon_n$ is the total-variation distance between the joint law of the locally resampled placebo sample and the exactly exchangeable one (for the binned version of the test, \citealp[Lemma~1]{kim2022local} bound the excess type-I error by the analogous distance between the law of the binned sample and its conditional-independence projection; we assume the bound carries over).
\end{lemma}

\begin{proof}
The steps below are a proof for the exact coupling; only the passage to the local permutation is a bound. Under the per-query strategy the imputer never sees $Y$ or $Z$ as predictors, so $S_{\mathrm{mis}}$ is measurable with respect to $\mathfrak{F} = \sigma(\Xb_{\mathrm{obs}}, \hatX, R)$. Under the internal null and the coupling, $(Y, U)$ are independent and identically distributed given $(\hatX, R, Z)$, so the swap $Y \leftrightarrow U$ leaves the joint conditional law of the completed data given $\mathfrak{F} \vee \sigma(Z)$ invariant; the fold assignment and the learner's internal randomness are drawn independently of $(Y, U)$ and held fixed under the swap. Because the two population learners coincide at $\psi^*$ (the training laws of $(\hatX, Y, Z)$ and $(\hatX, U, Z)$ are equal under the internal null and the coupling, for any learner), $\bar h^*$ is a single function of $(\hatX, Y, U, Z)$ that is antisymmetric under the swap; hence $\E[\bar h^*(W) \mid \mathfrak{F} \vee \sigma(Z)] = 0$ and, by a further conditioning, $\E[\bar h^*(W) \mid \mathfrak{F}] = 0$, the same symmetry as in Appendix~\ref{app:proofs} applied conditionally on $\mathfrak{F}$. By the tower property and $S_{\mathrm{mis}} \in L^2$, $\kappa_{\mathrm{score}} = \E\{\E[\bar h^*(W) \mid \mathfrak{F}]\, S_{\mathrm{mis}}^{\top}\} = 0$. With local permutation the swap invariance holds up to total-variation distance $\varepsilon_n$, and $|\bar h^*| \leq 2 L_\ell$ by (A2), giving $|\E[\bar h^*(W) \mid \mathfrak{F}]| \leq 4 L_\ell \varepsilon_n$ and the stated bound. Because the conditioning in the first step includes $Z$, the argument survives an imputer that conditions on $Z$; it is dependence of the imputer on $Y$ that breaks it. Under the cached strategy of Section~\ref{sec:computation}, whose imputer includes $Y$ as a predictor, $S_{\mathrm{mis}}$ is not $\mathfrak{F}$-measurable, and we write $\Delta_{\mathrm{cache}}$ for the resulting score channel; we do not bound it. The ratio test of Appendix~\ref{app:variance_ratio} monitors the consequence that matters, whether $T$ still over-states $\mathrm{Var}(\hat\mu)$, and Appendix~\ref{app:kappa} quantifies the decision-level consequences (96.0--97.1\% agreement with the per-query strategy).
\end{proof}

\paragraph{Condition (N1): first-order internal-null stability.}
The learner-drift component $\dot\lambda_{\mathrm{train}}$ arises because the full and partial models are trained on $\psi$-imputed data. Writing $g^{\mathrm{fu}}_\psi$ and $g^{\mathrm{pl}}_\psi$ for the population full and partial learners at imputation parameter $\psi$, which coincide at $\psi^*$ under the internal null by the swap symmetry above, we assume
\[
    \textup{(N1)}\qquad \partial_\psi g^{\mathrm{pl}}_\psi \big|_{\psi^*} = \partial_\psi g^{\mathrm{fu}}_\psi \big|_{\psi^*},
\]
which holds in particular when the internal null degrades at most sub-linearly in the imputation parameter, $\sup_{x, z} \|P_{Y \mid \hatX(\psi) = x, Z = z} - P_{Y \mid \hatX(\psi) = x}\|_{\mathrm{TV}} = o(|\psi - \psi^*|)$, and the population training map is Lipschitz in total variation. Under (N1) and the equality in law of $(\hatX, U, Z)$ and $(\hatX, Y, Z)$, the two expectations composing $\dot\lambda_{\mathrm{train}}$ are of the same function under the same law, so $\dot\lambda_{\mathrm{train}} = 0$ and $\dot\lambda = \kappa_{\mathrm{score}} = 0$ by Lemma~\ref{lem:sensitivity}. (N1) holds whenever the internal null holds for all $\psi$ in a neighborhood of $\psi^*$. Under the per-query design this is automatic: the completed set $\hatX(\psi)$ is generated from $(\Xb_{\mathrm{obs}}, R)$ and independent noise for every $\psi$, so $\hatX(\psi) \indep (Y, Z) \mid (\Xb_{\mathrm{obs}}, R)$, and $Y \indep Z \mid (\hatX(\psi), R)$ is equivalent to $Y \indep Z \mid (\Xb_{\mathrm{obs}}, R)$, a condition free of $\psi$; the internal null therefore holds for all $\psi$ or for none, and in the former case $g^{\mathrm{fu}}_\psi \equiv g^{\mathrm{pl}}_\psi$ near $\psi^*$. Under caching, where $\hatX$ depends on $Y$, (N1) has the same status as $\Delta_{\mathrm{cache}}$ above and is monitored by the same empirical ratio test.

\paragraph{Condition (N2): local-permutation accuracy.}
Part (iv) of the theorem, and hence the size bound, needs the placebo's total-variation discrepancy to vanish faster than the sampling error,
\[
    \textup{(N2)}\qquad \sqrt{n}\,\varepsilon_n \to 0,
\]
so that $\sqrt{n}\,\mu_n \to 0$ under the internal null. We state (N2) as an assumption because the available results do not deliver it for our implementation. For the binned local permutation test, \citet[Theorem~2]{kim2022local} bound the excess type-I error by a term of order $n^{1/2} h^{\min(\gamma, 2)}$, for bins of diameter $h$, when the conditional laws are Lipschitz in the generalized Hellinger distance of order $\gamma \geq 1$. If our nearest-neighbor neighborhoods behaved like bins of diameter $h_n \asymp n^{-1/(d+2)}$, as the choice of $k_{\textup{nn}}$ in Appendix~\ref{app:knn} implies, that term would vanish only when $\min(\gamma, 2) > (d+2)/2$, which fails for every $d \geq 2$ whatever the order $\gamma$. Their bound is an upper bound on the discrepancy, not its value, so (N2) may hold where the bound does not show it; empirically, the cells in which PAIR-CI exceeds the nominal level all have $|\Xb| \geq 5$, where the neighborhoods are coarser than at $|\Xb| = 2$ (Section~\ref{sec:exp_ci}).

\paragraph{Full statement of Theorem~\ref{thm:bridge}.}
The main text states Theorem~\ref{thm:bridge} in compressed form: its part (i) combines parts (i)--(iii) below and its part (ii) is part (iv) below. Under the hypotheses of Theorem~\ref{thm:bridge}, with $\mu_n = \E[\hat\mu_m]$, $\sigma^2_{\mathrm{obs}}$ and $\sigma^2_{\mathrm{draw}}$ the sampling and draw-to-draw components of $\mathrm{Var}(\hat\mu_m)$ defined in Step~4 below, and $v = \mathrm{Var}(\bar h^*(W) \mid R, W_{\mathrm{obs}})$, the four parts read in full as follows. (i)~For fixed $M$, $\sqrt{n}(\hat{\mu} - \mu_n) \xrightarrow{d} \mathcal{N}(0, \Sigma_M)$, where $\Sigma_M$ is the sandwich variance of \citet[Theorem~1]{robins2000inference}, which reduces, when the sensitivity $\dot\lambda$ vanishes, to $\Sigma_M = \mathrm{Var}(\bar h^*(W)) - (1 - 1/M)\,\E[v]$, decreasing in $M$ to $\Sigma_\infty = \mathrm{Var}(\E[\bar h^*(W) \mid R, W_{\mathrm{obs}}])$. (ii)~Under the internal null and in the limit of infinitely many imputations, $nT \xrightarrow{P} \Sigma_{\mathrm{Rubin}}$ with $\Sigma_{\mathrm{Rubin}} - \Sigma_\infty = 2\,\E[v] \geq 0$, the scalar case of \citet[Theorem~2]{robins2000inference} with $\dot\lambda = 0$ and $\tau = 1$, the latter because the parameter enters the estimating function additively. (iii)~Under the internal null and for fixed $M$, $\E[T] - \mathrm{Var}(\hat\mu)$ is non-negative to first order and equals $2\,\E[v]/n$ to that order for every $M$: $T$ charges the draw-to-draw noise at $(1 + 1/M)$ against the $1/M$ share at which it enters $\mathrm{Var}(\hat\mu)$, and $\bar W$ exceeds $\sigma^2_{\mathrm{obs}}$ by the per-row imputation variance once $\dot\lambda$ vanishes; without that condition the sign is not determined \citep[Table~1]{robins2000inference}. (iv)~If in addition $\sqrt{n}\,\mu_n \to 0$ under the internal null, which the local permutation delivers when $\sqrt{n}\,\varepsilon_n \to 0$, then in the limit of infinitely many imputations $\tCI \xrightarrow{d} \mathcal{N}(0, \Sigma_\infty/\Sigma_{\mathrm{Rubin}})$, whose upper-tail probabilities are no larger than those of $\mathcal{N}(0,1)$ at every non-negative threshold; since $t_{\alpha,\nu} \geq z_\alpha$ for every $\nu > 0$ and $\alpha < 1/2$, $\lim_{M} \limsup_{n} P(\tCI > t_{\alpha,\nu}) \leq \alpha$, with equality when $\sigma^2_{\mathrm{draw}}$ is negligible relative to $\sigma^2_{\mathrm{obs}}$. For fixed $M$, $nT$ converges in distribution rather than in probability and the limit law of $\tCI$ is Rubin's $t$-type pivot; the size bound then rests on the Barnard--Rubin reference distribution dominating that limit law in the upper tail, which we do not prove and which Appendix~\ref{app:variance_ratio} examines empirically.

\begin{proof}[Proof of Theorem~\ref{thm:bridge}]
\emph{Step 1 (Within-imputation CLT).} By Lemma~\ref{lem:within_ral}, for each $m$, $\hat{\mu}_m$ is regular and asymptotically linear with a consistent variance estimate
$V_m^{\mathrm{Bayle}} = \hat{\sigma}^2_{m,\mathrm{in}} / n$
satisfying $n \cdot V_m^{\mathrm{Bayle}} / \sigma^2_m \xrightarrow{P} 1$.

\emph{Step 2 (Asymptotic normality).} Although our analysis model is nonparametric, asymptotic linearity implies that $\hat{\mu}_m$ is equivalent to an estimator admitting an estimating-equation representation, placing it within the framework of \citet{robins2000inference}. Under Assumptions~(A1$'$), (A3) and~\ref{ass:regularity}, their Conditions S1--S4 hold for this representation and their in-text conditions on the imputation-parameter estimate are satisfied by the MICE fit in the form of their Section~5, so their Theorem~1 yields $\sqrt{n}(\hat{\mu} - \mu_n) \xrightarrow{d} \mathcal{N}(0, \Sigma_M)$ as $n \to \infty$ with $M$ fixed, where $\mu_n = \E[\hat\mu_m]$ and $\Sigma_M$ is their sandwich variance. When $\dot\lambda = 0$ the sandwich reduces to the variance of the imputation-averaged influence function $M^{-1}\sum_m \bar h^*(W^{(m)})$; conditioning on $(R, W_{\mathrm{obs}})$, this is $\Sigma_M = \mathrm{Var}(\bar h^*(W)) - (1 - 1/M)\,\E[\mathrm{Var}(\bar h^*(W) \mid R, W_{\mathrm{obs}})]$, which decreases in $M$ to $\Sigma_\infty = \mathrm{Var}(\E[\bar h^*(W) \mid R, W_{\mathrm{obs}}])$. This is part (i); it uses S1--S4 with any finite $\dot\lambda$, so it holds under the alternative as well, where for forests it relies on Assumption~\ref{ass:regularity}(d$'$).

\emph{Step 3 (Sign of the bias of $T$, infinitely many imputations).} Under the internal null, Lemma~\ref{lem:sensitivity} with $\varepsilon_n \to 0$ and condition (N1) give $\dot\lambda = 0$. \citeauthor{robins2000inference}'s Theorem~2 then states that $nT$ converges in probability to a limit $\Sigma_{\mathrm{Rubin}}$ with, in the scalar case relevant here,
\[
    \Sigma_{\mathrm{Rubin}} - \Sigma_\infty = 2\,\E\!\left[\mathrm{Var}\!\left(\bar{h}^*(W) \mid R,\, W_{\mathrm{obs}}\right)\right] \geq 0 ,
\]
where the factor 2 arises because their bias term is $G_{\mathrm{RW}} + G_{\mathrm{RW}}'$ with $G_{\mathrm{RW}} = \E[\mathrm{Var}(\bar h^*(W) \mid R, W_{\mathrm{obs}})]$ once the sensitivity terms vanish, and $\tau = -\E[\partial \mathcal{U} / \partial \beta] = 1$ because $\mathcal{U} = \bar h^*(W) - \beta$ is linear in $\beta$. This is part (ii).

\emph{Step 4 (Fixed $M$).} Decompose $\hat\mu_m = \hat\mu_{\mathrm{data}} + \eta_m$, where $\eta_m$ collects the pipeline draw noise (imputation draw, fold split, placebo draw, learner seed), conditionally independent and identically distributed across $m$ given the data with $\E[\eta_m \mid \text{data}] = 0$ and $\mathrm{Var}(\eta_m) = \sigma^2_{\mathrm{draw}}$, which is what proper posterior draws (A3) and the per-imputation redrawing of folds, placebos, and learner seeds provide. Then $\mathrm{Var}(\hat\mu) = \sigma^2_{\mathrm{obs}} + \sigma^2_{\mathrm{draw}}/M$, since $\hat\mu_{\mathrm{data}} = \E[\hat\mu_m \mid \text{data}]$ removes the cross term, and $\E[B] = \sigma^2_{\mathrm{draw}}$ exactly, since $\hat\mu_m - \hat\mu = \eta_m - \bar\eta$. Hence
\[
    \E[T] - \mathrm{Var}(\hat\mu) = \bigl(\E[\bar W] - \sigma^2_{\mathrm{obs}}\bigr) + \sigma^2_{\mathrm{draw}} ,
\]
because $(1 + 1/M) - 1/M = 1$. The second term is non-negative. For the first, $\bar W$ is computed on completed dataset $m$, whose rows are draws from the completed-data law, so it captures the per-row imputation noise that $\sigma^2_{\mathrm{obs}}$ excludes and misses only the component shared across rows through $\psi_m$, which enters at first order only through $\dot\lambda$ and therefore drops out under the internal null by Lemma~\ref{lem:sensitivity} and (N1); explicitly, with $v = \mathrm{Var}(\bar h^*(W) \mid R, W_{\mathrm{obs}})$, $n(\E[\bar W] - \sigma^2_{\mathrm{obs}}) \to \E[v]$ when $\dot\lambda = 0$, whereas with $\dot\lambda \neq 0$ the limit acquires the sign-indefinite terms of \citeauthor{robins2000inference}'s Theorem~2 and can be negative. Under the internal null, therefore, $n(\E[T] - \mathrm{Var}(\hat\mu)) \to 2\E[v]$ for every fixed $M$, the constant of Step 3. This is part (iii). We state the direction of this bias and its first-order constant, not a finite-sample identity: the ratio test of Appendix~\ref{app:variance_ratio} shows the measured excess ordered as predicted but not equal to $\sigma^2_{\mathrm{draw}}$, and below one in three of twenty-four cells on the raw scale.

\emph{Step 5 (Size).} Suppose $\sqrt{n}\,\mu_n \to 0$ under the internal null, which requires the local permutation's discrepancy to satisfy condition (N2), $\sqrt{n}\,\varepsilon_n \to 0$, since $|\mu_n| \leq 4 L_\ell \varepsilon_n$ by the bound of Lemma~\ref{lem:sensitivity}. In the limit of infinitely many imputations, parts (i)--(ii) and Slutsky's lemma give $\tCI = \hat\mu/\sqrt{T} \xrightarrow{d} \mathcal{N}(0, \Sigma_\infty/\Sigma_{\mathrm{Rubin}})$ with $\Sigma_\infty/\Sigma_{\mathrm{Rubin}} \leq 1$, whose upper-tail probabilities are no larger than those of $\mathcal{N}(0,1)$ at every non-negative threshold; since $t_{\alpha,\nu} \geq z_\alpha$ for every $\nu > 0$ and $\alpha < 1/2$, $\lim_M \limsup_n P(\tCI > t_{\alpha,\nu}) \leq \alpha$, with equality when $\Sigma_{\mathrm{Rubin}} = \Sigma_\infty$, that is, when $\sigma^2_{\mathrm{draw}}$ is negligible relative to $\sigma^2_{\mathrm{obs}}$. This is part (iv). For fixed $M$, $nB$ converges in distribution rather than in probability, so $nT$ has a random limit and the limit law of $\tCI$ is Rubin's $t$-type pivot; the size bound then rests on the Barnard--Rubin reference distribution dominating that limit law in the upper tail, which we do not prove and which Appendix~\ref{app:variance_ratio} examines empirically.

\emph{Degrees of freedom.} The $t_\nu$ critical values use \citeauthor{barnard1999small}'s \citeyearpar{barnard1999small} adjustment $\nu = (1/\nu_M + 1/\hat\nu_{\mathrm{obs}})^{-1}$, with $\nu_M = (M - 1)\hat\gamma^{-2}$ Rubin's finite-$M$ degrees of freedom and $\hat\nu_{\mathrm{obs}} = \frac{\nu_{\mathrm{com}} + 1}{\nu_{\mathrm{com}} + 3}\,\nu_{\mathrm{com}}(1 - \hat\gamma)$ their small-sample correction for finite complete-data degrees of freedom, where $\hat\gamma = (1 + 1/M)B/T$. The complete-data degrees of freedom are $\nu_{\mathrm{com}} = n - K$: the within-fold variance estimator of \citet{bayle2020cross} pools the $n$ held-out per-observation loss differences after removing $K$ fold means, so each of the $K$ folds contributes $n/K - 1$; \citeauthor{bayle2020cross} themselves use a normal reference distribution and prescribe no degrees of freedom. (The submitted version stated $K - 1$, which did not match our implementation.) As $M \to \infty$ the between-imputation contribution to $\nu$ becomes unbounded and $\nu$ rises to $\hat\nu_{\mathrm{obs}} \leq \nu_{\mathrm{com}} = n - K$; with $n \to \infty$ as well, $\nu \to \infty$ and the correction vanishes.
\end{proof}

\paragraph{Downstream results.} Lemma~\ref{lem:within_ral} is untouched by the weakening from exact to conservative size. Corollary~\ref{cor:pc_consistency} uses only an upper bound on the type-I error, and Proposition~\ref{prop:consistency}'s divergence argument tolerates a constant-factor inflation of the variance, since $T = O_P(n^{-1})$ under the alternative by Lemma~\ref{lem:within_ral} and (R1)--(R2); neither is affected.

\section{Imputation Bias Decomposition}
\label{app:kappa}

The PAIR-CI test statistic targets the population loss difference $\mu = \E[\hat{\mu}]$, which vanishes under the internal null $Z \indep Y \mid \hatX$ (Proposition~\ref{prop:calibration}). Calibration for the scientific null $Z \indep Y \mid \Xb$ requires that the two nulls coincide, which Assumption~(A1) enforces asymptotically. Under MNAR, Assumption~(A1) fails and $\mu$ acquires a bias $\kappa$ under the scientific null. Defining $\delta = \Xb_{\textup{miss}} - \hatX_{\textup{miss}}$ as the imputation residual, we now decompose $\kappa$ into three components.

\begin{definition}[Components of imputation bias]
    \label{def:kappa_factors}
    For the imputation residual $\delta$,
    \begin{align*}
      \kappa_{\textup{imp}} &= \sqrt{\E\!\left[\|\delta\|^2\right]} \\
      \kappa_Y &= \big|\textup{corr}\!\left(Y,\delta \mid \hatX\right)\big| \\
      \kappa_Z &= \big|\textup{corr}\!\left(Z,\delta \mid \hatX\right)\big|
    \end{align*}
    denote the imputation-error magnitude and the residual correlations of $Y$ and $Z$ with $\delta$ given $\hatX$, respectively.
\end{definition}

\begin{decomposition}[Qualitative bias decomposition]
    \label{dec:kappa}
    A first-order expansion of the conditional loss difference around $\delta = 0$, assuming $Y \indep Z \mid (\hatX, \delta)$ under the scientific null, yields a leading-order bias proportional to $\kappa_Y \cdot \kappa_Z$. As a sufficient condition,
    \[
    \kappa_{\textup{imp}} = 0 \;\text{or}\; \kappa_Y = 0 \;\text{or}\; 
    \kappa_Z = 0 \;\Longrightarrow\; \kappa = 0.
    \]
    The converse need not hold when $\delta$ is multivariate, since $Y$ and $Z$ may depend on $\delta$ along non-overlapping directions.
\end{decomposition}

Two design choices keep $\kappa$ small. First, the paired comparison fits the full and partial models on the same $\hatX$, causing imputation bias to enter both losses identically and cancel in the difference. The residual $\kappa$ is thus controlled by the \emph{differential} dependence of $Y$ and $Z$ on $\delta$. Second, cached imputation includes $Y$ and $Z$ as predictors, absorbing into $\hatX_{\textup{miss}}$ the components of $\delta$ that are correlated with the test variables and thereby driving $\kappa_Y$ and $\kappa_Z$ toward 0. In practice, $\kappa$ tends to inflate only under adversarial nonlinear DGPs where a linear imputer cannot capture the dependence of the test variables on $\Xb_{\textup{miss}}$ (Figure~\ref{fig:kappa_fpr_trace}, Appendix~\ref{app:adversarial_kappa}).

\begin{remark}[How the paired design achieves MNAR robustness]
    \label{rem:paired_robust}
    Under Assumptions~(A2)--(A4), conditional permutation ensures that $U \mid \hatX \approx_d Y \mid \hatX$. Extending this calibration guarantee to the scientific null $Z \indep Y \mid \Xb$ requires imputation consistency (Assumption~A1), which holds under MAR but not MNAR. The residual bias $\kappa$ decomposes into three factors: imputation error, the residual correlation of $Y$ with the unrecovered component of $\Xb$, and the residual correlation of $Z$ with that component (Decomposition~\ref{dec:kappa}). If any factor vanishes, $\kappa = 0$. The paired comparison eliminates imputation bias shared by both models, while caching further reduces $\kappa$ by absorbing the test variables' correlation with $\Xb_{\textup{miss}}$ into $\hatX$. Crucially, this cancellation occurs when error varies randomly across observations---as with MICE under MNAR---but fails when it is systematic, as with mean or marginal imputation, where both models receive the same directional bias. Together, these design choices keep $\kappa$ small in practice, as corroborated in Figure~\ref{fig:kappa_fpr_trace}, Appendix~\ref{app:adversarial_kappa}.
\end{remark}

The expansion underlying Decomposition~\ref{dec:kappa} assumes approximately linear dependence of $Y$ and $Z$ on $\delta$ given $\hatX$. Under strongly nonlinear structures, higher-order terms may dominate, as observed empirically in the hub-nonlinear setting. Separately, any residual miscalibration at finite $M$ is guarded against by the conservatism of the combined variance estimator (Theorem~\ref{thm:bridge}(i), parts (ii)--(iii) of its full statement in Appendix~\ref{app:bridge}; Appendix~\ref{app:variance_ratio}), not by the Barnard--Rubin degrees-of-freedom adjustment, whose contribution to the under-rejection is minor.

\paragraph{Caching bias bound.}
Let $\hatX^{\textup{cache}}$ denote the cached imputation (fit with $Z$ and $Y$ included) and $\hatX^{\textup{excl}}$ the per-query imputation (fit without $Z$ and $Y$). Under $H_0$, the bias of the cached test relative to the per-query test can be defined as
\[
    \kappa_{\textup{cache}} \;=\; \E[\hat{\mu}^{\textup{cache}}] - \E[\hat{\mu}^{\textup{excl}}].
\]
The imputer targets $P(\Xb_{\textup{miss}} \mid \Xb_{\textup{obs}}, Z, Y)$. $Z$ contributes no information about $\Xb$ beyond $(Y, \Xb_{\textup{obs}})$ when $Z \indep \Xb \mid Y$ (a typical configuration when $Z$ is a descendant of $\Xb$). Hence, $Z$'s incremental contribution is controlled by its partial $R^2$ in the imputation model, given $\Xb_{\textup{obs}}$ and $Y$, denoted by $\rho_Z^2$.

For regularized imputers (e.g., Bayesian ridge regression with $p$ candidate predictors), the marginal contribution of a single predictor is $O(1/p)$, on average. The paired comparison cancels the component of imputation error shared by both models (since both use $\hatX^{\textup{cache}}$), leaving a residual bias proportional to $\rho_Z^2 \cdot \textup{signal}$ that vanishes under $H_0$. Under $H_1$, $\kappa_{\textup{cache}}$ shifts in the direction of the alternative, slightly increasing power rather than reducing conservatism.

\paragraph{Empirical validation.}
To validate imputation caching---the strategy of imputing once and using completed datasets across all CI tests---we compare it with a per-query variant that re-imputes while excluding the test variables. We consider $p \in \{5, 10, 20\}$ under MAR, using the general variant of PAIR-CI throughout (20 replicates per setting). Agreement rates are 96.0\%, 97.1\%, and 96.9\%, respectively, with median $p$-value discrepancies of 0.05--0.07. Disagreements do not systematically favor rejection over non-rejection. While decision-level agreement need not imply graph-level agreement, graph recovery results (Section~\ref{sec:exp_graph}) show no performance loss, indicating that this pathology does not arise in practice. We therefore use cached imputation throughout the graph recovery experiments.

\section{Power Plots}
\label{app:power_plots}

Under $H_1$ (signal $> 0$), PAIR-CI exhibits lower power than Fisher's $Z$ on complete Gaussian data---an expected drawback of nonparametric generality. In the linear Gaussian DGP with signal $= 0.3$, the average rejection rate (across missingness mechanisms and $|\Xb| \in \{2, 5\}$) is approximately 0.40 at $n = 500$, compared to nearly 1.00 for Fisher's $Z$ on complete data. Power rises to 0.94--1.00 at signal $= 1.0$ across all DGPs and sample sizes. Rejection rates are lower for the post-nonlinear setting at weak signals, reflecting the more challenging learning problem, while the latent confounder case lies between the two. Across all conditions, power under MNAR is slightly reduced relative to complete data and MAR, most noticeably at signal $= 0.3$.

In practice, PAIR-CI reliably detects moderate-to-strong conditional dependencies (signal $\geq 0.6$ and $n \geq 2{,}000$), with power exceeding 80\% across all DGPs. Sensitivity to weak effects at small $n$ remains limited---a regime in which even a calibrated test provides little actionable information. Although baseline methods exhibit higher rejection rates under the alternative, this comparison is uninformative: a test that rejects 45\% of true nulls will also reject under most true alternatives, and such rejections carry limited evidential value. For causal discovery, this tradeoff favors calibration: false positives (spurious edges) propagate errors through orientation, whereas false negatives (missing edges) primarily yield sparser graphs.

\begin{figure}[t]
    \centering
    \begin{subfigure}{\textwidth}
        \centering
        \includegraphics[width=\textwidth]{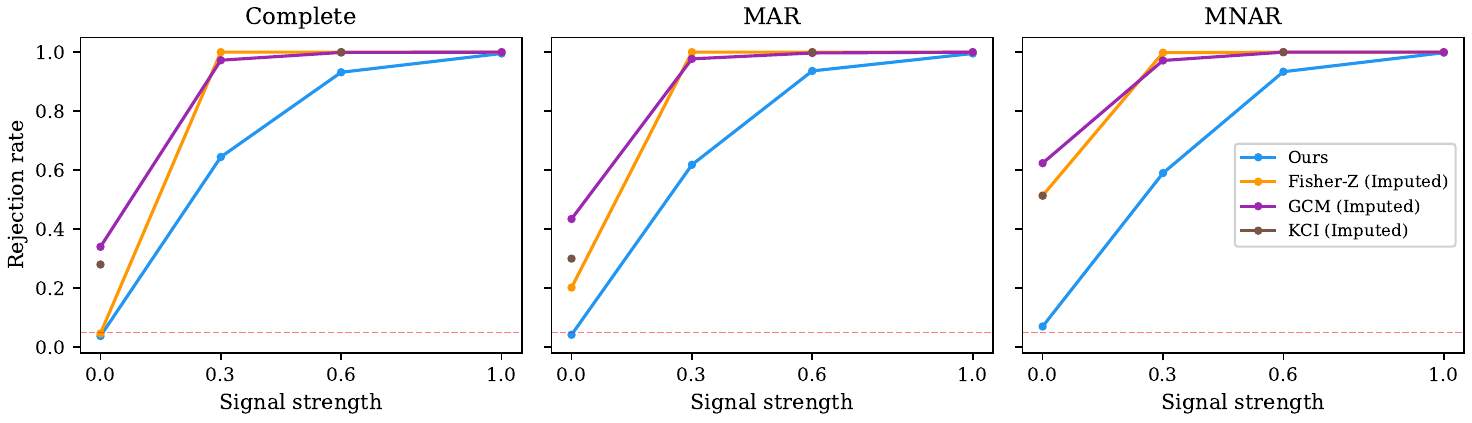}
        \caption{Linear Gaussian DGP}
        \label{fig:power_lg}
    \end{subfigure}
    
    \vspace{0.5em}
    
    \begin{subfigure}{\textwidth}
        \centering
        \includegraphics[width=\textwidth]{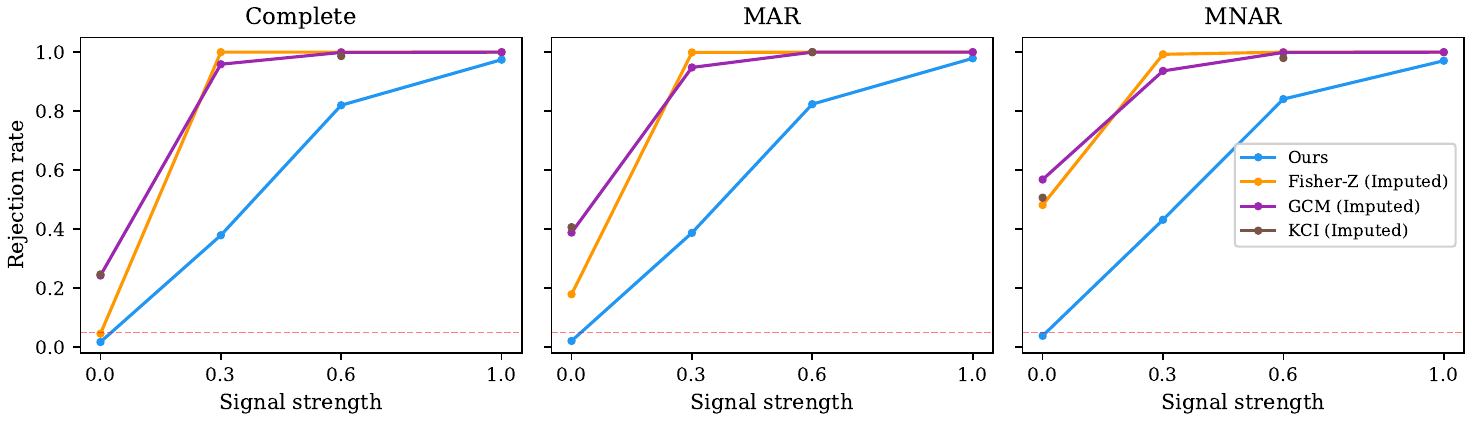}
        \caption{Post-nonlinear DGP}
        \label{fig:power_pnl}
    \end{subfigure}
    
    \vspace{0.5em}
    
    \begin{subfigure}{\textwidth}
        \centering
        \includegraphics[width=\textwidth]{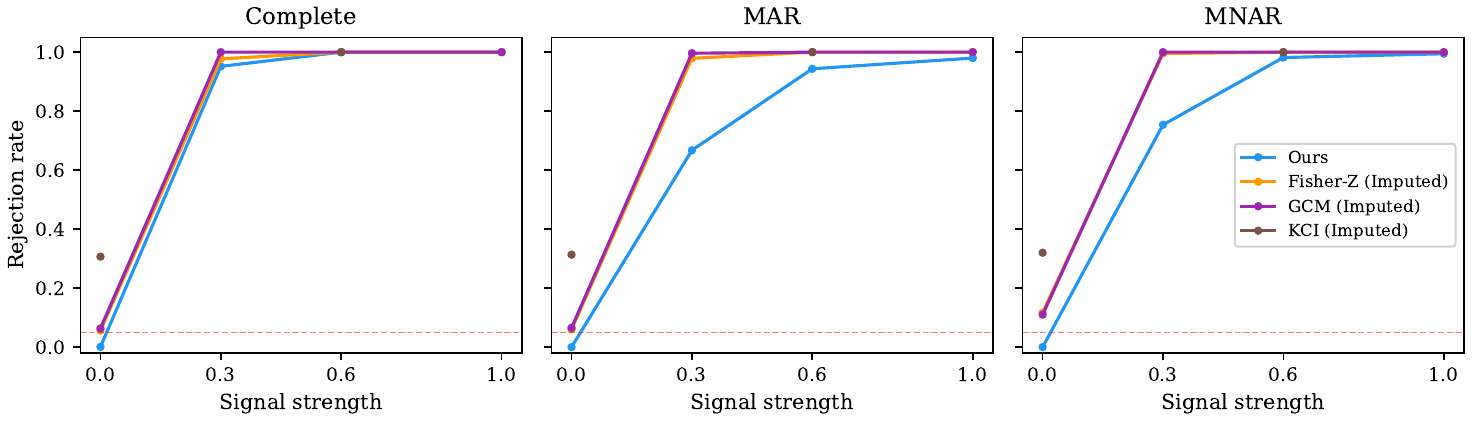}
        \caption{Latent confounder DGP}
        \label{fig:power_lc}
    \end{subfigure}

    \caption{\textbf{Power curves in standalone performance experiment.} Rejection rate under $H_1$ by signal strength and missingness mechanism for three DGPs (linear Gaussian, post-nonlinear, and latent confounder; described in Section~\ref{sec:exp_ci}). PAIR-CI (blue) has lower power than the miscalibrated baselines by design, so comparisons are uninformative (Table~\ref{tab:calibration}). Power exceeds 80\% at signal $\geq 0.6$ and $n \geq 2{,}000$ across all DGPs.}
    \label{fig:power_all}
\end{figure}

\section{Implementation Details}
\label{app:implementation}

\paragraph{Choice of learner.}
\label{app:learner_choice} In Step~2, any supervised learner satisfying modest stability conditions (Theorem~\ref{thm:bridge}) can be used to train the full and partial models. In subsequent validation experiments, we implement two variants of PAIR-CI: a \textit{general} variant based on random forests, which are consistent under the additive model of \citet{scornet2015consistency} and which we assume satisfy the stability conditions of Section~\ref{sec:theory}; and a \textit{fast} variant based on ExtraTrees \citep{geurts2006extremely} with early stopping, deployed at larger graph sizes where the PC algorithm requires hundreds of CI tests.\footnote{A comparison of random forests with ExtraTrees and other learners is conducted in Appendix~\ref{app:learner_comparison}. Feature bagging is adjusted to guarantee that the candidate variable is considered at every split when $p < 12$ (Appendix~\ref{app:implementation}).} Implementation details are provided in Appendices~\ref{app:implementation} and~\ref{app:early_stopping}. On the graph benchmarks with linear edges (the Erd\H{o}s--R\'enyi graphs with linear edges and ALARM with linear structural equations), both variants replace the forest with ridge regression with a negligible penalty ($10^{-6}$), matching the learner to the functional form as the parametric baselines do. This uses knowledge of the edge type that an analyst would have to supply; Appendix~\ref{app:learner_robustness} shows that a misspecified linear learner remains calibrated on nonlinear data, so the choice affects power rather than size.

\paragraph{PC algorithm with cached imputations.}
Algorithm~\ref{alg:pc} details the full procedure. It departs from the standard PC adjacency search of \citet[Algorithm~1]{kalisch2007estimating} in three ways: the upfront imputation step (Line~1), whose outputs are reused across all CI tests; conditioning sets drawn from the neighborhood of one endpoint of each pair only; and termination at the first level at which no edge is removed. PAIR-CI and the in-house baselines share the search, so their comparison holds it fixed, whereas MVPC runs the PC-stable search. On HAILFINDER, FZ-Rubin, MVPC and the $\alpha$-swept baselines of Appendix~\ref{app:pr_frontier}, and on the nonlinear Sachs benchmark, FZ-Rubin and MissDAG, were run in separate jobs on datasets drawn from the same distribution as, but not identical to, those of the other methods. Drawing conditioning sets from the neighborhoods of both endpoints, and not stopping early, extends the search, reduces the baselines' spurious adjacencies in every replicate, and narrows, but does not remove, PAIR-CI's advantage on the nonlinear networks: with the PC-stable search \citep[Algorithm~4.1]{colombo2014order} on the data of Table~\ref{tab:realworld} (20\% MNAR), test-wise deletion and FZ-vote lower their median SHD by 7 and 7.5 on ALARM and by 24 and 18.5 on HAILFINDER, while PAIR-CI, still on the simplified search, keeps a lower median SHD and a higher skeleton F1 than either of them, and than MVPC, in both cells, with a lower SHD in each of the 20 paired replicates against both rerun baselines.

\begin{algorithm}[t]
    \caption{PC with Integrated CI Oracle}
    \label{alg:pc}
    \begin{algorithmic}[1]
        \REQUIRE Data $D$ (may contain missing values), significance level $\alpha$
        \STATE Impute $D$ to obtain $\hat{D}^{(1)}, \ldots, \hat{D}^{(M)}$ 
        (cache for reuse)
        \STATE Initialize complete undirected graph $G$ on $p$ nodes
        \FOR{$d = 0, 1, 2, \ldots$}
            \FOR{each edge $(i, j)$ in $G$ with 
            $|\textup{adj}(i) \setminus \{j\}| \geq d$}
                \FOR{each $\mathbf{S} \subseteq \textup{adj}(i) \setminus \{j\}$ 
                with $|\mathbf{S}| = d$}
                    \STATE Compute $t_{\mathrm{CI}}$ for $X_i \indep X_j \mid 
                    \mathbf{S}$ using cached $\hat{D}^{(m)}$
                    \IF{$p\textup{-value} > \alpha$}
                        \STATE Remove edge $(i,j)$; record $\mathbf{S}$ as 
                        separating set; \textbf{break}
                    \ENDIF
                \ENDFOR
            \ENDFOR
            \IF{no edges removed at depth $d$}
                \STATE \textbf{break}
            \ENDIF
        \ENDFOR
        \STATE Orient edges via v-structures and Meek's rules
        \STATE \textbf{return} CPDAG estimate $\hat{G}$
    \end{algorithmic}
\end{algorithm}

\paragraph{Neighborhood size of the local permutation.}
\label{app:knn}
We use a local (nearest-neighbor) scheme in the spirit of the local permutation of \citet{runge2018conditional}: each row's placebo value is the $Y$-value of a neighbor drawn at random from its $k_{\textup{nn}}$ nearest neighbors in $\hatX$-space other than itself, independently across rows, so the placebo is a local resample drawn with replacement rather than a permutation (we keep the name local permutation for this step), with $k_{\textup{nn}} = \max(2, \lfloor n^{2/(d+2)} \rfloor)$, $d = |\Xb|$: for a Lipschitz density in dimension $d$ the optimal kernel bandwidth is $h_n \asymp n^{-1/(d+2)}$ \citep{stone1980optimal,tsybakov2009introduction}, and $n h_n^d \asymp n^{2/(d+2)}$ sample points fall within it. The Type I error cost of a local permutation is governed by the within-neighborhood discrepancy between the conditional laws of $U$ and $Y$ given $\hatX$; for the binned version, the total-variation distance between the binned law and its conditional-independence projection, which this within-neighborhood variation drives, bounds the excess type-I error \citep[Lemma~1]{kim2022local}, and \citet[Theorem~2]{kim2022local} bound it under Lipschitz conditions. That bound does not vanish at this neighborhood scale when $d \geq 2$, so the size bound rests on condition (N2) of Appendix~\ref{app:bridge}: at $|\Xb| \geq 5$ and the sample sizes we consider the neighborhoods are coarse ($k_{\textup{nn}} \leq 11$), which is where the cells above the nominal level in Table~\ref{tab:calibration_full} arise.

\paragraph{Software.} We implement the algorithm in Python, using the scikit-learn library \citep{pedregosa2011scikit} for random forests and MICE (\texttt{IterativeImputer} with Bayesian ridge regression). The code is compatible with the \texttt{causal-learn} library for constraint-based causal discovery \citep{zheng2024causallearn}, whose KCI implementation and released MVPC code we use for those baselines. Third-party code and data are used under their licences: \texttt{causal-learn} is MIT-licensed; the FZ-Rubin baseline is a Python port of the R package \texttt{micd} (GPL-3) and is released under that licence; the MissDAG baseline \citep{gao2022missdag} runs a Gaussian linear reimplementation whose E-step is adapted from the authors' public code, which specifies no licence, so that module is not redistributed with our code, and whose linear NOTEARS M-step \citep{zheng2018dags} is modified from the Apache-2.0 reference implementation. The Sachs data \citep{sachs2005causal} are the copy distributed by the bnlearn repository \citep{scutari2010learning}, which our code downloads rather than redistributes, and the ALARM and HAILFINDER experiments use only the published graph structures \citep{beinlich1989alarm,abramson1996hailfinder}, with simulated data. Since the default \texttt{IterativeImputer} settings in scikit-learn do not perform proper posterior sampling, as required by Assumption~(A3), we set \texttt{sample\_posterior=True}. Conditional permutation uses the nearest-neighbor local permutation of \citet{runge2018conditional} with neighborhood size $k_{\mathrm{nn}} = \max(2, \lfloor n^{2/(d+2)} \rfloor)$ (Section~\ref{sec:ci_test}). Default settings are reported in Table~\ref{tab:defaults}.

\begin{table}[h]
    \centering
    \caption{\textbf{Default hyperparameters for PAIR-CI and PC algorithms}.}
    \label{tab:defaults}
    \begin{tabular}{lll}
        \toprule
        Parameter & Default & Description \\
        \midrule
        $M$ & 5 & Multiply imputed datasets \\
        $K$ & 5 \mbox{(fast)} / 10 \mbox{(general)} & CV folds per imputed dataset \\
        $n_{\textup{trees}}$ & 100 & Trees in random forest \\
        min\_samples\_leaf & 5 & Leaf regularization \\
        max\_subsample & 2{,}000 & Observation cap per CI test \\
        $\alpha$ & 0.05 & Significance level for PC \\
        \bottomrule
    \end{tabular}
\end{table}

\paragraph{Variants and feature bagging.}
We disable feature bagging (\texttt{max\_features} $=$ None) when $p < 12$ to ensure that the candidate variable is included at every split. For $12 \leq p \leq 80$, we set \texttt{max\_features} $= 12$; for $p > 80$, \texttt{max\_features} $= \sqrt{p}$. In the scaling experiments ($p \geq 20$), we employ the fast variant of PAIR-CI: ExtraTrees classifiers \citep{geurts2006extremely} with $M = 5$, $K = 5$, 100 trees, and early stopping (skipping remaining imputations if the absolute $t$-statistic exceeds 4.0 after two imputations; see Appendix~\ref{app:early_stopping} for calibration validation).

\paragraph{Loss function selection.}
Variables with $\leq 20$ unique values are treated as discrete (binary cross-entropy loss; classification), others as continuous (mean squared error loss; regression).

\paragraph{Nonlinear edge mechanisms.}
In the nonlinear graph recovery experiments (Section~\ref{sec:exp_graph}), each edge is assigned one of four nonlinear functions, drawn uniformly at random:
\begin{align}                                
      f_1(x, w) &= w x^2 && \text{(quadratic)} \\
      f_2(x, w) &= w \sin(2x) && \text{(sinusoidal)} \\
      f_3(x, w) &= w |x| && \text{(absolute value)} \\
      f_4(x, w) &= w \tanh(1.5 x) && \text{(saturating)}
\end{align}                                                                                    
\subsection{Sensitivity to Hyperparameters}
\label{app:sensitivity}

We assess the sensitivity of graph recovery performance to four key hyperparameters, varying one at a time from the main-text default configuration ($M = 5$, $K = 5$, 100 trees, $\texttt{min\_samples\_leaf} = 5$). The experiment involves 10 random graphs with nonlinear edge mechanisms, $p = 10$, $n = 1{,}000$, and MAR missingness (affecting approximately 11\% of cells), with 20 datasets per graph. Test-wise deletion and FZ-vote are included as baselines; MVPC has no counterpart to these hyperparameters and is not included.

At the default configuration, PAIR-CI achieves median SHD~10 (IQR 8--11), against 11 (9--13) for both test-wise deletion and FZ-vote. Each of the four analyses below varies one hyperparameter while holding the others fixed.

\paragraph{Number of imputations ($M$).} Median SHD is invariant across $M \in \{3, 5, 10\}$, with all settings yielding SHD~10 (IQR 8--11) and median F1 rising modestly from 0.700 at $M = 3$ to 0.720 at $M = 10$. Runtime scales near-linearly with $M$ (91s, 139s, 256s). We default to $M = 5$: although between-imputation variance shrinks at a rate of $1/M$, the dominant contribution to the test statistic's variance is within-imputation prediction noise, which is not affected by $M$. The marginal F1 gain at $M = 10$ does not justify the doubled runtime.

\paragraph{Number of CV folds ($K$).} Median SHD and F1 are invariant across $K \in \{3, 5, 10\}$---SHD~10 (IQR 8--12) and F1 $= 0.706$ in all settings---while runtime scales linearly (82s, 139s, 276s for $K = 3, 5, 10$). This stability stems from the \citeauthor{bayle2020cross} within-imputation variance estimator, which absorbs the fold-correlation correction that would otherwise penalize smaller $K$. We default to $K = 5$ for the fast variant, balancing learner quality and runtime, and to $K = 10$ for the general variant, where the absence of a fold-correlation penalty justifies the higher $K$ (Appendix~\ref{app:bridge}).

\paragraph{Number of trees.} Median SHD is invariant across $n_{\textup{trees}} \in \{20, 50, 100, 200\}$ at SHD~10, with F1 stable at 0.700--0.706 and runtime growing roughly linearly (71s, 102s, 139s, 241s). Performance appears to plateau at $n_{\textup{trees}} = 50$. We default to 100 as a conservative margin, though smaller forests likely suffice at $p = 10$.

\paragraph{Minimum samples per leaf.} Median SHD and F1 are stable across $\texttt{min\_samples\_leaf} \in \{5, 10, 20\}$ (SHD~10, F1 $\approx 0.706$, runtime 137--140s). At $\texttt{min\_samples\_leaf} = 1$, performance degrades: SHD rises to 11 (IQR 9--13) and F1 drops to 0.571, as fully grown trees produce high-variance per-fold predictions whose noise inflates the loss-difference statistic beyond what the \citeauthor{bayle2020cross} estimator corrects. The lower bound $\texttt{min\_samples\_leaf} = 5$ is thus consequential, with larger values performing similarly and smaller values inflating SHD and reducing F1.

\paragraph{Summary.} Median SHD at $p = 10$ is largely insensitive to $M$, $K$, and $n_{\textup{trees}}$: all reasonable configurations deliver SHD~10 (IQR 8--12), consistently below the baselines at SHD~11. The only consequential hyperparameter is $\texttt{min\_samples\_leaf}$, with values below 5 degrading both SHD (to 11) and F1 (0.571 vs.\ 0.706). Increasing $M$ or $K$ beyond their defaults nearly doubles runtime for at most a 1.4 percentage point gain in F1. We therefore default to $M = K = 5$, reserving $K = 10$ for the general variant, where \citeauthor{bayle2020cross}'s estimator eliminates the fold-correlation penalty.

\section{Robustness to Missingness Rate and Sample Size}
\label{app:robustness}
We evaluate graph recovery at $p = 10$ across all combinations of 
missingness rate $\in \{10\%, 30\%, 50\%\}$, sample size $n \in \{500, 1{,}000, 2{,}000\}$, and mechanism (MAR, MNAR), yielding 18 conditions per edge type (10 graphs, 20 datasets per graph). Tables~\ref{tab:robustness_mar} and~\ref{tab:robustness_mnar} report the full MAR and MNAR sweeps with nonlinear edges, with MVPC run on the same datasets. Across all 18 conditions PAIR-CI and MVPC stay within one SHD unit of each other, and both are at or below test-wise deletion and FZ-vote in every condition. Table~\ref{tab:shd_combined} reports median SHD across all missingness mechanisms for both linear and nonlinear edges at the default 30\% missingness rate.

\begin{table}[h]
    \centering
    \caption{\textbf{Robustness to missingness rate and sample size, MAR.} Median SHD at $p = 10$ with nonlinear edges and MAR missingness. $^\ddagger$MVPC runs in its authors' released PC implementation.}
    \label{tab:robustness_mar}
    \begin{tabular}{lccccccccc}
    \toprule
                  & \multicolumn{3}{c}{$n{=}500$} & \multicolumn{3}{c}{$n{=}1000$} & \multicolumn{3}{c}{$n{=}2000$} \\
    \cmidrule(lr){2-4} \cmidrule(lr){5-7} \cmidrule(lr){8-10}
    Method        & 10\% & 30\% & 50\%            & 10\%  & 30\%  & 50\%           & 10\%  & 30\%  & 50\%  \\
    \midrule
    PAIR-CI       &  9.0 & 10.0 & 10.0 &  9.0 &  9.0 &  9.0 &  9.0 &  9.0 &  9.0 \\
    Complete case &  9.0 &  9.0 &  9.5 &  9.0 &  9.0 &  9.0 & 10.0 &  9.0 &  9.0 \\
    Test-wise     & 10.0 & 10.0 & 10.0 & 10.0 & 10.0 & 10.0 & 11.0 & 11.0 & 10.0 \\
    FZ-vote       & 10.0 & 10.0 & 10.0 & 10.0 & 10.0 & 10.0 & 11.0 & 11.0 & 11.0 \\
    MVPC$^\ddagger$ & 10.0 &  9.0 &  9.0 & 10.0 &  9.0 &  9.0 & 10.0 &  9.0 & 10.0 \\
    \bottomrule
\end{tabular}

\end{table}

\begin{table}[h]
    \centering
    \caption{\textbf{Robustness to missingness rate and sample size, MNAR.} Median SHD at $p = 10$ with nonlinear edges and MNAR missingness. $^\ddagger$MVPC runs in its authors' released PC implementation.}
    \label{tab:robustness_mnar}
    \begin{tabular}{lccccccccc}
    \toprule
                  & \multicolumn{3}{c}{$n{=}500$} & \multicolumn{3}{c}{$n{=}1000$} & \multicolumn{3}{c}{$n{=}2000$} \\
    \cmidrule(lr){2-4} \cmidrule(lr){5-7} \cmidrule(lr){8-10}
    Method        & 10\% & 30\% & 50\%            & 10\%  & 30\%  & 50\%           & 10\%  & 30\%  & 50\%  \\
    \midrule
    PAIR-CI       & 10.0 & 10.0 & 10.0 & 10.0 &  9.0 &  9.0 &  9.0 &  9.0 &  9.0 \\
    Complete case & 10.0 & 10.0 & 10.0 & 10.0 & 10.0 & 10.0 & 10.5 & 10.0 & 10.0 \\
    Test-wise     & 10.0 & 10.0 & 10.0 & 10.0 & 10.0 & 10.0 & 11.0 & 11.0 & 10.0 \\
    FZ-vote       & 10.0 & 10.0 & 10.0 & 11.0 & 10.0 & 11.0 & 10.5 & 11.0 & 11.0 \\
    MVPC$^\ddagger$ & 10.0 &  9.0 &  9.0 & 10.0 &  9.0 &  9.0 & 10.0 & 10.0 & 10.0 \\
    \bottomrule
\end{tabular}

\end{table}

\paragraph{F1 scores.}
The precision--recall profile is stable across all 18 conditions. PAIR-CI attains median precision of 1.000 under both MAR and MNAR, median recall of 0.500, and F1 of 0.667 (MAR) and 0.628 (MNAR). Test-wise deletion and FZ-vote obtain higher recall (0.60--0.67) but lower precision ($\approx$0.75), resulting in F1 of 0.67--0.71. Complete-case analysis records the highest F1 under MAR (0.778) but degrades under MNAR (F1 $= 0.667$). MVPC sits between PAIR-CI and the deletion and voting baselines in precision (0.83) and recall (0.60--0.64), and its F1 (0.67--0.73) is at least as high as theirs. The distinctive feature of PAIR-CI at this scale is perfect precision: conservatism manifests as missing edges rather than spurious ones. F1 advantages become more pronounced at $p \geq 20$ as the cost of miscalibration compounds (Section~\ref{sec:exp_graph}).

\begin{table}[h]
    \centering
    \caption{\textbf{Median SHD {\scriptsize(skeleton SHD)} under linear Gaussian (Panel~A) and nonlinear (Panel~B) edges} ($p = 10$, $n = 1{,}000$). Skeleton SHD counts missing and extra edges only, while total SHD additionally penalizes orientation errors. $\dagger$As MissDAG outputs a fully oriented DAG rather than a CPDAG, its total SHD (shown in gray) is not directly comparable to that of constraint-based methods; skeleton SHD provides the appropriate like-for-like metric. The ``Complete'' column is the complete-data condition, with no values masked. $^\ddagger$MVPC runs in its authors' released PC implementation.}
    \label{tab:shd_combined}
    \vspace{0.5em}
\textbf{Panel A: Linear Gaussian edges}\\[0.4em]
\begin{tabular}{lcccc}
    \toprule
    Method & Complete & MAR & MNAR & Mixed \\
    \midrule
    PAIR-CI                                           & 5\,{\scriptsize(1)} & 5.5\,{\scriptsize(1)} & 6\,{\scriptsize(1)} & 6\,{\scriptsize(1)} \\
    Complete case                                     & 6\,{\scriptsize(2)} & 6\,{\scriptsize(2)} & 7\,{\scriptsize(2)} & 6.5\,{\scriptsize(2)} \\
    Test-wise                                         & 6\,{\scriptsize(2)} & 6\,{\scriptsize(2)} & 6\,{\scriptsize(2)} & 6\,{\scriptsize(2)} \\
    FZ-vote                                           & 6\,{\scriptsize(2)} & 7\,{\scriptsize(2)} & 7\,{\scriptsize(3)} & 7\,{\scriptsize(3)} \\
    FZ-Rubin (\citeauthor{witte2022multiple})         & 6\,{\scriptsize(2)} & 6\,{\scriptsize(2)} & 7\,{\scriptsize(2)} & 7\,{\scriptsize(2)} \\
    MVPC$^\ddagger$ (\citeauthor{tu2019causal})       & 5\,{\scriptsize(1)} & 5\,{\scriptsize(1)} & 6\,{\scriptsize(1)} & 6\,{\scriptsize(1)} \\
    MissDAG$^\dagger$ (\citeauthor{gao2022missdag})   & \textcolor{lightgray}{3\,{\scriptsize(3)}}$^\dagger$ & \textcolor{lightgray}{5\,{\scriptsize(5)}}$^\dagger$ & \textcolor{lightgray}{6\,{\scriptsize(6)}}$^\dagger$ & \textcolor{lightgray}{5.5\,{\scriptsize(5)}}$^\dagger$ \\
    \bottomrule
\end{tabular}

\vspace{1em}
\textbf{Panel B: Nonlinear edges}\\[0.4em]
\begin{tabular}{lcccc}
    \toprule
    Method & Complete & MAR & MNAR & Mixed \\
    \midrule
    PAIR-CI                                           & 11\,{\scriptsize(6)} & 11\,{\scriptsize(6)} & 11\,{\scriptsize(7)} & 11\,{\scriptsize(7)} \\
    Complete case                                     & 13\,{\scriptsize(8)} & 12\,{\scriptsize(6)} & 12\,{\scriptsize(7)} & 12\,{\scriptsize(6)} \\
    Test-wise                                         & 13\,{\scriptsize(8)} & 12\,{\scriptsize(7)} & 12\,{\scriptsize(7)} & 12\,{\scriptsize(7)} \\
    FZ-vote                                           & 13\,{\scriptsize(8)} & 13\,{\scriptsize(8)} & 13\,{\scriptsize(8)} & 13\,{\scriptsize(8)} \\
    FZ-Rubin (\citeauthor{witte2022multiple})         & 13\,{\scriptsize(8)} & 12\,{\scriptsize(8)} & 12\,{\scriptsize(8)} & 12\,{\scriptsize(8)} \\
    MVPC$^\ddagger$ (\citeauthor{tu2019causal})       & 12\,{\scriptsize(7)} & 11\,{\scriptsize(7)} & 11\,{\scriptsize(6)} & 11\,{\scriptsize(7)} \\
    MissDAG$^\dagger$ (\citeauthor{gao2022missdag})   & \textcolor{lightgray}{9\,{\scriptsize(8)}}$^\dagger$ & \textcolor{lightgray}{9.5\,{\scriptsize(9)}}$^\dagger$ & \textcolor{lightgray}{10\,{\scriptsize(10)}}$^\dagger$ & \textcolor{lightgray}{10\,{\scriptsize(10)}}$^\dagger$ \\
    \bottomrule
\end{tabular}

\end{table}

Inspecting skeleton SHD makes the asymmetry explicit. PAIR-CI's skeleton SHD (1 with linear edges, 6--7 with nonlinear edges) is at or below that of complete-case analysis, test-wise deletion, FZ-vote and FZ-Rubin (2--3 linear, 6--8 nonlinear) in every cell except the mixed mechanism with nonlinear edges, where complete-case analysis reaches 6. MVPC's takes the same values (1 linear, 6--7 nonlinear) but not cell by cell: with nonlinear edges it is one unit above PAIR-CI's on complete data and under MAR and one unit below under MNAR. MissDAG's is higher (3--6 linear, 8--10 nonlinear). MissDAG's lower total SHD reflects the absence of orientation penalties in its DAG output rather than superior skeleton recovery.

\paragraph{Computational cost.} PAIR-CI's runtime scales approximately linearly with the number of CI tests in the PC skeleton search, which grows with graph size. The fast variant reduces per-test cost relative to the general variant: using the latter, $\sim$70 tests at $p=10$ complete in $\sim$3 minutes per replicate; with the former, $\sim$550 tests at $p=20$ and $\sim$1{,}000 tests at $p=30$ finish in $\sim$5 and $\sim$10 minutes, respectively. At $p=56$ in the HAILFINDER analysis (Section~\ref{sec:exp_realworld}), $\sim$2{,}400 tests take $\sim$100 minutes per replicate, with a stable per-test time of 2.5\,s. Larger graphs benefit from parallelization or more aggressive early stopping (Appendix~\ref{app:implementation}).

\paragraph{Robustness to the missingness rate and $n$.} Table~\ref{tab:robustness_mar} reports results across missingness rates of 10--50\% and sample sizes $n \in \{500, 1{,}000, 2{,}000\}$ at $p = 10$ under nonlinear edges. PAIR-CI achieves SHD~9--10 in all nine conditions, at or below test-wise deletion and FZ-vote (SHD~10--11) in every condition and strictly below both in seven, varying by at most one unit across missingness rates and sample sizes. Complete-case analysis is within one unit of PAIR-CI in SHD and records higher F1 (0.778 against 0.667). FZ-vote exhibits the clearest deterioration with increasing $n$ (SHD~10$\to$11), since larger samples raise the probability that its miscalibrated CI test detects spurious associations.

\paragraph{Full skeleton metrics.} Tables~\ref{tab:graph_full_linear_p10}--\ref{tab:graph_full_nonlinear_p30} report, for every Erd\H{o}s--R\'enyi cell and mechanism, total SHD, skeleton SHD, precision, recall, F1, and the median number of output edges for all methods including MVPC. Skeleton SHD is the number of missing plus extra adjacencies; total SHD additionally counts orientation errors on the CPDAG, comparing PC's output CPDAG with the CPDAG of the true DAG.

\begin{table}[h]
    \centering
    \caption{\textbf{Erd\H{o}s--R\'enyi $p = 10$, linear edges: full skeleton metrics by mechanism} (medians over 10 graphs $\times$ 20 datasets). $^\ddagger$MVPC runs in its authors' released PC implementation.}
    \label{tab:graph_full_linear_p10}
    \begin{tabular}{llcccccc}
    \toprule
    Mechanism & Method & SHD & Skel.\ SHD & Precision & Recall & F1 & Edges \\
    \midrule
    Complete  & P-CI               & 5 & 1 & 1.00 & 1.00 & 0.95 & 9 \\
              & CC                 & 6 & 2 & 0.86 & 1.00 & 0.91 & 11 \\
              & TW                 & 6 & 2 & 0.86 & 1.00 & 0.91 & 11 \\
              & FZ-v               & 6 & 2 & 0.86 & 1.00 & 0.91 & 11 \\
              & FZ-R               & 6 & 2 & 0.86 & 1.00 & 0.91 & 11 \\
              & MVPC$^\ddagger$    & 5 & 1 & 1.00 & 1.00 & 0.95 & 9 \\
              & MissDAG$^\dagger$  & \textcolor{lightgray}{3}$^\dagger$ & 3 & 1.00 & 0.67 & 0.80 & 6 \\
    \midrule
    MAR       & P-CI               & 5.5 & 1 & 1.00 & 1.00 & 0.95 & 9 \\
              & CC                 & 6 & 2 & 0.89 & 1.00 & 0.91 & 10 \\
              & TW                 & 6 & 2 & 0.87 & 1.00 & 0.90 & 10 \\
              & FZ-v               & 7 & 2 & 0.82 & 1.00 & 0.89 & 11 \\
              & FZ-R               & 6 & 2 & 0.83 & 1.00 & 0.90 & 11 \\
              & MVPC$^\ddagger$    & 5 & 1 & 1.00 & 1.00 & 0.95 & 9 \\
              & MissDAG$^\dagger$  & \textcolor{lightgray}{5}$^\dagger$ & 5 & 1.00 & 0.45 & 0.62 & 4 \\
    \midrule
    MNAR      & P-CI               & 6 & 1 & 1.00 & 0.91 & 0.94 & 9 \\
              & CC                 & 7 & 2 & 0.86 & 1.00 & 0.90 & 10 \\
              & TW                 & 6 & 2 & 0.86 & 1.00 & 0.90 & 10 \\
              & FZ-v               & 7 & 3 & 0.82 & 1.00 & 0.86 & 11 \\
              & FZ-R               & 7 & 2 & 0.82 & 1.00 & 0.87 & 11 \\
              & MVPC$^\ddagger$    & 6 & 1 & 1.00 & 0.92 & 0.94 & 9 \\
              & MissDAG$^\dagger$  & \textcolor{lightgray}{6}$^\dagger$ & 6 & 1.00 & 0.33 & 0.50 & 4 \\
    \midrule
    Mixed     & P-CI               & 6 & 1 & 1.00 & 0.90 & 0.94 & 9 \\
              & CC                 & 6.5 & 2 & 0.90 & 0.91 & 0.90 & 10 \\
              & TW                 & 6 & 2 & 0.86 & 1.00 & 0.90 & 10 \\
              & FZ-v               & 7 & 3 & 0.82 & 1.00 & 0.86 & 11 \\
              & FZ-R               & 7 & 2 & 0.82 & 1.00 & 0.89 & 11 \\
              & MVPC$^\ddagger$    & 6 & 1 & 1.00 & 0.91 & 0.94 & 9 \\
              & MissDAG$^\dagger$  & \textcolor{lightgray}{5.5}$^\dagger$ & 5 & 1.00 & 0.44 & 0.62 & 4 \\
    \bottomrule
\end{tabular}

\end{table}

\begin{table}[h]
    \centering
    \caption{\textbf{Erd\H{o}s--R\'enyi $p = 20$, linear edges: full skeleton metrics by mechanism} (medians over 10 graphs $\times$ 20 datasets). $^\ddagger$MVPC runs in its authors' released PC implementation.}
    \label{tab:graph_full_linear_p20}
    \begin{tabular}{llcccccc}
    \toprule
    Mechanism & Method & SHD & Skel.\ SHD & Precision & Recall & F1 & Edges \\
    \midrule
    Complete  & P-CI               & 30 & 14 & 0.88 & 0.72 & 0.79 & 31 \\
              & CC                 & 32 & 15 & 0.81 & 0.79 & 0.79 & 36 \\
              & TW                 & 32 & 15 & 0.81 & 0.79 & 0.79 & 36 \\
              & FZ-v               & 32 & 15 & 0.81 & 0.79 & 0.79 & 36 \\
              & FZ-R               & 32 & 15 & 0.81 & 0.79 & 0.79 & 36 \\
              & MVPC$^\ddagger$    & 28 & 11 & 0.94 & 0.75 & 0.83 & 30 \\
              & MissDAG$^\dagger$  & \textcolor{lightgray}{19}$^\dagger$ & 17 & 1.00 & 0.55 & 0.70 & 21 \\
    \midrule
    MAR       & P-CI               & 31 & 15 & 0.88 & 0.70 & 0.78 & 30 \\
              & CC                 & 32 & 16 & 0.85 & 0.69 & 0.76 & 31 \\
              & TW                 & 32 & 15 & 0.83 & 0.76 & 0.79 & 34 \\
              & FZ-v               & 33 & 15 & 0.81 & 0.78 & 0.80 & 36 \\
              & FZ-R               & 32 & 15 & 0.81 & 0.78 & 0.79 & 36 \\
              & MVPC$^\ddagger$    & 29 & 12 & 0.95 & 0.72 & 0.82 & 28 \\
              & MissDAG$^\dagger$  & \textcolor{lightgray}{23}$^\dagger$ & 22 & 0.95 & 0.44 & 0.60 & 17 \\
    \midrule
    MNAR      & P-CI               & 30 & 15 & 0.87 & 0.71 & 0.78 & 30 \\
              & CC                 & 35 & 19 & 0.85 & 0.58 & 0.69 & 26 \\
              & TW                 & 32 & 15 & 0.82 & 0.74 & 0.78 & 34 \\
              & FZ-v               & 33 & 16 & 0.79 & 0.78 & 0.79 & 36 \\
              & FZ-R               & 32.5 & 15 & 0.80 & 0.78 & 0.79 & 36 \\
              & MVPC$^\ddagger$    & 28 & 13 & 0.93 & 0.70 & 0.80 & 28 \\
              & MissDAG$^\dagger$  & \textcolor{lightgray}{25}$^\dagger$ & 24 & 0.94 & 0.36 & 0.53 & 15 \\
    \midrule
    Mixed     & P-CI               & 31 & 14 & 0.87 & 0.71 & 0.78 & 30 \\
              & CC                 & 34 & 18 & 0.85 & 0.63 & 0.72 & 27 \\
              & TW                 & 31.5 & 15 & 0.83 & 0.75 & 0.79 & 34 \\
              & FZ-v               & 32 & 15 & 0.81 & 0.78 & 0.79 & 36 \\
              & FZ-R               & 32.5 & 15 & 0.81 & 0.78 & 0.79 & 36 \\
              & MVPC$^\ddagger$    & 28 & 12 & 0.94 & 0.72 & 0.82 & 28 \\
              & MissDAG$^\dagger$  & \textcolor{lightgray}{24}$^\dagger$ & 23 & 0.95 & 0.40 & 0.56 & 16 \\
    \bottomrule
\end{tabular}

\end{table}

\begin{table}[h]
    \centering
    \caption{\textbf{Erd\H{o}s--R\'enyi $p = 30$, linear edges: full skeleton metrics by mechanism} (medians over 10 graphs $\times$ 20 datasets). $^\ddagger$MVPC runs in its authors' released PC implementation.}
    \label{tab:graph_full_linear_p30}
    \begin{tabular}{llcccccc}
    \toprule
    Mechanism & Method & SHD & Skel.\ SHD & Precision & Recall & F1 & Edges \\
    \midrule
    Complete  & P-CI               & 53 & 26 & 0.85 & 0.70 & 0.77 & 51 \\
              & CC                 & 60 & 28 & 0.79 & 0.77 & 0.78 & 61.5 \\
              & TW                 & 60 & 28 & 0.79 & 0.77 & 0.78 & 61.5 \\
              & FZ-v               & 60 & 28 & 0.79 & 0.77 & 0.78 & 61.5 \\
              & FZ-R               & 60 & 28 & 0.79 & 0.77 & 0.78 & 61.5 \\
              & MVPC$^\ddagger$    & 49 & 21 & 0.92 & 0.73 & 0.82 & 50 \\
              & MissDAG$^\dagger$  & \textcolor{lightgray}{27}$^\dagger$ & 25 & 0.98 & 0.60 & 0.75 & 39 \\
    \midrule
    MAR       & P-CI               & 53 & 28 & 0.86 & 0.68 & 0.76 & 50 \\
              & CC                 & 57 & 32 & 0.83 & 0.64 & 0.72 & 49 \\
              & TW                 & 58 & 29 & 0.80 & 0.73 & 0.76 & 58 \\
              & FZ-v               & 60 & 30 & 0.78 & 0.75 & 0.76 & 61 \\
              & FZ-R               & 60.5 & 30 & 0.78 & 0.75 & 0.76 & 61 \\
              & MVPC$^\ddagger$    & 48 & 22 & 0.94 & 0.70 & 0.80 & 47 \\
              & MissDAG$^\dagger$  & \textcolor{lightgray}{34}$^\dagger$ & 33 & 0.97 & 0.50 & 0.66 & 33 \\
    \midrule
    MNAR      & P-CI               & 53 & 28 & 0.86 & 0.68 & 0.76 & 50 \\
              & CC                 & 64 & 47 & 0.75 & 0.40 & 0.52 & 34 \\
              & TW                 & 58 & 30 & 0.80 & 0.71 & 0.75 & 56.5 \\
              & FZ-v               & 61 & 31 & 0.78 & 0.75 & 0.76 & 61 \\
              & FZ-R               & 61 & 30 & 0.79 & 0.75 & 0.76 & 61 \\
              & MVPC$^\ddagger$    & 49 & 24 & 0.94 & 0.67 & 0.79 & 46 \\
              & MissDAG$^\dagger$  & \textcolor{lightgray}{37}$^\dagger$ & 36 & 0.97 & 0.44 & 0.61 & 30 \\
    \midrule
    Mixed     & P-CI               & 53 & 27 & 0.87 & 0.69 & 0.77 & 51 \\
              & CC                 & 62 & 40 & 0.79 & 0.48 & 0.60 & 38 \\
              & TW                 & 57 & 28 & 0.81 & 0.73 & 0.77 & 57 \\
              & FZ-v               & 60 & 28 & 0.80 & 0.76 & 0.78 & 61 \\
              & FZ-R               & 59 & 28 & 0.80 & 0.76 & 0.78 & 61 \\
              & MVPC$^\ddagger$    & 48 & 23 & 0.95 & 0.69 & 0.79 & 46 \\
              & MissDAG$^\dagger$  & \textcolor{lightgray}{37}$^\dagger$ & 35 & 0.97 & 0.47 & 0.63 & 31 \\
    \bottomrule
\end{tabular}

\end{table}

\begin{table}[h]
    \centering
    \caption{\textbf{Erd\H{o}s--R\'enyi $p = 10$, nonlinear edges: full skeleton metrics by mechanism} (medians over 10 graphs $\times$ 20 datasets). $^\ddagger$MVPC runs in its authors' released PC implementation.}
    \label{tab:graph_full_nonlinear_p10}
    \begin{tabular}{llcccccc}
    \toprule
    Mechanism & Method & SHD & Skel.\ SHD & Precision & Recall & F1 & Edges \\
    \midrule
    Complete  & P-CI               & 11 & 6 & 1.00 & 0.53 & 0.69 & 7 \\
              & CC                 & 13 & 8 & 0.73 & 0.64 & 0.67 & 11 \\
              & TW                 & 13 & 8 & 0.73 & 0.64 & 0.67 & 11 \\
              & FZ-v               & 13 & 8 & 0.73 & 0.64 & 0.67 & 11 \\
              & FZ-R               & 13 & 8 & 0.73 & 0.64 & 0.67 & 11 \\
              & MVPC$^\ddagger$    & 12 & 7 & 0.75 & 0.60 & 0.67 & 10 \\
              & MissDAG$^\dagger$  & \textcolor{lightgray}{9}$^\dagger$ & 8 & 1.00 & 0.33 & 0.50 & 4 \\
    \midrule
    MAR       & P-CI               & 11 & 6 & 1.00 & 0.53 & 0.67 & 7 \\
              & CC                 & 12 & 6 & 0.80 & 0.70 & 0.75 & 11 \\
              & TW                 & 12 & 7 & 0.78 & 0.67 & 0.70 & 10 \\
              & FZ-v               & 13 & 8 & 0.70 & 0.67 & 0.69 & 12 \\
              & FZ-R               & 12 & 8 & 0.73 & 0.67 & 0.69 & 11 \\
              & MVPC$^\ddagger$    & 11 & 7 & 0.83 & 0.64 & 0.72 & 9 \\
              & MissDAG$^\dagger$  & \textcolor{lightgray}{9.5}$^\dagger$ & 9 & 1.00 & 0.25 & 0.40 & 3 \\
    \midrule
    MNAR      & P-CI               & 11 & 7 & 1.00 & 0.50 & 0.64 & 6 \\
              & CC                 & 12 & 7 & 0.78 & 0.60 & 0.67 & 10 \\
              & TW                 & 12 & 7 & 0.75 & 0.62 & 0.67 & 10 \\
              & FZ-v               & 13 & 8 & 0.70 & 0.64 & 0.67 & 11 \\
              & FZ-R               & 12 & 8 & 0.71 & 0.64 & 0.67 & 11 \\
              & MVPC$^\ddagger$    & 11 & 6 & 0.85 & 0.60 & 0.70 & 9 \\
              & MissDAG$^\dagger$  & \textcolor{lightgray}{10}$^\dagger$ & 10 & 1.00 & 0.20 & 0.33 & 3 \\
    \midrule
    Mixed     & P-CI               & 11 & 7 & 1.00 & 0.47 & 0.64 & 6 \\
              & CC                 & 12 & 6 & 0.82 & 0.67 & 0.74 & 10 \\
              & TW                 & 12 & 7 & 0.75 & 0.62 & 0.67 & 10 \\
              & FZ-v               & 13 & 8 & 0.68 & 0.67 & 0.67 & 11 \\
              & FZ-R               & 12 & 8 & 0.71 & 0.64 & 0.67 & 11 \\
              & MVPC$^\ddagger$    & 11 & 7 & 0.82 & 0.60 & 0.70 & 9 \\
              & MissDAG$^\dagger$  & \textcolor{lightgray}{10}$^\dagger$ & 10 & 1.00 & 0.22 & 0.36 & 3 \\
    \bottomrule
\end{tabular}

\end{table}

\begin{table}[h]
    \centering
    \caption{\textbf{Erd\H{o}s--R\'enyi $p = 20$, nonlinear edges: full skeleton metrics by mechanism} (medians over 10 graphs $\times$ 20 datasets). $^\ddagger$MVPC runs in its authors' released PC implementation.}
    \label{tab:graph_full_nonlinear_p20}
    \begin{tabular}{llcccccc}
    \toprule
    Mechanism & Method & SHD & Skel.\ SHD & Precision & Recall & F1 & Edges \\
    \midrule
    Complete  & P-CI               & 35 & 21 & 0.90 & 0.52 & 0.66 & 23 \\
              & CC                 & 43 & 28 & 0.67 & 0.60 & 0.63 & 36 \\
              & TW                 & 43 & 28 & 0.67 & 0.60 & 0.63 & 36 \\
              & FZ-v               & 43 & 28 & 0.67 & 0.60 & 0.63 & 36 \\
              & FZ-R               & 43 & 28 & 0.67 & 0.60 & 0.63 & 36 \\
              & MVPC$^\ddagger$    & 39 & 25 & 0.73 & 0.58 & 0.64 & 31.5 \\
              & MissDAG$^\dagger$  & \textcolor{lightgray}{29}$^\dagger$ & 28 & 0.92 & 0.34 & 0.48 & 15 \\
    \midrule
    MAR       & P-CI               & 35 & 22 & 0.90 & 0.49 & 0.64 & 22 \\
              & CC                 & 39 & 23 & 0.77 & 0.57 & 0.66 & 28 \\
              & TW                 & 39.5 & 25 & 0.72 & 0.60 & 0.65 & 33 \\
              & FZ-v               & 43.5 & 26 & 0.67 & 0.62 & 0.64 & 36 \\
              & FZ-R               & 41 & 26 & 0.68 & 0.61 & 0.64 & 35 \\
              & MVPC$^\ddagger$    & 36 & 22 & 0.80 & 0.57 & 0.67 & 28 \\
              & MissDAG$^\dagger$  & \textcolor{lightgray}{31}$^\dagger$ & 31 & 0.90 & 0.27 & 0.40 & 12 \\
    \midrule
    MNAR      & P-CI               & 35 & 23 & 0.89 & 0.47 & 0.62 & 21 \\
              & CC                 & 40 & 28 & 0.76 & 0.42 & 0.54 & 21 \\
              & TW                 & 39 & 25 & 0.72 & 0.56 & 0.63 & 30 \\
              & FZ-v               & 43 & 27 & 0.67 & 0.60 & 0.63 & 36 \\
              & FZ-R               & 41.5 & 27 & 0.68 & 0.59 & 0.63 & 34 \\
              & MVPC$^\ddagger$    & 35 & 23 & 0.83 & 0.52 & 0.64 & 24 \\
              & MissDAG$^\dagger$  & \textcolor{lightgray}{33}$^\dagger$ & 33 & 0.88 & 0.22 & 0.35 & 10 \\
    \midrule
    Mixed     & P-CI               & 35 & 23 & 0.88 & 0.47 & 0.62 & 21 \\
              & CC                 & 39 & 25 & 0.78 & 0.47 & 0.59 & 24 \\
              & TW                 & 40 & 25 & 0.70 & 0.56 & 0.62 & 31 \\
              & FZ-v               & 43 & 27 & 0.65 & 0.60 & 0.63 & 36 \\
              & FZ-R               & 41 & 26 & 0.68 & 0.60 & 0.63 & 34 \\
              & MVPC$^\ddagger$    & 36 & 23 & 0.81 & 0.53 & 0.64 & 26 \\
              & MissDAG$^\dagger$  & \textcolor{lightgray}{33}$^\dagger$ & 32 & 0.88 & 0.24 & 0.36 & 11 \\
    \bottomrule
\end{tabular}

\end{table}

\begin{table}[h]
    \centering
    \caption{\textbf{Erd\H{o}s--R\'enyi $p = 30$, nonlinear edges: full skeleton metrics by mechanism} (medians over 10 graphs $\times$ 20 datasets). $^\ddagger$MVPC runs in its authors' released PC implementation.}
    \label{tab:graph_full_nonlinear_p30}
    \begin{tabular}{llcccccc}
    \toprule
    Mechanism & Method & SHD & Skel.\ SHD & Precision & Recall & F1 & Edges \\
    \midrule
    Complete  & P-CI               & 56 & 36 & 0.89 & 0.50 & 0.64 & 35.5 \\
              & CC                 & 73 & 48 & 0.64 & 0.59 & 0.61 & 59 \\
              & TW                 & 73 & 48 & 0.64 & 0.59 & 0.61 & 59 \\
              & FZ-v               & 73 & 48 & 0.64 & 0.59 & 0.61 & 59 \\
              & FZ-R               & 73 & 48 & 0.64 & 0.59 & 0.61 & 59 \\
              & MVPC$^\ddagger$    & 64 & 43 & 0.71 & 0.57 & 0.63 & 52 \\
              & MissDAG$^\dagger$  & \textcolor{lightgray}{51}$^\dagger$ & 50 & 0.84 & 0.33 & 0.47 & 27 \\
    \midrule
    MAR       & P-CI               & 58 & 38 & 0.89 & 0.47 & 0.62 & 34 \\
              & CC                 & 67 & 44 & 0.74 & 0.50 & 0.59 & 43 \\
              & TW                 & 68 & 44 & 0.70 & 0.57 & 0.63 & 52 \\
              & FZ-v               & 75 & 49 & 0.64 & 0.60 & 0.62 & 60 \\
              & FZ-R               & 73 & 47 & 0.65 & 0.59 & 0.62 & 58 \\
              & MVPC$^\ddagger$    & 60 & 39 & 0.79 & 0.55 & 0.64 & 45 \\
              & MissDAG$^\dagger$  & \textcolor{lightgray}{53.5}$^\dagger$ & 53 & 0.82 & 0.27 & 0.39 & 22 \\
    \midrule
    MNAR      & P-CI               & 57 & 38 & 0.89 & 0.46 & 0.61 & 34 \\
              & CC                 & 68 & 56 & 0.64 & 0.28 & 0.38 & 27 \\
              & TW                 & 68 & 44 & 0.71 & 0.55 & 0.62 & 50 \\
              & FZ-v               & 74 & 48 & 0.63 & 0.60 & 0.62 & 60 \\
              & FZ-R               & 72 & 47 & 0.65 & 0.58 & 0.62 & 57 \\
              & MVPC$^\ddagger$    & 58 & 39 & 0.82 & 0.53 & 0.64 & 41 \\
              & MissDAG$^\dagger$  & \textcolor{lightgray}{54}$^\dagger$ & 53 & 0.83 & 0.23 & 0.35 & 18 \\
    \midrule
    Mixed     & P-CI               & 59 & 38 & 0.88 & 0.47 & 0.61 & 34 \\
              & CC                 & 68 & 51 & 0.72 & 0.35 & 0.47 & 32 \\
              & TW                 & 68 & 44 & 0.71 & 0.56 & 0.62 & 52 \\
              & FZ-v               & 75 & 48 & 0.64 & 0.60 & 0.62 & 60 \\
              & FZ-R               & 72 & 47 & 0.66 & 0.59 & 0.62 & 57.5 \\
              & MVPC$^\ddagger$    & 59 & 39 & 0.80 & 0.53 & 0.64 & 43 \\
              & MissDAG$^\dagger$  & \textcolor{lightgray}{53.5}$^\dagger$ & 52 & 0.84 & 0.25 & 0.38 & 20 \\
    \bottomrule
\end{tabular}

\end{table}

\section{Skeleton Recovery: Precision--Recall Profiles}
\label{app:pr_scatter}

Figure~\ref{fig:pr_scatter} visualizes the per-replicate precision--recall profile across all seven methods, three graph sizes, and both DGPs. Each panel pools the four missingness conditions; individual replicates are shown in low opacity, with per-method medians overlaid as large markers. PAIR-CI consistently occupies the high-precision region in every panel (precision $= 1.000$ at $p = 10$),  exhibiting lower recall than parametric baselines. MVPC has higher precision than PAIR-CI at similar recall with linear edges (0.94 against 0.86--0.87 at $p \geq 20$) and lower precision but higher recall with nonlinear edges (Table~\ref{tab:scaling}, Panel~B). MissDAG appears in the shaded ``declare nothing'' region (top-left, precision $> 0.8$, recall $< 0.3$) in all nonlinear panels, moving closer to the other methods in linear panels, where its Gaussian assumptions are satisfied.

\begin{figure}[h]
    \centering
    \includegraphics[width=\linewidth]{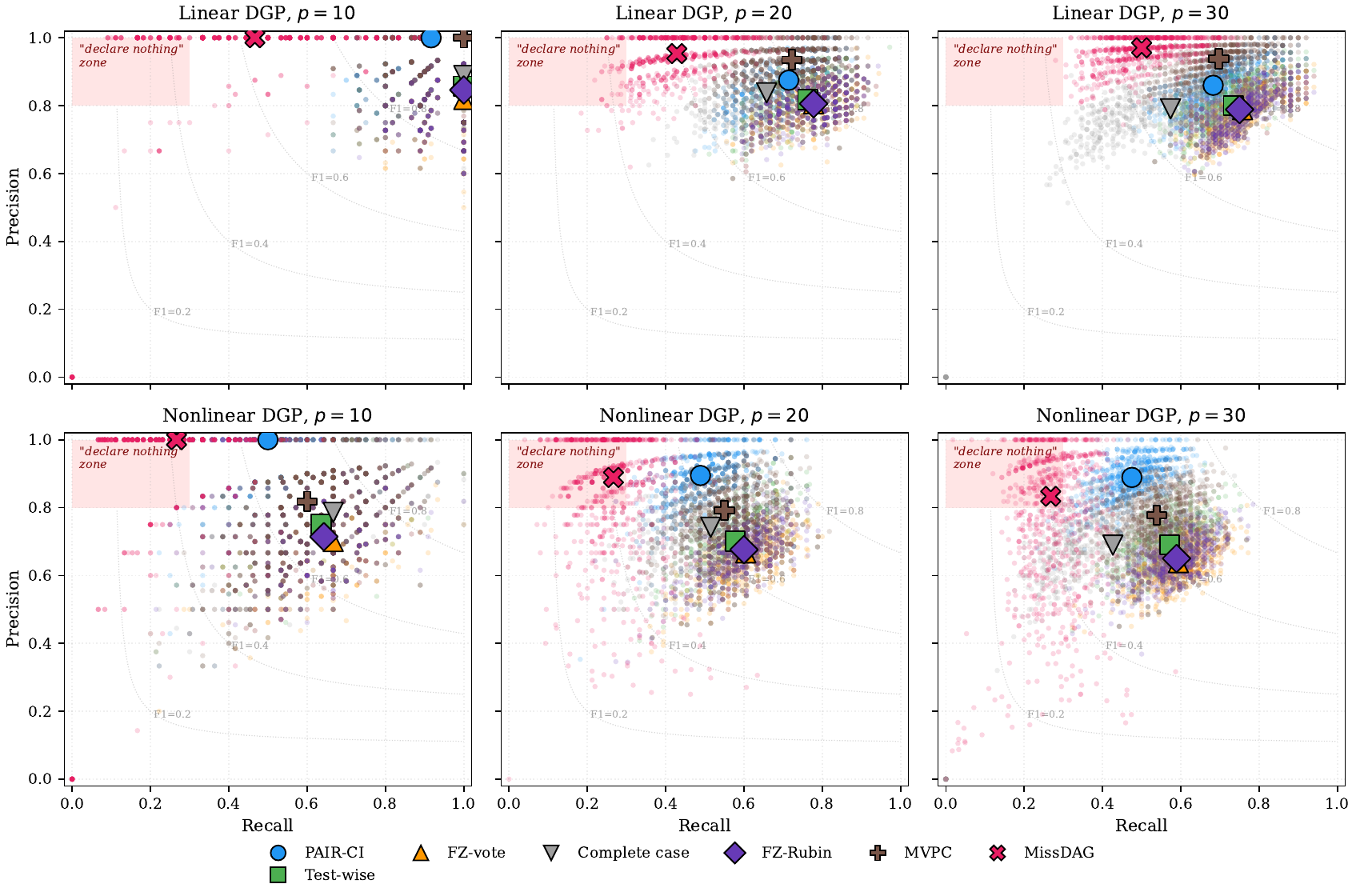}
    \caption{\textbf{Precision--recall profiles by scale and DGP.} Each panel pools all four missingness conditions. Per-replicate scatter is shown in low alpha, with per-method medians overlaid as large markers. F1 iso-curves are included for reference.}
    \label{fig:pr_scatter}
\end{figure}

\section{ALARM Network: Full Results}
\label{app:alarm}

Tables~\ref{tab:alarm_full} and~\ref{tab:alarm_full_nonlinear} report median SHD across missingness rates and mechanisms for linear Gaussian and nonlinear settings, respectively. Missingness is injected into 10 of the 25 non-root variables, while edge weights and noise follow the same specification as in the synthetic graph recovery experiments (Section~\ref{sec:exp_graph}). MVPC is included in both tables; Tables~\ref{tab:alarm_edges} and~\ref{tab:alarm_nonlinear_edges} report median output edges and Tables~\ref{tab:network_metrics_alarm_linear} and~\ref{tab:network_metrics_alarm_nonlinear} the full skeleton metrics at 20\% missingness.

\paragraph{Linear regime.}
In linear Gaussian conditions, MVPC is the strongest method (SHD 17.5--22, skeleton F1 0.92--0.95), ahead of PAIR-CI (22.5--25.5, F1 0.89--0.91), as expected when its parametric assumptions hold; test-wise deletion (SHD 28.5--34) trails PAIR-CI by 6--11.5 units. Although Fisher's $Z$ is near-optimal for Gaussian data, its advantage erodes under test-wise deletion as miscalibration compounds across the PC algorithm's many CI tests at $p = 37$. FZ-vote (30.5--38) and complete-case analysis (30--46) perform substantially worse, illustrating the costs of majority-vote MI pooling and aggressive row deletion, respectively. MVPC registers the lowest total SHD in every cell (17.5--22); MissDAG's (20.5--30) is below PAIR-CI's only on complete data and at 10\% and 20\% MAR, and MissDAG yields lower F1 than PAIR-CI (0.59--0.77 vs.\ 0.89--0.92), reflecting its conservative skeleton and the DAG--CPDAG asymmetry discussed in Section~\ref{sec:exp_graph}.

\begin{table}[h]
    \centering
    \caption{\textbf{ALARM network results with linear Gaussian edge mechanisms.} Median SHD {\scriptsize(skeleton SHD)} across missingness rates and mechanisms (20 replicates, 37 nodes, 46 edges). Skeleton SHD counts missing and extra edges only, while total SHD additionally penalizes orientation errors. Methods: P-CI = PAIR-CI (fast); CC = complete case; TW = test-wise deletion; FZ-v = FZ-vote; FZ-R = FZ-Rubin. $^\ddagger$MVPC runs in its authors' released PC implementation. $\dagger$As MissDAG outputs a fully oriented DAG rather than a CPDAG, its total SHD (shown in gray) is not directly comparable to that of constraint-based methods; skeleton SHD provides the appropriate like-for-like metric.}
    \label{tab:alarm_full}
    \begin{tabular}{llccccccc}
    \toprule
    Rate & Mechanism & P-CI & CC & TW & FZ-v & FZ-R & MVPC$^\ddagger$ & MissDAG$^\dagger$ \\
    \midrule
    0\% & Complete & 25.0\,{\scriptsize(8)} & 33.0\,{\scriptsize(14)} & 33.0\,{\scriptsize(14)} & 33.0\,{\scriptsize(14)} & 33.0\,{\scriptsize(15)} & 20.0\,{\scriptsize(5)} & \textcolor{lightgray}{20.5\,{\scriptsize(18)}}$^\dagger$ \\
    \midrule
    10\% & MAR   & 24.0\,{\scriptsize(9)} & 32.0\,{\scriptsize(14)} & 31.0\,{\scriptsize(14)} & 31.5\,{\scriptsize(16)} & 32.5\,{\scriptsize(17)} & 19.5\,{\scriptsize(6)} & \textcolor{lightgray}{22.5\,{\scriptsize(20)}}$^\dagger$ \\
         & MNAR  & 23.0\,{\scriptsize(9)} & 34.5\,{\scriptsize(15)} & 31.0\,{\scriptsize(15)} & 32.5\,{\scriptsize(15)} & 31.0\,{\scriptsize(15)} & 21.0\,{\scriptsize(7)} & \textcolor{lightgray}{24.0\,{\scriptsize(21)}}$^\dagger$ \\
         & Mixed & 22.5\,{\scriptsize(8)} & 30.0\,{\scriptsize(13)} & 28.5\,{\scriptsize(13)} & 32.0\,{\scriptsize(14)} & 30.5\,{\scriptsize(14)} & 17.5\,{\scriptsize(5)} & \textcolor{lightgray}{24.5\,{\scriptsize(22)}}$^\dagger$ \\
    \midrule
    20\% & MAR   & 24.0\,{\scriptsize(9)} & 34.0\,{\scriptsize(14)} & 33.0\,{\scriptsize(14)} & 30.5\,{\scriptsize(15)} & 30.5\,{\scriptsize(15)} & 20.0\,{\scriptsize(6)} & \textcolor{lightgray}{21.5\,{\scriptsize(19)}}$^\dagger$ \\
         & MNAR  & 22.5\,{\scriptsize(9)} & 40.0\,{\scriptsize(20)} & 34.0\,{\scriptsize(15)} & 36.0\,{\scriptsize(17)} & 35.0\,{\scriptsize(17)} & 20.0\,{\scriptsize(7)} & \textcolor{lightgray}{27.0\,{\scriptsize(25)}}$^\dagger$ \\
         & Mixed & 24.0\,{\scriptsize(8)} & 38.5\,{\scriptsize(17)} & 33.0\,{\scriptsize(13)} & 35.5\,{\scriptsize(15)} & 34.0\,{\scriptsize(14)} & 18.5\,{\scriptsize(5)} & \textcolor{lightgray}{27.5\,{\scriptsize(25)}}$^\dagger$ \\
    \midrule
    40\% & MAR   & 25.5\,{\scriptsize(10)} & 35.5\,{\scriptsize(16)} & 32.5\,{\scriptsize(15)} & 38.0\,{\scriptsize(19)} & 35.0\,{\scriptsize(17)} & 20.0\,{\scriptsize(6)} & \textcolor{lightgray}{27.0\,{\scriptsize(26)}}$^\dagger$ \\
         & MNAR  & 25.5\,{\scriptsize(9)} & \textcolor{lightgray}{46.0\,{\scriptsize(46)}}$^\dagger$ & 33.0\,{\scriptsize(14)} & 36.5\,{\scriptsize(18)} & 35.0\,{\scriptsize(17)} & 22.0\,{\scriptsize(7)} & \textcolor{lightgray}{30.0\,{\scriptsize(28)}}$^\dagger$ \\
         & Mixed & 24.0\,{\scriptsize(10)} & \textcolor{lightgray}{46.0\,{\scriptsize(46)}}$^\dagger$ & 32.5\,{\scriptsize(14)} & 36.0\,{\scriptsize(18)} & 34.5\,{\scriptsize(17)} & 18.5\,{\scriptsize(6)} & \textcolor{lightgray}{29.0\,{\scriptsize(27)}}$^\dagger$ \\
    \bottomrule
\end{tabular}

\end{table}

\paragraph{Nonlinear regime.}
In nonlinear settings, PAIR-CI's advantage widens substantially: test-wise deletion deteriorates to SHD 48--56.5, while PAIR-CI remains at 30--35.5, opening a gap of 13.5--25.5 SHD units (compared to 6--11.5 under linear edges). FZ-vote (SHD 53.5--62) and complete-case analysis (SHD 46--56.5) are characterized by the same pattern at a larger scale. MVPC (SHD 38.5--45, F1 0.64--0.69) is the strongest parametric baseline but remains 3.5--12.5 SHD behind PAIR-CI under missingness (7--10.5 at 20\%) as its Fisher-$Z$ core's precision falls (Section~\ref{sec:exp_realworld}). MissDAG comes closest to PAIR-CI on total SHD (34--39.5) except at 40\% MNAR and mixed, where MVPC is closer, yet yields substantially lower F1 (0.32--0.51 vs.\ 0.73--0.82 for PAIR-CI). In one replicate of the rate-0.1 mixed condition, complete-case Fisher-$Z$ fails outright due to a singular correlation submatrix, demonstrating a brittleness that random-forest-based testing avoids by construction.

\begin{table}[h]
    \centering
    \caption{\textbf{ALARM network results with nonlinear edge mechanisms.} Median SHD {\scriptsize(skeleton SHD)} across missingness rates and mechanisms (20 replicates, 37 nodes, 46 edges). Skeleton SHD counts missing and extra edges only, while total SHD additionally penalizes orientation errors. Methods: P-CI = PAIR-CI (fast); CC = complete case; TW = test-wise deletion; FZ-v = FZ-vote; FZ-R = FZ-Rubin. $^\ddagger$MVPC runs in its authors' released PC implementation. $\dagger$As MissDAG outputs a fully oriented DAG rather than a CPDAG, its total SHD (shown in gray) is not directly comparable to that of constraint-based methods; skeleton SHD provides the appropriate like-for-like metric.}
    \label{tab:alarm_full_nonlinear}
    \begin{tabular}{llccccccc}
    \toprule
    Rate & Mechanism & P-CI & CC & TW & FZ-v & FZ-R & MVPC$^\ddagger$ & MissDAG$^\dagger$ \\
    \midrule
    0\% & Complete & 31.0\,{\scriptsize(14)} & 56.5\,{\scriptsize(41)} & 56.5\,{\scriptsize(41)} & 56.5\,{\scriptsize(41)} & 56.5\,{\scriptsize(41)} & 49.0\,{\scriptsize(35)} & \textcolor{lightgray}{33.5\,{\scriptsize(31)}}$^\dagger$ \\
    \midrule
    10\% & MAR   & 31.0\,{\scriptsize(16)} & 50.0\,{\scriptsize(36)} & 51.5\,{\scriptsize(37)} & 54.0\,{\scriptsize(40)} & 54.0\,{\scriptsize(41)} & 42.5\,{\scriptsize(30)} & \textcolor{lightgray}{35.0\,{\scriptsize(34)}}$^\dagger$ \\
         & MNAR  & 32.5\,{\scriptsize(17)} & 51.5\,{\scriptsize(35)} & 52.0\,{\scriptsize(39)} & 58.0\,{\scriptsize(44)} & 56.0\,{\scriptsize(43)} & 45.0\,{\scriptsize(31)} & \textcolor{lightgray}{37.5\,{\scriptsize(36)}}$^\dagger$ \\
         & Mixed & 30.0\,{\scriptsize(15)} & 51.0\,{\scriptsize(35)} & 50.0\,{\scriptsize(36)} & 53.5\,{\scriptsize(39)} & 54.0\,{\scriptsize(40)} & 42.0\,{\scriptsize(27)} & \textcolor{lightgray}{35.5\,{\scriptsize(35)}}$^\dagger$ \\
    \midrule
    20\% & MAR   & 31.5\,{\scriptsize(16)} & 51.0\,{\scriptsize(34)} & 52.0\,{\scriptsize(37)} & 57.5\,{\scriptsize(43)} & 55.5\,{\scriptsize(40)} & 42.0\,{\scriptsize(30)} & \textcolor{lightgray}{34.0\,{\scriptsize(33)}}$^\dagger$ \\
         & MNAR  & 33.0\,{\scriptsize(16)} & 51.0\,{\scriptsize(40)} & 48.0\,{\scriptsize(36)} & 55.5\,{\scriptsize(41)} & 53.0\,{\scriptsize(41)} & 40.0\,{\scriptsize(29)} & \textcolor{lightgray}{37.0\,{\scriptsize(36)}}$^\dagger$ \\
         & Mixed & 33.5\,{\scriptsize(17)} & 48.0\,{\scriptsize(36)} & 52.0\,{\scriptsize(35)} & 56.0\,{\scriptsize(41)} & 54.0\,{\scriptsize(39)} & 40.5\,{\scriptsize(27)} & \textcolor{lightgray}{35.0\,{\scriptsize(35)}}$^\dagger$ \\
    \midrule
    40\% & MAR   & 35.5\,{\scriptsize(20)} & 51.5\,{\scriptsize(39)} & 51.0\,{\scriptsize(34)} & 57.5\,{\scriptsize(45)} & 56.0\,{\scriptsize(39)} & 40.5\,{\scriptsize(28)} & \textcolor{lightgray}{39.5\,{\scriptsize(37)}}$^\dagger$ \\
         & MNAR  & 35.5\,{\scriptsize(21)} & \textcolor{lightgray}{46.0\,{\scriptsize(46)}}$^\dagger$ & 52.0\,{\scriptsize(37)} & 62.0\,{\scriptsize(47)} & 58.5\,{\scriptsize(43)} & 39.0\,{\scriptsize(28)} & \textcolor{lightgray}{39.5\,{\scriptsize(39)}}$^\dagger$ \\
         & Mixed & 35.0\,{\scriptsize(21)} & \textcolor{lightgray}{46.0\,{\scriptsize(46)}}$^\dagger$ & 48.5\,{\scriptsize(38)} & 60.0\,{\scriptsize(47)} & 55.0\,{\scriptsize(41)} & 38.5\,{\scriptsize(29)} & \textcolor{lightgray}{39.0\,{\scriptsize(38)}}$^\dagger$ \\
    \bottomrule
\end{tabular}

\end{table}

\begin{table}[h]
    \centering
    \caption{\textbf{ALARM: median output edges by rate and mechanism} (linear edges; 46 true edges). $^\ddagger$MVPC runs in its authors' released PC implementation.}
    \label{tab:alarm_edges}
    \begin{tabular}{llccccccc}
    \toprule
    Rate & Mechanism & P-CI & CC & TW & FZ-v & FZ-R & MVPC$^\ddagger$ & MissDAG$^\dagger$ \\
    \midrule
    0\% & Complete & 46.0 & 55 & 55 & 55 & 55 & 45 & \textcolor{lightgray}{31}$^\dagger$ \\
    \midrule
    10\% & MAR   & 45 & 53 & 53 & 55 & 55 & 44 & \textcolor{lightgray}{29}$^\dagger$ \\
         & MNAR  & 45 & 51.5 & 54 & 54 & 54 & 43 & \textcolor{lightgray}{26.5}$^\dagger$ \\
         & Mixed & 45 & 50 & 52.5 & 54 & 54 & 44 & \textcolor{lightgray}{29}$^\dagger$ \\
    \midrule
    20\% & MAR   & 45.5 & 51.5 & 54 & 56.0 & 55 & 44 & \textcolor{lightgray}{30}$^\dagger$ \\
         & MNAR  & 44.5 & 45.5 & 52.5 & 56 & 55 & 44 & \textcolor{lightgray}{25}$^\dagger$ \\
         & Mixed & 45.5 & 48 & 54 & 56 & 56 & 44 & \textcolor{lightgray}{26}$^\dagger$ \\
    \midrule
    40\% & MAR   & 45 & 48.5 & 52 & 59.5 & 55.5 & 42 & \textcolor{lightgray}{24}$^\dagger$ \\
         & MNAR  & 45.0 & \textcolor{lightgray}{0}$^\dagger$ & 51.5 & 59 & 56.5 & 41 & \textcolor{lightgray}{20}$^\dagger$ \\
         & Mixed & 45.0 & \textcolor{lightgray}{0}$^\dagger$ & 52.5 & 58 & 56 & 42 & \textcolor{lightgray}{24}$^\dagger$ \\
    \bottomrule
\end{tabular}

\end{table}
\begin{table}[h]
    \centering
    \caption{\textbf{ALARM, nonlinear edges: median output edges by rate and mechanism} (46 true edges). $^\ddagger$MVPC runs in its authors' released PC implementation.}
    \label{tab:alarm_nonlinear_edges}
    \begin{tabular}{llccccccc}
    \toprule
    Rate & Mechanism & P-CI & CC & TW & FZ-v & FZ-R & MVPC$^\ddagger$ & MissDAG$^\dagger$ \\
    \midrule
    0\% & Complete & 34.5 & 54.5 & 54.5 & 54.5 & 54.5 & 45 & \textcolor{lightgray}{19}$^\dagger$ \\
    \midrule
    10\% & MAR   & 33.5 & 49.0 & 50.5 & 55.5 & 54 & 41 & \textcolor{lightgray}{17}$^\dagger$ \\
         & MNAR  & 32.5 & 43 & 50 & 55.5 & 54 & 39.5 & \textcolor{lightgray}{14}$^\dagger$ \\
         & Mixed & 34 & 43.5 & 49 & 55.5 & 54 & 40 & \textcolor{lightgray}{17.5}$^\dagger$ \\
    \midrule
    20\% & MAR   & 33 & 46.5 & 50.5 & 57 & 54.5 & 38.5 & \textcolor{lightgray}{16}$^\dagger$ \\
         & MNAR  & 33.5 & 35 & 47.5 & 56 & 51 & 36.5 & \textcolor{lightgray}{12}$^\dagger$ \\
         & Mixed & 35 & 39.5 & 48 & 55.5 & 54 & 39 & \textcolor{lightgray}{16}$^\dagger$ \\
    \midrule
    40\% & MAR   & 32 & 41 & 47 & 57 & 54 & 38 & \textcolor{lightgray}{15.5}$^\dagger$ \\
         & MNAR  & 31.5 & \textcolor{lightgray}{0}$^\dagger$ & 46.5 & 60.5 & 54.5 & 33 & \textcolor{lightgray}{11}$^\dagger$ \\
         & Mixed & 32 & \textcolor{lightgray}{0}$^\dagger$ & 47.5 & 57.0 & 53 & 35 & \textcolor{lightgray}{14}$^\dagger$ \\
    \bottomrule
\end{tabular}

\end{table}
\begin{table}[h]
    \centering
    \caption{\textbf{ALARM at 20\% missingness: full skeleton metrics by mechanism} (linear edges). Reps: replicates in the median; DNF: replicates that did not finish within the 1-hour cap. $^\ddagger$MVPC runs in its authors' released PC implementation.}
    \label{tab:network_metrics_alarm_linear}
    \resizebox{\textwidth}{!}{\begin{tabular}{llccccccc}
    \toprule
    Mechanism & Method & SHD & Skel.\ SHD & Precision & Recall & F1 & Edges & Reps \\
    \midrule
    Complete (0\%)   & P-CI               & 25 & 8 & 0.91 & 0.91 & 0.92 & 46.0 & 20 \\
                     & CC                 & 33 & 14 & 0.78 & 0.95 & 0.86 & 55 & 20 \\
                     & TW                 & 33 & 14 & 0.78 & 0.95 & 0.86 & 55 & 20 \\
                     & FZ-v               & 33 & 14 & 0.78 & 0.95 & 0.86 & 55 & 20 \\
                     & FZ-R               & 33 & 15 & 0.78 & 0.95 & 0.86 & 55 & 20 \\
                     & MVPC$^\ddagger$    & 20 & 5 & 0.96 & 0.93 & 0.95 & 45 & 20 \\
                     & MissDAG$^\dagger$  & \textcolor{lightgray}{20.5}$^\dagger$ & 18 & 0.97 & 0.65 & 0.77 & 31 & 20 \\
    \midrule
    MAR (20\%)       & P-CI               & 24 & 9 & 0.91 & 0.89 & 0.90 & 45.5 & 20 \\
                     & CC                 & 34 & 14 & 0.81 & 0.89 & 0.85 & 51.5 & 20 \\
                     & TW                 & 33 & 14 & 0.79 & 0.93 & 0.86 & 54 & 20 \\
                     & FZ-v               & 30.5 & 15 & 0.77 & 0.93 & 0.85 & 56.0 & 20 \\
                     & FZ-R               & 30.5 & 15 & 0.79 & 0.93 & 0.85 & 55 & 20 \\
                     & MVPC$^\ddagger$    & 20 & 6 & 0.95 & 0.91 & 0.93 & 44 & 20 \\
                     & MissDAG$^\dagger$  & \textcolor{lightgray}{21.5}$^\dagger$ & 19 & 0.96 & 0.62 & 0.76 & 30 & 20 \\
    \midrule
    MNAR (20\%)      & P-CI               & 22.5 & 9 & 0.91 & 0.88 & 0.89 & 44.5 & 20 \\
                     & CC                 & 40 & 20 & 0.78 & 0.76 & 0.78 & 45.5 & 20 \\
                     & TW                 & 34 & 15 & 0.80 & 0.91 & 0.84 & 52.5 & 20 \\
                     & FZ-v               & 36 & 17 & 0.76 & 0.93 & 0.84 & 56 & 20 \\
                     & FZ-R               & 35 & 17 & 0.77 & 0.93 & 0.84 & 55 & 20 \\
                     & MVPC$^\ddagger$    & 20 & 7 & 0.95 & 0.91 & 0.93 & 44 & 20 \\
                     & MissDAG$^\dagger$  & \textcolor{lightgray}{27}$^\dagger$ & 25 & 0.95 & 0.48 & 0.64 & 25 & 20 \\
    \midrule
    Mixed (20\%)     & P-CI               & 24 & 8 & 0.91 & 0.91 & 0.91 & 45.5 & 20 \\
                     & CC                 & 38.5 & 17 & 0.79 & 0.85 & 0.82 & 48 & 20 \\
                     & TW                 & 33 & 13 & 0.81 & 0.93 & 0.87 & 54 & 20 \\
                     & FZ-v               & 35.5 & 15 & 0.77 & 0.93 & 0.85 & 56 & 20 \\
                     & FZ-R               & 34 & 14 & 0.79 & 0.93 & 0.86 & 56 & 20 \\
                     & MVPC$^\ddagger$    & 18.5 & 5 & 0.98 & 0.91 & 0.95 & 44 & 20 \\
                     & MissDAG$^\dagger$  & \textcolor{lightgray}{27.5}$^\dagger$ & 25 & 0.94 & 0.50 & 0.64 & 26 & 20 \\
    \bottomrule
\end{tabular}
}
\end{table}
\begin{table}[h]
    \centering
    \caption{\textbf{ALARM at 20\% missingness: full skeleton metrics by mechanism} (nonlinear edges). $^\ddagger$MVPC runs in its authors' released PC implementation.}
    \label{tab:network_metrics_alarm_nonlinear}
    \resizebox{\textwidth}{!}{\begin{tabular}{llccccccc}
    \toprule
    Mechanism & Method & SHD & Skel.\ SHD & Precision & Recall & F1 & Edges & Reps \\
    \midrule
    Complete (0\%)   & P-CI               & 31 & 14 & 0.96 & 0.72 & 0.82 & 34.5 & 20 \\
                     & CC                 & 56.5 & 41 & 0.55 & 0.64 & 0.60 & 54.5 & 20 \\
                     & TW                 & 56.5 & 41 & 0.55 & 0.64 & 0.60 & 54.5 & 20 \\
                     & FZ-v               & 56.5 & 41 & 0.55 & 0.64 & 0.60 & 54.5 & 20 \\
                     & FZ-R               & 56.5 & 41 & 0.55 & 0.64 & 0.60 & 54.5 & 20 \\
                     & MVPC$^\ddagger$    & 49 & 35 & 0.62 & 0.62 & 0.63 & 45 & 20 \\
                     & MissDAG$^\dagger$  & \textcolor{lightgray}{33.5}$^\dagger$ & 31 & 0.90 & 0.35 & 0.51 & 19 & 20 \\
    \midrule
    MAR (20\%)       & P-CI               & 31.5 & 16 & 0.94 & 0.70 & 0.80 & 33 & 20 \\
                     & CC                 & 51 & 34 & 0.63 & 0.64 & 0.65 & 46.5 & 20 \\
                     & TW                 & 52 & 37 & 0.59 & 0.62 & 0.61 & 50.5 & 20 \\
                     & FZ-v               & 57.5 & 43 & 0.53 & 0.64 & 0.57 & 57 & 20 \\
                     & FZ-R               & 55.5 & 40 & 0.56 & 0.63 & 0.59 & 54.5 & 20 \\
                     & MVPC$^\ddagger$    & 42 & 30 & 0.71 & 0.62 & 0.66 & 38.5 & 20 \\
                     & MissDAG$^\dagger$  & \textcolor{lightgray}{34}$^\dagger$ & 33 & 0.94 & 0.32 & 0.45 & 16 & 20 \\
    \midrule
    MNAR (20\%)      & P-CI               & 33 & 16 & 0.94 & 0.67 & 0.80 & 33.5 & 20 \\
                     & CC                 & 51 & 40 & 0.59 & 0.43 & 0.49 & 35 & 20 \\
                     & TW                 & 48 & 36 & 0.61 & 0.63 & 0.62 & 47.5 & 20 \\
                     & FZ-v               & 55.5 & 41 & 0.54 & 0.63 & 0.58 & 56 & 20 \\
                     & FZ-R               & 53 & 41 & 0.55 & 0.62 & 0.58 & 51 & 20 \\
                     & MVPC$^\ddagger$    & 40 & 29 & 0.75 & 0.59 & 0.65 & 36.5 & 20 \\
                     & MissDAG$^\dagger$  & \textcolor{lightgray}{37}$^\dagger$ & 36 & 0.92 & 0.24 & 0.38 & 12 & 20 \\
    \midrule
    Mixed (20\%)     & P-CI               & 33.5 & 17 & 0.93 & 0.68 & 0.79 & 35 & 20 \\
                     & CC                 & 48 & 36 & 0.63 & 0.54 & 0.59 & 39.5 & 20 \\
                     & TW                 & 52 & 35 & 0.61 & 0.65 & 0.63 & 48 & 20 \\
                     & FZ-v               & 56 & 41 & 0.55 & 0.67 & 0.59 & 55.5 & 20 \\
                     & FZ-R               & 54 & 39 & 0.57 & 0.67 & 0.62 & 54 & 20 \\
                     & MVPC$^\ddagger$    & 40.5 & 27 & 0.76 & 0.64 & 0.69 & 39 & 20 \\
                     & MissDAG$^\dagger$  & \textcolor{lightgray}{35}$^\dagger$ & 35 & 0.88 & 0.26 & 0.41 & 16 & 20 \\
    \bottomrule
\end{tabular}
}
\end{table}

\section{HAILFINDER Network: Full Results}
\label{app:hailfinder}

\citet{witte2022multiple} assume MAR for the multiple-imputation route and report that test-wise deletion outperforms it in their MNAR scenarios.

At 20\% missingness, PAIR-CI's skeleton F1 of 0.59 compares with 0.45--0.49 for the parametric pipelines and 0.57 for MVPC at 63 output edges, and complete-case analysis attains SHD~66 under MNAR, within 2 units of PAIR-CI's 64, by recovering an empty skeleton (F1 exactly 0.000). At 40\% missingness MVPC and PAIR-CI converge (F1 0.57 at SHD 69 against 0.56 at SHD 66), and PAIR-CI maintains F1 $\geq 0.56$ under every rate and mechanism.

Tables~\ref{tab:hailfinder_full} and~\ref{tab:hailfinder_f1} show median SHD and F1 across all missingness rates and mechanisms (20 replicates per cell, nonlinear edges, missingness induced in 15 non-root variables), with MVPC (Section~\ref{sec:exp_graph}) included; Tables~\ref{tab:hailfinder_edges} and~\ref{tab:network_metrics_hailfinder} report median output edges across the grid and the full skeleton metrics at 20\% missingness.

\paragraph{Gap acceleration.}
The SHD gap over the best baseline grows superlinearly with graph size: at $p = 56$ and 20\% missingness, PAIR-CI reaches SHD 62.5--65.5, while test-wise deletion records 103.5--112.5 and FZ-vote exceeds 120 in most cells. MVPC comes closest (SHD 69--88 across rates and mechanisms; 72--83 at 20\%); test-wise deletion and the Fisher-$Z$ pipelines trail PAIR-CI by more than 30 SHD in every cell, and complete-case analysis is discussed below. MissDAG was not evaluated on HAILFINDER due to computational cost.

\paragraph{``Winning by giving up'' pattern.}
Complete-case analysis provides the clearest instance of this phenomenon. At 20\% MNAR missingess, it yields SHD 66---close to PAIR-CI's 64---but with F1 of exactly 0, as the surviving observations support no edge recovery. The same degenerate outcome occurs at 40\% mixed and 40\% MNAR missingness. F1 thus reveals that SHD alone cannot distinguish between recovering the graph and declaring nothing (Appendix~\ref{app:pr_scatter}).

\paragraph{Fisher--$Z$ brittleness.}
FZ-vote fails on one replicate at 40\% MNAR missingness due to a singular correlation matrix, providing a further illustration of the brittleness noted in Appendix~\ref{app:alarm}.

\begin{table}[h]
    \centering
    \caption{\textbf{HAILFINDER network results (SHD) with nonlinear edge mechanisms.} Median SHD across missingness rates and mechanisms (20 replicates, 56 nodes, 66 edges). Methods: P-CI = PAIR-CI (fast); CC = complete case; TW = test-wise deletion; FZ-v = FZ-vote; FZ-R = FZ-Rubin. $^\ddagger$MVPC runs in its authors' released PC implementation. $\dagger$Complete-case cells shown in gray have median F1 = 0.000 (empty recovered skeleton; Table~\ref{tab:hailfinder_f1}); their SHD reflects degenerate output rather than competitive recovery.}
    \label{tab:hailfinder_full}
    \begin{tabular}{llcccccc}
    \toprule
    Rate & Mechanism & P-CI & CC & TW & FZ-v & FZ-R & MVPC$^\ddagger$ \\
    \midrule
    0\% & Complete & 63.0\,{\scriptsize(46)} & 125.0\,{\scriptsize(95)} & 125.0\,{\scriptsize(95)} & 125.0\,{\scriptsize(95)} & 125.0\,{\scriptsize(95)} & 107.5\,{\scriptsize(79)} \\
    \midrule
    10\% & MAR   & 64.0\,{\scriptsize(45)} & 99.5\,{\scriptsize(75)} & 114.0\,{\scriptsize(88)} & 123.0\,{\scriptsize(96)} & 118.0\,{\scriptsize(91)} & 88.0\,{\scriptsize(67)} \\
         & MNAR  & 63.0\,{\scriptsize(44)} & 89.5\,{\scriptsize(73)} & 107.0\,{\scriptsize(80)} & 121.5\,{\scriptsize(90)} & 119.0\,{\scriptsize(88)} & 80.5\,{\scriptsize(59)} \\
         & Mixed & 64.0\,{\scriptsize(46)} & 89.5\,{\scriptsize(69)} & 108.5\,{\scriptsize(84)} & 123.0\,{\scriptsize(93)} & 122.5\,{\scriptsize(94)} & 88.0\,{\scriptsize(63)} \\
    \midrule
    20\% & MAR   & 62.5\,{\scriptsize(46)} & 98.0\,{\scriptsize(77)} & 112.5\,{\scriptsize(84)} & 128.0\,{\scriptsize(97)} & 121.0\,{\scriptsize(96)} & 83.0\,{\scriptsize(63)} \\
         & MNAR  & 64.0\,{\scriptsize(45)} & \textcolor{lightgray}{66.0\,{\scriptsize(66)}}$^\dagger$ & 103.5\,{\scriptsize(78)} & 122.0\,{\scriptsize(94)} & 120.0\,{\scriptsize(90)} & 72.0\,{\scriptsize(54)} \\
         & Mixed & 65.5\,{\scriptsize(47)} & 89.5\,{\scriptsize(75)} & 104.5\,{\scriptsize(80)} & 124.5\,{\scriptsize(92)} & 120.0\,{\scriptsize(92)} & 77.5\,{\scriptsize(58)} \\
    \midrule
    40\% & MAR   & 66.0\,{\scriptsize(46)} & 94.0\,{\scriptsize(77)} & 103.5\,{\scriptsize(79)} & 125.5\,{\scriptsize(97)} & 118.5\,{\scriptsize(94)} & 78.5\,{\scriptsize(59)} \\
         & MNAR  & 66.0\,{\scriptsize(49)} & \textcolor{lightgray}{66.0\,{\scriptsize(66)}}$^\dagger$ & 97.5\,{\scriptsize(74)} & 130.5\,{\scriptsize(100)} & 116.5\,{\scriptsize(90)} & 69.0\,{\scriptsize(50)} \\
         & Mixed & 65.5\,{\scriptsize(48)} & \textcolor{lightgray}{66.0\,{\scriptsize(66)}}$^\dagger$ & 102.0\,{\scriptsize(77)} & 124.0\,{\scriptsize(98)} & 120.0\,{\scriptsize(92)} & 73.0\,{\scriptsize(55)} \\
    \bottomrule
\end{tabular}

\end{table}

\begin{table}[h]
    \centering
    \caption{\textbf{HAILFINDER network results (F1) with nonlinear edge mechanisms.} Median F1 across missingness rates and mechanisms (20 replicates, 56 nodes, 66 edges). Methods: P-CI = PAIR-CI (fast); CC = complete case; TW = test-wise deletion; FZ-v = FZ-vote; FZ-R = FZ-Rubin. $^\ddagger$MVPC runs in its authors' released PC implementation. $\dagger$Complete-case cells shown in gray have median F1 = 0.000, indicating an empty recovered skeleton.}
    \label{tab:hailfinder_f1}
    \begin{tabular}{llcccccc}
    \toprule
    Rate & Mechanism & P-CI & CC & TW & FZ-v & FZ-R & MVPC$^\ddagger$ \\
    \midrule
    0\% & Complete & 0.586 & 0.467 & 0.467 & 0.467 & 0.467 & 0.508 \\
    \midrule
    10\% & MAR   & 0.602 & 0.498 & 0.479 & 0.457 & 0.468 & 0.537 \\
         & MNAR  & 0.585 & 0.425 & 0.494 & 0.479 & 0.491 & 0.563 \\
         & Mixed & 0.584 & 0.489 & 0.503 & 0.478 & 0.460 & 0.544 \\
    \midrule
    20\% & MAR   & 0.590 & 0.462 & 0.484 & 0.450 & 0.460 & 0.553 \\
         & MNAR  & 0.594 & \textcolor{lightgray}{0.000}$^\dagger$ & 0.490 & 0.468 & 0.481 & 0.571 \\
         & Mixed & 0.578 & 0.296 & 0.492 & 0.481 & 0.471 & 0.557 \\
    \midrule
    40\% & MAR   & 0.573 & 0.402 & 0.489 & 0.474 & 0.451 & 0.543 \\
         & MNAR  & 0.560 & \textcolor{lightgray}{0.000}$^\dagger$ & 0.505 & 0.445 & 0.453 & 0.574 \\
         & Mixed & 0.571 & \textcolor{lightgray}{0.000}$^\dagger$ & 0.485 & 0.458 & 0.468 & 0.552 \\
    \bottomrule
\end{tabular}

\end{table}

\begin{table}[h]
    \centering
    \caption{\textbf{HAILFINDER: median output edges by rate and mechanism} (66 true edges). $^\ddagger$MVPC runs in its authors' released PC implementation.}
    \label{tab:hailfinder_edges}
    \begin{tabular}{llcccccc}
    \toprule
    Rate & Mechanism & P-CI & CC & TW & FZ-v & FZ-R & MVPC$^\ddagger$ \\
    \midrule
    0\% & Complete & 44 & 115.5 & 115.5 & 115.5 & 115.5 & 95.5 \\
    \midrule
    10\% & MAR   & 46 & 82.5 & 100.5 & 111 & 107 & 78 \\
         & MNAR  & 40.5 & 60.5 & 92 & 111.0 & 109 & 68 \\
         & Mixed & 45 & 66.5 & 97.5 & 109.5 & 109.5 & 73.5 \\
    \midrule
    20\% & MAR   & 42 & 76 & 96 & 112.0 & 107.5 & 70.5 \\
         & MNAR  & 42.0 & \textcolor{lightgray}{0}$^\dagger$ & 88.5 & 108.5 & 107.5 & 63 \\
         & Mixed & 43 & 49 & 91.0 & 112 & 107.5 & 65.5 \\
    \midrule
    40\% & MAR   & 42.5 & 62.5 & 89 & 114.5 & 107 & 64.5 \\
         & MNAR  & 43.5 & \textcolor{lightgray}{0}$^\dagger$ & 81.5 & 114.5 & 100.5 & 55 \\
         & Mixed & 42.5 & \textcolor{lightgray}{0}$^\dagger$ & 83.5 & 109.5 & 105 & 58 \\
    \bottomrule
\end{tabular}

\end{table}
\begin{table}[h]
    \centering
    \caption{\textbf{HAILFINDER at 20\% missingness: full skeleton metrics by mechanism.} $^\ddagger$MVPC runs in its authors' released PC implementation.}
    \label{tab:network_metrics_hailfinder}
    \resizebox{\textwidth}{!}{\begin{tabular}{llccccccc}
    \toprule
    Mechanism & Method & SHD & Skel.\ SHD & Precision & Recall & F1 & Edges & Reps \\
    \midrule
    Complete (0\%)   & P-CI               & 63 & 46 & 0.72 & 0.48 & 0.59 & 44 & 20 \\
                     & CC                 & 125 & 95 & 0.37 & 0.64 & 0.47 & 115.5 & 20 \\
                     & TW                 & 125 & 95 & 0.37 & 0.64 & 0.47 & 115.5 & 20 \\
                     & FZ-v               & 125 & 95 & 0.37 & 0.64 & 0.47 & 115.5 & 20 \\
                     & FZ-R               & 125 & 95 & 0.37 & 0.64 & 0.47 & 115.5 & 20 \\
                     & MVPC$^\ddagger$    & 107.5 & 79 & 0.43 & 0.62 & 0.51 & 95.5 & 20 \\
    \midrule
    MAR (20\%)       & P-CI               & 62.5 & 46 & 0.71 & 0.48 & 0.59 & 42 & 20 \\
                     & CC                 & 98 & 77 & 0.43 & 0.51 & 0.46 & 76 & 20 \\
                     & TW                 & 112.5 & 84 & 0.41 & 0.58 & 0.48 & 96 & 20 \\
                     & FZ-v               & 128 & 97 & 0.36 & 0.62 & 0.45 & 112.0 & 20 \\
                     & FZ-R               & 121 & 96 & 0.36 & 0.62 & 0.46 & 107.5 & 20 \\
                     & MVPC$^\ddagger$    & 83 & 63 & 0.53 & 0.57 & 0.55 & 70.5 & 20 \\
    \midrule
    MNAR (20\%)      & P-CI               & 64 & 45 & 0.73 & 0.48 & 0.59 & 42.0 & 20 \\
                     & CC                 & 66 & 66 & 0.00 & 0.00 & 0.00 & 0 & 20 \\
                     & TW                 & 103.5 & 78 & 0.43 & 0.58 & 0.49 & 88.5 & 20 \\
                     & FZ-v               & 122 & 94 & 0.38 & 0.61 & 0.47 & 108.5 & 20 \\
                     & FZ-R               & 120 & 90 & 0.39 & 0.61 & 0.48 & 107.5 & 20 \\
                     & MVPC$^\ddagger$    & 72 & 54 & 0.59 & 0.55 & 0.57 & 63 & 19 (1 DNF) \\
    \midrule
    Mixed (20\%)     & P-CI               & 65.5 & 47 & 0.72 & 0.48 & 0.58 & 43 & 20 \\
                     & CC                 & 89.5 & 75 & 0.35 & 0.25 & 0.30 & 49 & 20 \\
                     & TW                 & 104.5 & 80 & 0.42 & 0.58 & 0.49 & 91.0 & 20 \\
                     & FZ-v               & 124.5 & 92 & 0.39 & 0.62 & 0.48 & 112 & 20 \\
                     & FZ-R               & 120 & 92 & 0.38 & 0.62 & 0.47 & 107.5 & 20 \\
                     & MVPC$^\ddagger$    & 77.5 & 58 & 0.56 & 0.56 & 0.56 & 65.5 & 20 \\
    \bottomrule
\end{tabular}
}
\end{table}

\section{Sachs Network: Full Results}
\label{app:sachs}

Although $p=11$ is close to the $p=10$ threshold in Section~\ref{sec:exp_graph}, we use the fast variant of PAIR-CI for consistency with larger real-world benchmarks (ALARM, HAILFINDER) and because it produces near-identical results to the general variant at this scale.

\paragraph{Linear (real data).}
Table~\ref{tab:sachs_linear} reports results across all missingness rates (10\%, 20\%, 40\%) and mechanisms. PAIR-CI achieves median SHD of 19.5--23.5, compared to 17.0--19.0 for the deletion and voting baselines, 17.5--19.0 for FZ-Rubin, and 16.0--18.0 for MVPC. MVPC and MissDAG (15.5--20.0) record among the lowest SHD, consistent with the approximate optimality of their linear Gaussian assumptions at this scale, though F1 advantages are more mixed: 0.45--0.54 for PAIR-CI, 0.49--0.63 for the deletion and voting baselines, 0.53--0.64 for FZ-Rubin, 0.52--0.62 for MVPC, and 0.63--0.71 for MissDAG. As discussed in Section~\ref{sec:discussion}, Sachs represents a worst-case regime for PAIR-CI: a small graph with approximately linear relationships and hyperparameters tuned for $p \geq 20$.

\begin{table}[h]
    \centering
    \caption{\textbf{Sachs network results on the real flow-cytometry data (bootstrap resamples of the 853 observations).} Total SHD {\scriptsize(skeleton SHD)} across missingness rates and mechanisms (20 replicates, 11 nodes, 17 true edges). Methods: P-CI = PAIR-CI (fast); CC = complete case; TW = test-wise deletion; FZ-v = FZ-vote; FZ-R = FZ-Rubin. $^\ddagger$MVPC runs in its authors' released PC implementation. $\dagger$As MissDAG outputs a fully oriented DAG rather than a CPDAG, its total SHD (shown in gray) is not directly comparable to that of constraint-based methods; skeleton SHD provides the appropriate like-for-like metric.}
    \label{tab:sachs_linear}
    \resizebox{\textwidth}{!}{\begin{tabular}{llccccccc}
    \toprule
    Rate & Mechanism & P-CI & CC & TW & FZ-v & FZ-R & MVPC$^\ddagger$ & MissDAG$^\dagger$ \\
    \midrule
    0\% & Complete & 23.5\,{\scriptsize(17)} & 18.0\,{\scriptsize(11)} & 18.0\,{\scriptsize(11)} & 18.0\,{\scriptsize(11)} & 18.0\,{\scriptsize(11)} & 18.0\,{\scriptsize(11)} & \textcolor{lightgray}{15.5\,{\scriptsize(10)}}$^\dagger$ \\
    \midrule
    10\% & MAR   & 21.0\,{\scriptsize(16)} & 18.5\,{\scriptsize(11)} & 18.0\,{\scriptsize(10)} & 18.0\,{\scriptsize(11)} & 19.0\,{\scriptsize(11)} & 18.0\,{\scriptsize(10)} & \textcolor{lightgray}{15.5\,{\scriptsize(10)}}$^\dagger$ \\
         & MNAR  & 23.5\,{\scriptsize(18)} & 17.0\,{\scriptsize(12)} & 17.5\,{\scriptsize(11)} & 18.0\,{\scriptsize(12)} & 18.0\,{\scriptsize(12)} & 16.0\,{\scriptsize(11)} & \textcolor{lightgray}{16.5\,{\scriptsize(11)}}$^\dagger$ \\
         & Mixed & 23.5\,{\scriptsize(18)} & 19.0\,{\scriptsize(12)} & 18.0\,{\scriptsize(11)} & 18.0\,{\scriptsize(10)} & 17.5\,{\scriptsize(10)} & 17.5\,{\scriptsize(10)} & \textcolor{lightgray}{18.0\,{\scriptsize(12)}}$^\dagger$ \\
    \midrule
    20\% & MAR   & 20.0\,{\scriptsize(15)} & 18.0\,{\scriptsize(12)} & 18.0\,{\scriptsize(11)} & 19.0\,{\scriptsize(12)} & 18.0\,{\scriptsize(11)} & 18.0\,{\scriptsize(10)} & \textcolor{lightgray}{17.5\,{\scriptsize(11)}}$^\dagger$ \\
         & MNAR  & 22.5\,{\scriptsize(18)} & 18.0\,{\scriptsize(13)} & 17.5\,{\scriptsize(13)} & 18.5\,{\scriptsize(13)} & 18.0\,{\scriptsize(13)} & 17.0\,{\scriptsize(12)} & \textcolor{lightgray}{17.5\,{\scriptsize(12)}}$^\dagger$ \\
         & Mixed & 22.0\,{\scriptsize(16)} & 18.0\,{\scriptsize(13)} & 17.0\,{\scriptsize(13)} & 18.0\,{\scriptsize(12)} & 17.5\,{\scriptsize(12)} & 17.0\,{\scriptsize(12)} & \textcolor{lightgray}{19.5\,{\scriptsize(14)}}$^\dagger$ \\
    \midrule
    40\% & MAR   & 19.5\,{\scriptsize(15)} & 19.0\,{\scriptsize(11)} & 18.0\,{\scriptsize(11)} & 19.0\,{\scriptsize(12)} & 18.5\,{\scriptsize(11)} & 18.0\,{\scriptsize(10)} & \textcolor{lightgray}{20.0\,{\scriptsize(13)}}$^\dagger$ \\
         & MNAR  & 22.0\,{\scriptsize(15)} & 17.0\,{\scriptsize(13)} & 18.0\,{\scriptsize(13)} & 19.0\,{\scriptsize(12)} & 18.0\,{\scriptsize(12)} & 17.0\,{\scriptsize(12)} & \textcolor{lightgray}{18.0\,{\scriptsize(12)}}$^\dagger$ \\
         & Mixed & 20.0\,{\scriptsize(14)} & 18.0\,{\scriptsize(12)} & 18.0\,{\scriptsize(11)} & 19.0\,{\scriptsize(12)} & 18.0\,{\scriptsize(11)} & 17.0\,{\scriptsize(11)} & \textcolor{lightgray}{19.0\,{\scriptsize(13)}}$^\dagger$ \\
    \bottomrule
\end{tabular}
}
\end{table}

\paragraph{Nonlinear (synthetic).}
To isolate the effect of nonlinearity from those of scale and topology, we generate synthetic data on the same 11-node Sachs DAG using nonlinear edge mechanisms (Section~\ref{sec:exp_graph}). Table~\ref{tab:sachs_nonlinear} shows median SHD across all conditions. With nonlinear edges, PAIR-CI yields SHD 14--15 across all conditions, relative to 16--18 for test-wise deletion and FZ-vote, 15--17.5 for complete-case analysis, 16--19 for FZ-Rubin, and 15--16.5 for MVPC (F1 0.62--0.69). MissDAG returns the lowest SHD (13--14.5) in all but one condition (40\% MAR, where PAIR-CI's 14.0 is lower), albeit with substantially lower F1 (0.30--0.43 vs.\ 0.74--0.85 for PAIR-CI). This is because MissDAG recovers a sparse skeleton with high precision but low recall, reflecting a genuine precision--recall tradeoff rather than the ``winning by giving up'' pattern observed at larger $p$. PAIR-CI records the highest F1 among all methods (0.74--0.85), confirming that its advantage arises from nonlinearity rather than scale alone.

\begin{table}[h]
    \centering
    \caption{\textbf{Sachs network results with nonlinear edge mechanisms.} Median SHD {\scriptsize(skeleton SHD)} across missingness rates and mechanisms (20 replicates, 11 nodes, 17 edges). Skeleton SHD counts missing and extra edges only, while total SHD additionally penalizes orientation errors. Methods: P-CI = PAIR-CI (fast); CC = complete case; TW = test-wise deletion; FZ-v = FZ-vote; FZ-R = FZ-Rubin. $^\ddagger$MVPC runs in its authors' released PC implementation. $\dagger$As MissDAG outputs a fully oriented DAG rather than a CPDAG, its total SHD (shown in gray) is not directly comparable to that of constraint-based methods; skeleton SHD provides the appropriate like-for-like metric.}
    \label{tab:sachs_nonlinear}
    \resizebox{\textwidth}{!}{\begin{tabular}{llccccccc}
    \toprule
    Rate & Mechanism & P-CI & CC & TW & FZ-v & FZ-R & MVPC$^\ddagger$ & MissDAG$^\dagger$ \\
    \midrule
    0\% & Complete & 14.0\,{\scriptsize(5)} & 17.0\,{\scriptsize(12)} & 17.0\,{\scriptsize(12)} & 17.0\,{\scriptsize(12)} & 17.0\,{\scriptsize(12)} & 16.5\,{\scriptsize(12)} & \textcolor{lightgray}{13.0\,{\scriptsize(13)}}$^\dagger$ \\
    \midrule
    10\% & MAR   & 14.0\,{\scriptsize(5)} & 17.5\,{\scriptsize(11)} & 16.5\,{\scriptsize(11)} & 18.0\,{\scriptsize(11)} & 17.0\,{\scriptsize(11)} & 16.0\,{\scriptsize(10)} & \textcolor{lightgray}{14.0\,{\scriptsize(13)}}$^\dagger$ \\
         & MNAR  & 14.0\,{\scriptsize(6)} & 16.5\,{\scriptsize(10)} & 17.0\,{\scriptsize(11)} & 17.0\,{\scriptsize(11)} & 18.0\,{\scriptsize(11)} & 16.0\,{\scriptsize(10)} & \textcolor{lightgray}{14.0\,{\scriptsize(14)}}$^\dagger$ \\
         & Mixed & 15.0\,{\scriptsize(5)} & 15.0\,{\scriptsize(10)} & 17.0\,{\scriptsize(11)} & 17.0\,{\scriptsize(11)} & 17.0\,{\scriptsize(10)} & 16.0\,{\scriptsize(10)} & \textcolor{lightgray}{13.5\,{\scriptsize(13)}}$^\dagger$ \\
    \midrule
    20\% & MAR   & 15.0\,{\scriptsize(5)} & 17.0\,{\scriptsize(12)} & 16.0\,{\scriptsize(11)} & 17.0\,{\scriptsize(12)} & 19.0\,{\scriptsize(12)} & 16.0\,{\scriptsize(11)} & \textcolor{lightgray}{14.0\,{\scriptsize(13)}}$^\dagger$ \\
         & MNAR  & 14.0\,{\scriptsize(6)} & 16.0\,{\scriptsize(10)} & 16.5\,{\scriptsize(11)} & 17.0\,{\scriptsize(11)} & 17.0\,{\scriptsize(12)} & 16.0\,{\scriptsize(10)} & \textcolor{lightgray}{14.0\,{\scriptsize(14)}}$^\dagger$ \\
         & Mixed & 15.0\,{\scriptsize(7)} & 16.0\,{\scriptsize(11)} & 17.0\,{\scriptsize(10)} & 17.0\,{\scriptsize(11)} & 17.0\,{\scriptsize(10)} & 15.0\,{\scriptsize(10)} & \textcolor{lightgray}{13.0\,{\scriptsize(13)}}$^\dagger$ \\
    \midrule
    40\% & MAR   & 14.0\,{\scriptsize(6)} & 16.0\,{\scriptsize(9)} & 17.0\,{\scriptsize(10)} & 16.0\,{\scriptsize(10)} & 16.0\,{\scriptsize(12)} & 16.0\,{\scriptsize(9)} & \textcolor{lightgray}{14.5\,{\scriptsize(14)}}$^\dagger$ \\
         & MNAR  & 15.0\,{\scriptsize(8)} & 16.0\,{\scriptsize(10)} & 16.5\,{\scriptsize(11)} & 16.0\,{\scriptsize(12)} & 17.5\,{\scriptsize(11)} & 15.5\,{\scriptsize(11)} & \textcolor{lightgray}{14.0\,{\scriptsize(14)}}$^\dagger$ \\
         & Mixed & 15.0\,{\scriptsize(5)} & 17.0\,{\scriptsize(11)} & 16.0\,{\scriptsize(10)} & 16.5\,{\scriptsize(10)} & 16.0\,{\scriptsize(10)} & 15.5\,{\scriptsize(10)} & \textcolor{lightgray}{14.0\,{\scriptsize(14)}}$^\dagger$ \\
    \bottomrule
\end{tabular}
}
\end{table}

\section{Variance Estimator Validation}
\label{app:variance_comparison}

Our test uses the provably consistent within-imputation variance estimator of \citet[Theorem~4]{bayle2020cross}:
\begin{equation}
    U_m^{\mathrm{Bayle}} = \frac{1}{n} \cdot \frac{1}{K} \sum_{k=1}^{K} 
    \hat{\sigma}^2_{k,m}, \qquad
    \hat{\sigma}^2_{k,m} = \frac{1}{n_k - 1} \sum_{i \in F_k} 
    (d_{ikm} - \hat{\mu}_{km})^2,
    \label{eq:bayle}
\end{equation}
where $d_{ikm}$ denotes the individual loss difference for observation $i$ in fold $k$ and imputation $m$ (Equation~\ref{eq:app_fold_diff}). We evaluate this against the Nadeau--Bengio estimator
\[
    U_m^{\mathrm{NB}} = \left(K^{-1} + \frac{n_k}{n - n_k}\right) s^2_m,
    \qquad
    s^2_m := \frac{1}{K-1} \sum_k (\hat{\mu}_{km} - \hat{\mu}_m)^2,
\]
which \citet{nadeau1999inference} note may over- or underestimate the variance (their conservativeness conjecture concerns a different estimator).

\paragraph{Calibration.}
Under the null (signal $= 0$), false positive rates for the \citeauthor{bayle2020cross} estimator range from 0.9\% to 2.4\%, compared with 0.2--0.9\% for the Nadeau--Bengio estimator. All values are well below the nominal 5\%, but the \citeauthor{bayle2020cross} estimator lies closer to the asymptotic target. This pattern is consistent across mechanisms (complete: 0.2\%$\to$1.8\%; MAR: 0.7\%$\to$0.9\%; MNAR: 0.9\%$\to$2.4\%).

\paragraph{Power.}
A consistent variance estimate produces larger $t$-statistics and higher rejection rates. At signal $= 0.3$, the \citeauthor{bayle2020cross} estimator achieves 38.4--54.9\% power vs.\ 28.2--46.2\% for Nadeau--Bengio (a 9--10 percentage-point gain); at signal $= 1.0$, the equivalent figures are 96.2--99.3\% vs.\ 91.3--97.1\%. \citeauthor{bayle2020cross}'s advantage is therefore largest at moderate signal strengths, where the test operates on the steepest part of the power curve.

\paragraph{Graph recovery.}
At $p = 10$ with nonlinear edges, both estimators achieve identical median SHD (10 under both MAR and MNAR) and perfect median precision (1.000). The \citeauthor{bayle2020cross} estimator delivers modest gains in recall under MAR (0.467 vs.\ 0.455) and MNAR (0.500 vs.\ 0.449), corresponding to median F1 scores of 0.636 vs.\ 0.625 and 0.667 vs.\ 0.615, respectively.

\paragraph{Summary}
The \citeauthor{bayle2020cross} estimator provides provably consistent variance estimation while improving power by 9--10 percentage points at moderate signal strengths, yielding a clear advantage over the Nadeau--Bengio approach. The cost is a slightly narrower finite-sample calibration margin (2.4\% vs.\ 0.9\% maximum false positive rate), though both approaches remain below the 5\% target.

\section{Power for Sample-Size Planning}
\label{app:power_planning}

Extending Section~\ref{sec:exp_ci}'s analysis, Table~\ref{tab:power} presents average rejection rates under $H_1$ for PAIR-CI across three DGPs (linear Gaussian, post-nonlinear, and latent confounder), three signal strengths ($\{0.3, 0.6, 1.0\}$), and three sample sizes ($n \in \{500, 1{,}000, 2{,}000\}$) at $|\Xb| = 5$, averaged over MAR and MNAR mechanisms with 30\% missingness.

\begin{table}[t]
    \centering
    \caption{\textbf{Power for sample-size planning.} Average rejection rate under $H_1$ averaged over MAR and MNAR mechanisms with 30\% missingness, $|\Xb| = 5$, and $\alpha = 0.05$ (100 replicates per cell).}
    \label{tab:power}
    \begin{tabular}{llccc}
    \toprule
    DGP & Signal & $n{=}500$ & $n{=}1000$ & $n{=}2000$ \\
    \midrule
    Linear Gaussian    & 0.3 & 0.35   & 0.56   & 0.74   \\
                       & 0.6 & 0.88   & 0.96   & 0.99   \\
                       & 1.0 & 0.99   & 1.00   & 1.00   \\
    \midrule
    Post-nonlinear     & 0.3 & 0.12   & 0.29   & 0.55   \\
                       & 0.6 & 0.57   & 0.77   & 0.95   \\
                       & 1.0 & 0.91   & 0.97   & 1.00   \\
    \midrule
    Latent confounder  & 0.3 & 0.14   & 0.49   & 0.87   \\
                       & 0.6 & 0.69   & 0.96   & 1.00   \\
                       & 1.0 & 0.86   & 1.00   & 1.00   \\
    \bottomrule
\end{tabular}

\end{table}

\paragraph{Rule of thumb.}
For 80\% power under MAR or MNAR missingness, moderate signal ($= 0.6$) requires approximately $n \approx 500$ for the linear Gaussian DGP, $n \approx 1{,}000$ for the latent-confounder DGP, and $n \approx 2{,}000$ for the post-nonlinear DGP. A weak signal ($= 0.3$) demands substantially larger samples of $n > 2{,}000$ across all settings, with the post-nonlinear case requiring the most data. Strong signals ($= 1.0$) result in near-100\% power by $n \approx 1{,}000$ across conditions. Practitioners working with weak effects at small sample sizes should thus expect limited power in post-nonlinear conditions. If the functional form is approximately linear, the requisite sample size is substantially smaller.

\section{Sensitivity to Number of Imputations}
\label{app:m_sensitivity}

A natural question raised by the power limitation discussed in Section~\ref{sec:discussion} is whether increasing $M$ recovers power by tightening the Barnard--Rubin reference distribution. We examine this issue by extending the standalone experiment in Section~\ref{sec:exp_ci}, varying $M \in \{3, 5, 10, 20\}$ across two DGPs (linear Gaussian and post-nonlinear), two sample sizes ($n \in \{500, 1{,}000\}$), three missingness mechanisms (complete, MAR, and MNAR at 30\% missingness), and four signal strengths ($\{0, 0.3, 0.6, 1.0\}$), with $|\Xb| = 5$, $K = 5$, and 50 replicates per configuration.

\paragraph{Calibration preserved.}
The false positive rate at signal $= 0$ remains below the nominal 5\% level for all $M$, with 4.0\% the highest observed rate ($M = 5$ under MNAR).

\paragraph{Modest power gains at weak signals.}
At signal $= 0.3$, average rejection rates range from 0.25 to 0.34 at $M = 3$ and from 0.25 to 0.35 at $M = 20$, depending on the missingness mechanism. The largest gain occurs under MAR (approximately 6.5 percentage points, from 0.25 to 0.32), decreasing to about 1.5 percentage points under MNAR. At signal $= 0.6$, increasing $M$ from 5 to 20 yields gains of 0.5--3 percentage points. At signal $= 1.0$, power is already near 1, and additional imputations provide no further gain.

\paragraph{Why the gain is smaller than expected.}
Barnard--Rubin degrees of freedom depend on the fraction of missing information (i.e., the component of uncertainty attributable to missing data rather than sampling variability): $\gamma = (1 + 1/M)\,B / T$, where $B$ is between-imputation variance and $T = \bar{W} + (1 + 1/M)B$. When $\gamma$ is large, increasing $M$ substantially reduces the correction; when $\gamma$ is small, the reference distribution is already close to normal and more imputations have little effect on $\nu$. In our paired design, $\gamma$ is modest by construction: since both models receive the same imputed $\hatX$, imputation error is differenced out in $\hat{\mu}_m$, and variability across imputations is dominated by within-imputation cross-validation noise ($\bar{W}$) rather than between-imputation variance ($B$). In short, the error-canceling mechanism that confers MNAR robustness (Remark~\ref{rem:paired_robust}) is precisely what renders the Barnard--Rubin correction mild.

\paragraph{Implication.}
Residual underrejection at weak signals is therefore learner-bound rather than inference-bound, reflecting the signal-to-noise ratio at which random forests detect conditional dependence at $n = 500$--$1{,}000$ with $|\Xb| = 5$. In practice, power is primarily controlled by $n$, learner capacity, and $K$ rather than by $M$. Combined with the \citeauthor{bayle2020cross} estimator's gain of 9--10 percentage points (Appendix~\ref{app:variance_comparison}), the total tunable power budget is 9--16 percentage points at weak signals.

\section{Adversarial Robustness of the $\kappa$-Cancellation}
\label{app:adversarial_kappa}

We stress-test the $\kappa$-cancellation argument in Remark~\ref{rem:paired_robust} by constructing adversarial DGPs that simultaneously maximize all three components in the $\kappa$ decomposition (Appendix~\ref{app:kappa}): imputation error, the correlation of $Y$ with the unrecovered component of $\Xb$, and the dependence of $Z$ on that component. All tests are conducted at signal $= 0$ (true null) with $n = 500$, 30--50\% MNAR missingness (logistic missingness model with steepness parameter $= 5$), and cached imputation. We report 100 replicates per cell with exact Clopper--Pearson 95\% confidence intervals.

\paragraph{Adversarial DGPs.}
We consider six topologies: \emph{Hub} (one incomplete hub variable $X_0$ and 10 weakly correlated children; $Y = X_0 + \varepsilon_Y$, $Z = X_0 + \varepsilon_Z$); \emph{Chain} (chain $X_0 \to X_1 \to X_2 \to X_3$, with $X_2$ incomplete and both $Y$ and $Z$ functions of $X_2$); \emph{Dense-block} (5-variable block-correlated graph, 3 incomplete variables, within-block correlation $\rho = 0.8$, cross-block correlation $\rho = 0.2$); \emph{Weak-hub} (hub with $\alpha = 0.2$, inducing poor imputation); \emph{Branch-separator} ($Y$ depends only on an incomplete $A$, $Z$ on an independent incomplete $B$, with $A \indep B$); and \emph{Hub-nonlinear} (hub with $Y = \sin(X_0) + \varepsilon$, $Z = X_0^2 + \varepsilon$). We evaluate both full and incomplete-only conditioning sets.

\paragraph{Results with linear MICE (default).}
Across five of six adversarial DGPs---hub, chain, dense-block, weak-hub, and branch-separator---cached linear imputation maintains a false positive rate $\leq 1.5\%$ (upper CI $\leq 4.3\%$) in every cell, even at 50\% MNAR missingness. The paired design's $\kappa$-cancellation is robust in these settings because cached imputation includes $Y$ and $Z$ in the imputer, allowing the linear MICE model to absorb information about $\Xb_{\textup{miss}}$ into $\hatX$.

\paragraph{Failure on adversarial nonlinear hub.}
The single exception is the hub-nonlinear case (false positive rate 87--100\% with linear MICE), where $Y = \sin(X_0)$ and $Z = X_0^2$ induce nonlinear dependence that a linear imputer cannot capture. Random forest-based MICE lowers false positives in this setting (to 57--82\%) but inflates them in the linear adversarial DGP, reaching 83\% in the most extreme case due to overfitting at $n = 500$. The linear imputer is preferable in approximately linear settings: although it fails under strong nonlinear dependence between the test variables and $X_{\mathrm{miss}}$ (for instance, the hub-nonlinear case), it avoids the overfitting that inflates false positive rates with random forest-based \textsc{MICE} at small $n$.

\paragraph{Relationship between $\kappa$ and false positive rate.}
Figure~\ref{fig:kappa_fpr_trace} plots the estimated product $\hat\kappa_{\mathrm{imp}}^2\,|\hat\kappa_Y|\,|\hat\kappa_Z|$ of Definition~\ref{def:kappa_factors}'s three factors, with $\kappa_{\mathrm{imp}}$ squared as in the first-order expansion, against the PAIR-CI false positive rate for each adversarial cell. We estimate the factors by replicating the data-generating and imputation steps, which are unconditional analogues of the conditional correlations in Definition~\ref{def:kappa_factors}. Under MNAR, the two residual correlations can differ because the missingness mechanism induces $\E[\delta \mid \hatX] \neq 0$, though their ordering across cells tracks the formal residual empirically.

Three regimes emerge. In the 20 cells with an estimated product below $10^{-2}$ the false positive rate is at most 0.01; in the six cells between $10^{-2}$ and $10^{-1}$ it ranges from 0.00 to 0.16; and in the 22 cells above $10^{-1}$, which include the hub-nonlinear pathology under linear MICE and the overfitting setting of random-forest MICE, it ranges from 0.09 to 1.00. Over all 48 cells the Spearman correlation between the estimated product and the false positive rate is 0.89, calibrated and inflated cells have median products of 0.002 and 0.18, and on the 25 inflated cells the least-squares slope of $\log(\mathrm{FPR} - 0.05)$ on the log product is $0.69 \pm 0.15$.

The fitted slope is of the order the first-order expansion of Decomposition~\ref{dec:kappa} predicts, and the two-orders-of-magnitude separation between calibrated and inflated cells is what makes the estimated product usable as a diagnostic: the benchmark MNAR settings of Section~\ref{sec:experiments} sit in the calibrated regime because $|\mathrm{corr}(\delta, Y)| \leq 0.011$ in every cell (Appendix~\ref{app:imputer_quality}).

\begin{figure}[t]
    \centering
    \includegraphics[width=0.85\textwidth]{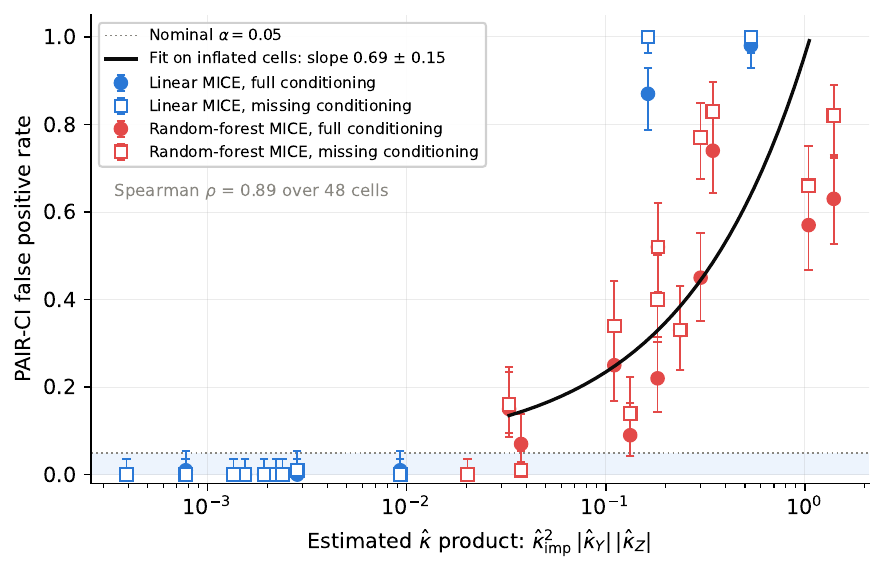}
    \caption{\textbf{Residual bias and calibration across adversarial DGPs.} PAIR-CI's false positive rate against the estimated product $\hat\kappa_{\mathrm{imp}}^2\,|\hat\kappa_Y|\,|\hat\kappa_Z|$ of Definition~\ref{def:kappa_factors}'s three factors (with $\kappa_{\mathrm{imp}}$ squared, the convention of the first-order expansion) over 48 cells (six DGPs $\times$ two imputers $\times$ two conditioning strategies $\times$ two rates; 100 replicates each, with 95\% intervals). Spearman correlation 0.89 over all cells; calibrated and inflated cells have median products 0.002 and 0.18. The black curve is the least-squares fit of $\log(\mathrm{FPR} - 0.05)$ on $\log$ of the product over the 25 inflated cells (slope $0.69 \pm 0.15$), drawn where the fitted rate is below 1.}
    \label{fig:kappa_fpr_trace}
\end{figure}

\paragraph{Operating envelope.}
The $\kappa$-cancellation with cached linear MICE remains robust unless the incomplete variable simultaneously exhibits (i) nonlinear influence on both test variables, (ii) strong value-dependent MNAR missingness, and (iii) no linear proxy in the observed data. This conjunction defines a narrow failure mode that does not arise in our experimental benchmarks (Erd\H{o}s--R\'enyi, ALARM, HAILFINDER, Sachs). Practitioners working with strongly nonlinear DGPs and adversarial missingness should verify calibration empirically prior to deployment.

\section{Learner Comparison}
\label{app:learner_comparison}

To assess whether the power limitation in Section~\ref{sec:discussion} can be mitigated by a stronger base learner, we compare random forests with three alternatives---ExtraTrees \citep{geurts2006extremely}, LightGBM, and XGBoost---with LightGBM evaluated under both default and tuned settings  ($n_{\textup{trees}} = 200$, learning rate $0.05$, $\texttt{min\_child\_samples} = 2$, $\texttt{num\_leaves} = 63$, subsample $= 0.8$). All variants differ only in the choice of learner, with the remainder of the PAIR-CI architecture held fixed (paired permutation, \citeauthor{bayle2020cross} within-imputation variance, Rubin's rules, Barnard--Rubin degrees of freedom). We run 50 replicates at $n \in \{500, 1{,}000\}$, $|\Xb| = 5$, and signal $\in \{0, 0.3, 0.6, 1.0\}$ across two DGPs (linear Gaussian and post-nonlinear) under complete, MAR, and MNAR mechanisms.

\paragraph{Calibration.}
ExtraTrees exhibits the highest standalone calibration error under MNAR (false positive rate 9\%, 95\% CI [5.4\%, 13.9\%]). Random forests are slightly elevated (5.5\%) but within sampling error of the nominal level, while gradient-boosting variants (LightGBM, LightGBM-tuned, XGBoost) maintain false positive rates $\leq 1.5\%$. The additional split randomization in ExtraTrees inflates fold-level variance that the \citeauthor{bayle2020cross} estimator does not fully absorb. %

\paragraph{Power.}
Random forests are most powerful at weak signals. At signal $= 0.3$ under MNAR, they achieve 36\% rejection, against 24\% for XGBoost, 23\% for LightGBM, and 22\% for LightGBM-tuned. At signal $= 0.6$, the gap closes: all four learners fall within 5 percentage points under MNAR (random forests 82\%, XGBoost 82\%, LightGBM 83\%, LightGBM-tuned 78\%). Bagging-based variance reduction in random forests interacts more favorably with the \citeauthor{bayle2020cross} estimator at weak signals, and the choice of learner matters less at moderate-to-strong signals.

\paragraph{Runtime.}
Among calibrated learners, random forests are fastest in median runtime (3.5\,s/test vs.\ 4.4\,s for LightGBM, 5.2\,s for XGBoost, and 18.1\,s for LightGBM-tuned). ExtraTrees is faster still (2.4\,s) but exhibits elevated false positive rates under MNAR in standalone evaluation (as noted above), motivating its restriction to the fast variant with early stopping. 

\paragraph{Summary}
In sum, random forests offer the best combination of weak-signal power, near-nominal calibration, and competitive runtime, justifying their selection as the default for the general variant.

\section{Early Stopping Calibration}
\label{app:early_stopping}

The fast variant of PAIR-CI employs an early-stopping heuristic: if the absolute $t$-statistic exceeds 4.0 after $M' = 2$, remaining imputations are skipped. This rule reduces computation when signal is strong but could introduce anti-conservative bias in borderline cases where the first two imputations yield large $t$-statistics by chance. We investigate this possibility by repeating the standalone experiment (signal $= 0$, $n = 1{,}000$) with and without early stopping, comparing false positive rates across DGPs and missingness mechanisms. Across the 600 null evaluations (2 DGPs $\times$ 3 mechanisms $\times$ 100 replicates), decision-level agreement is 100\%, confirming that the early-stopping heuristic does not affect calibration. Under $H_0$, $t$-statistics remain small and the threshold of 4.0 is never reached.

\section{Imputation Degradation: Calibration under Poor Imputation}
\label{app:imputation_degradation}

To test Remark~\ref{rem:paired_robust}'s contention that calibration for the internal null $Z \indep Y \mid \hatX$ extends to the scientific null when imputation is adequate, we deliberately degrade imputation quality in the standalone calibration experiment (Section~\ref{sec:exp_ci}).

\paragraph{Setup.}
We implement PAIR-CI under the null (signal $= 0$) with three imputation strategies of declining quality:
\begin{itemize}
    \item \textbf{MICE}: scikit-learn's \texttt{IterativeImputer} function with stochastic posterior draws ($M = 5$ distinct datasets);
    \item \textbf{Mean}: deterministic column-mean imputation ($M = 5$ identical datasets, so between-imputation variance $B = 0$);
    \item \textbf{Marginal}: random draws from the observed marginal of each column ($M = 5$ stochastic datasets, without conditioning on other variables).
\end{itemize}
We use 200 repetitions per configuration with $n = 500$, $|\Xb| = 5$, and 30\% missingness across linear Gaussian and post-nonlinear DGPs under complete, MAR, and MNAR mechanisms. Power at signal $= 0.6$ is also evaluated.

\paragraph{False positive rate.}
Panel~A in Table~\ref{tab:degradation_combined} displays rejection rates under $H_0$. MICE maintains false positive rates $\leq 3.75\%$ across all mechanisms, consistent with Table~\ref{tab:calibration}. 
Mean and marginal imputation inflate false positives substantially: mean imputation reaches 29.0\% under MAR, marginal imputation 43.2\% under MNAR. All strategies coincide at 2.0\% on complete data---when there is nothing to impute---confirming that inflation arises only when low-quality imputation interacts with genuine missingness.

\begin{table}[h]
    \centering
    \caption{\textbf{Calibration and power under degraded imputation.} Panel~A reports average rejection rates under $H_0$ (signal $= 0$); Panel~B reports rejection rates at signal $= 0.6$. Both panels average over linear Gaussian and post-nonlinear DGPs across 200 repetitions ($n = 500$, $|\Xb| = 5$, $\alpha = 0.05$). In Panel~A, values exceeding $\alpha$ are shown in \textcolor{red}{red}. In Panel~B, rejection rates for mean and marginal imputation under MAR and MNAR (shown in gray) are uninterpretable as power due to inflated false positives (Panel~A).}
    \label{tab:degradation_combined}
    \begin{tabular}{lccc}
    \toprule
    Strategy & Complete & MAR & MNAR \\
    \midrule
    \multicolumn{4}{l}{\textit{Panel A: False positive rate (signal $= 0$)}} \\
    MICE      & 0.020 & 0.028 & 0.037 \\
    Mean      & 0.020 & \textcolor{red}{0.290} & \textcolor{red}{0.172} \\
    Marginal  & 0.020 & \textcolor{red}{0.307} & \textcolor{red}{0.432} \\
    \midrule
    \multicolumn{4}{l}{\textit{Panel B: Power (signal $= 0.6$)}} \\
    MICE      & 0.757 & 0.767 & 0.760 \\
    Mean$^\dagger$ & 0.757 & \textcolor{lightgray}{0.830} & \textcolor{lightgray}{0.740} \\
    Marginal$^\dagger$ & 0.757 & \textcolor{lightgray}{0.823} & \textcolor{lightgray}{0.760} \\
    \bottomrule
\end{tabular}

\end{table}

\paragraph{Power.}
Panel~B reports rejection rates at signal $= 0.6$. MICE's power is stable across mechanisms (75.7--76.7\%). Mean and marginal imputation yield higher rejection rates under MAR and MNAR, though these figures are not interpretable as power: the corresponding false positive rate is up to 8.6 times the nominal level (Panel~A), so excess rejections conflate true and false positives.

\paragraph{Interpretation.}
Mean and marginal imputation distort the conditioning set sufficiently that $\hatX$ induces spurious associations between $Z$ and $Y$---exactly the failure mode identified in Proposition~\ref{prop:impossibility}. The paired design does not protect against this: with mean or marginal imputation, PAIR-CI's false positive rate is 29--31\% under MAR, five times that of Rubin's rules with MICE (5.8\% for FZ-Rubin in Table~\ref{tab:calibration}), and 17--43\% under MNAR, against FZ-Rubin's 28\%. Imputation quality is therefore load-bearing for calibration---not merely for power---and practitioners should not opt for cruder imputation strategies without verifying calibration empirically.

\section{Additional Standalone Calibration Results}
\label{app:calib_extra}

\subsection{Baseline Configurations and Audit}
\label{app:audit}

\paragraph{Configurations.} GCM \citep{shah2020hardness} in Table~\ref{tab:calibration} regresses $Y$ and $Z$ on $\Xb$ with 50-tree random forests (\texttt{min\_samples\_leaf} $= 5$), obtains residuals by 5-fold cross-fitting, and uses the standard normalized residual-product statistic. KCI \citep{zhang2011kernel} uses \texttt{causal-learn}'s default configuration: Gaussian kernels with the empirical bandwidth, unstandardized inputs, and the gamma approximation to the null. Fisher's $Z$ uses partial correlations on the completed or complete-case data. Rubin's rules pool Fisher's $Z$ across the $M = 5$ imputations; the imputed GCM pipeline tests a single posterior draw, as FZ-single does, and the imputed KCI pipeline combines the $M$ per-imputation $p$-values by Fisher's method. On complete data the $M$ completed datasets coincide, so Fisher's method pools five identical $p$-values as if independent, which triples the effective level (it rejects whenever the single $p$-value is below $e^{-\chi^2_{10}(0.95)/10} \approx 0.16$); the submitted version of Table~\ref{tab:calibration} reported this imputed-pipeline figure (0.278) for KCI on complete data, whereas single-test KCI on complete data sits at 0.127--0.173 across DGPs. Table~\ref{tab:calibration} now reports complete-case variants in the Complete column for every method. RCIT \citep{strobl2019approximate} is a random-Fourier approximation of the same kernel family as KCI and is excluded on the same grounds; the conclusions of the audit below transfer to it.

\paragraph{Audit.} Table~\ref{tab:kci_gcm_audit} re-runs both tests on complete data across configurations. For GCM we vary the forest size ($50$, $100$, $400$ trees) and the residual construction (5-fold cross-fitted against honest out-of-bag), 400 replicates per cell over $n \in \{500, 1{,}000, 2{,}000, 5{,}000\}$ and $|\Xb| \in \{2, 5, 10\}$ (172,800 test evaluations). All six configurations sit at 21--23\% pooled false positive rate, all are worst on the linear-Gaussian DGP (0.33--0.36), and all rise with $n$ (0.197 at $n = 500$ to 0.248 at $n = 5{,}000$) and sharply with $|\Xb|$ (0.071 at $|\Xb| = 2$ to 0.423 at $|\Xb| = 10$); quadrupling the forest and using honest out-of-bag residuals move the rate by under two points. GCM's guarantee \citep[Theorem~6, and Theorem~8 for regressions fitted on independent auxiliary data, the single-split analogue of our cross-fitting]{shah2020hardness} is conditional on the product of the two regressions' mean-squared prediction errors being $o_P(n^{-1})$, equivalently on the product of their root-mean-squared errors being $o_P(n^{-1/2})$, which forest rates do not deliver, and the pattern across DGPs matches: the errors correlate where both regressions share features (linear Gaussian, post-nonlinear) and the product vanishes on the latent-confounder DGP, where the $\Xb$-dependencies are disjoint and GCM is calibrated (0.054). A correctly specified linear learner satisfies the condition, and Table~\ref{tab:gcm_ols} confirms that GCM with ordinary least squares calibrates in every complete-data cell (linear DGP 0.048; cells 0.030--0.070), flat in $n$ and $|\Xb|$ where the forests grow in both. Unknown functional form is the use case, so the flexible configuration is the like-for-like comparator. For KCI we standardize the inputs, set the kernel length-scale to the median pairwise distance, and simulate the null from the kernel eigenvalues instead of the gamma approximation ($200$ replicates per cell, $n \in \{500, 1{,}000\}$ because of the $O(n^3)$ cost). These changes halve the pooled complete-data rate from 0.144 to 0.072 (0.078 with the gamma null); the tuned test improves with $n$ (0.082 at $n = 500$, 0.062 at $n = 1{,}000$; power 0.98 and 1.00 at signal 0.3), and in the regime of the original KCI papers ($n = 1{,}000$, $|\Xb| \leq 5$) it is nominal at 0.044. Tuned KCI still sits at 0.102 when $|\Xb| = 10$, so no configuration stays nominal everywhere, and Table~\ref{tab:calibration} reports the tuned variant beside the default.

\begin{table}[h]
    \centering
    \caption{\textbf{Complete-data audit of GCM and KCI configurations.} False positive rate at signal $0$, pooled and disaggregated by DGP, $n$, and $|\Xb|$, and power at signal $0.3$. GCM: 400 replicates per cell; KCI: 200 replicates per cell at $n \in \{500, 1{,}000\}$ only. Values exceeding $\alpha = 0.05$ in \textcolor{red}{red}.}
    \label{tab:kci_gcm_audit}
    \resizebox{\textwidth}{!}{\begin{tabular}{lcccccccccccc}
    \toprule
     & & \multicolumn{3}{c}{DGP} & \multicolumn{4}{c}{$n$} & \multicolumn{3}{c}{$|\Xb|$} & \\
    \cmidrule(lr){3-5} \cmidrule(lr){6-9} \cmidrule(lr){10-12}
    Configuration & Pooled & Lin.\ G. & Post-nl & Latent & 500 & 1{,}000 & 2{,}000 & 5{,}000 & 2 & 5 & 10 & Power \\
    \midrule
    GCM, RF-50, cross-fitted (Table 1)   & \textcolor{red}{0.229} & \textcolor{red}{0.355} & \textcolor{red}{0.274} & \textcolor{red}{0.056} & \textcolor{red}{0.208} & \textcolor{red}{0.217} & \textcolor{red}{0.240} & \textcolor{red}{0.249} & \textcolor{red}{0.068} & \textcolor{red}{0.188} & \textcolor{red}{0.430} & 0.969 \\
    GCM, RF-100, cross-fitted            & \textcolor{red}{0.231} & \textcolor{red}{0.355} & \textcolor{red}{0.280} & \textcolor{red}{0.059} & \textcolor{red}{0.208} & \textcolor{red}{0.225} & \textcolor{red}{0.237} & \textcolor{red}{0.254} & \textcolor{red}{0.076} & \textcolor{red}{0.188} & \textcolor{red}{0.429} & 0.969 \\
    GCM, RF-400, cross-fitted            & \textcolor{red}{0.232} & \textcolor{red}{0.354} & \textcolor{red}{0.288} & \textcolor{red}{0.054} & \textcolor{red}{0.203} & \textcolor{red}{0.224} & \textcolor{red}{0.242} & \textcolor{red}{0.258} & \textcolor{red}{0.080} & \textcolor{red}{0.184} & \textcolor{red}{0.432} & 0.970 \\
    GCM, RF-50, honest OOB               & \textcolor{red}{0.216} & \textcolor{red}{0.330} & \textcolor{red}{0.259} & \textcolor{red}{0.058} & \textcolor{red}{0.191} & \textcolor{red}{0.208} & \textcolor{red}{0.222} & \textcolor{red}{0.242} & \textcolor{red}{0.072} & \textcolor{red}{0.164} & \textcolor{red}{0.412} & 0.969 \\
    GCM, RF-100, honest OOB              & \textcolor{red}{0.218} & \textcolor{red}{0.341} & \textcolor{red}{0.260} & \textcolor{red}{0.053} & \textcolor{red}{0.187} & \textcolor{red}{0.222} & \textcolor{red}{0.226} & \textcolor{red}{0.236} & \textcolor{red}{0.068} & \textcolor{red}{0.166} & \textcolor{red}{0.420} & 0.970 \\
    GCM, RF-400, honest OOB              & \textcolor{red}{0.214} & \textcolor{red}{0.330} & \textcolor{red}{0.259} & \textcolor{red}{0.053} & \textcolor{red}{0.186} & \textcolor{red}{0.207} & \textcolor{red}{0.218} & \textcolor{red}{0.246} & \textcolor{red}{0.066} & \textcolor{red}{0.164} & \textcolor{red}{0.413} & 0.969 \\
    \midrule
    KCI, default (Table 1)               & \textcolor{red}{0.144} & \textcolor{red}{0.152} & \textcolor{red}{0.138} & \textcolor{red}{0.144} & \textcolor{red}{0.166} & \textcolor{red}{0.123} & -- & -- & \textcolor{red}{0.073} & \textcolor{red}{0.144} & \textcolor{red}{0.216} & 0.883 \\
    KCI, standardised inputs             & \textcolor{red}{0.143} & \textcolor{red}{0.142} & \textcolor{red}{0.134} & \textcolor{red}{0.154} & \textcolor{red}{0.147} & \textcolor{red}{0.140} & -- & -- & \textcolor{red}{0.061} & \textcolor{red}{0.143} & \textcolor{red}{0.226} & 0.890 \\
    KCI, + median-heuristic width        & \textcolor{red}{0.078} & \textcolor{red}{0.077} & \textcolor{red}{0.088} & \textcolor{red}{0.070} & \textcolor{red}{0.084} & \textcolor{red}{0.073} & -- & -- & \textcolor{red}{0.057} & \textcolor{red}{0.072} & \textcolor{red}{0.107} & 0.994 \\
    KCI, + spectral null (tuned)         & \textcolor{red}{0.072} & \textcolor{red}{0.075} & \textcolor{red}{0.068} & \textcolor{red}{0.072} & \textcolor{red}{0.082} & \textcolor{red}{0.062} & -- & -- & 0.044 & \textcolor{red}{0.069} & \textcolor{red}{0.102} & 0.991 \\
    \bottomrule
\end{tabular}
}
\end{table}

\begin{table}[h]
    \centering
    \caption{\textbf{GCM with a correctly specified linear learner.} False positive rate at signal $0$ on complete data by DGP, $n$, and $|\Xb|$ (300 replicates per cell; the $n = 500$, $|\Xb| = 10$ cells were not run), with pooled false positive rate and power at signal $0.3$ per DGP.}
    \label{tab:gcm_ols}
    \begin{tabular}{llcccc}
    \toprule
    DGP & $n$ & $|\Xb|{=}2$ & $|\Xb|{=}5$ & $|\Xb|{=}10$ & Pooled FPR / power \\
    \midrule
    Linear Gaussian    & 500 & 0.033 & 0.043 & -- & 0.048 / 1.000 \\
                       & 1,000 & \textcolor{red}{0.067} & 0.040 & 0.037 &  \\
                       & 2,000 & \textcolor{red}{0.067} & \textcolor{red}{0.053} & \textcolor{red}{0.057} &  \\
                       & 5,000 & \textcolor{red}{0.053} & 0.043 & 0.033 &  \\
    \midrule
    Post-nonlinear     & 500 & \textcolor{red}{0.060} & 0.047 & -- & 0.051 / 1.000 \\
                       & 1,000 & 0.037 & \textcolor{red}{0.070} & 0.050 &  \\
                       & 2,000 & \textcolor{red}{0.063} & 0.037 & 0.043 &  \\
                       & 5,000 & \textcolor{red}{0.067} & 0.050 & 0.037 &  \\
    \midrule
    Latent confounder  & 500 & 0.040 & \textcolor{red}{0.067} & -- & 0.049 / 0.975 \\
                       & 1,000 & \textcolor{red}{0.057} & \textcolor{red}{0.057} & 0.030 &  \\
                       & 2,000 & \textcolor{red}{0.057} & 0.050 & \textcolor{red}{0.057} &  \\
                       & 5,000 & 0.047 & 0.030 & 0.050 &  \\
    \bottomrule
\end{tabular}

\end{table}

\subsection{Calibration Disaggregated by Sample Size and Conditioning Dimension}
\label{app:calibration_disagg}

Panel~A of Table~\ref{tab:calibration} pools over DGPs, sample sizes, and conditioning-set sizes by design; Table~\ref{tab:calibration_by_dgp} gives the per-DGP breakdown, Panel~B of Table~\ref{tab:calibration} the per-$n$ breakdown, and Table~\ref{tab:calibration_full} the full mechanism $\times$ $n$ $\times$ $|\Xb|$ grid for PAIR-CI. The false positive rate never exceeds the nominal level at the mechanism $\times$ $n$ grain (maximum 0.049) and reaches at most 0.067 at the mechanism $\times$ $|\Xb|$ grain and 0.077 at the mechanism $\times$ $n$ $\times$ $|\Xb|$ grain, in the $|\Xb| = 10$ cells; the neighborhoods of the local permutation are coarse there at these sample sizes (Section~\ref{sec:ci_test}).

\begin{table}[h]
    \centering
    \caption{\textbf{PAIR-CI false positive rate by mechanism, $n$, and $|\Xb|$.} Signal $0$, pooled over the three DGPs (300 tests per cell; the $n = 500$, $|\Xb| = 10$ cells are not in the grid). Values exceeding $\alpha = 0.05$ in \textcolor{red}{red}.}
    \label{tab:calibration_full}
    \begin{tabular}{llccc}
    \toprule
    Mechanism & $n$ & $|\Xb|{=}2$ & $|\Xb|{=}5$ & $|\Xb|{=}10$ \\
    \midrule
    Complete  & 500 & 0.003 & 0.010 & -- \\
              & 1,000 & 0.003 & 0.010 & 0.033 \\
              & 2,000 & 0.003 & 0.010 & \textcolor{red}{0.053} \\
              & 5,000 & 0.000 & 0.010 & \textcolor{red}{0.067} \\
    \midrule
    MAR       & 500 & 0.000 & 0.000 & -- \\
              & 1,000 & 0.003 & 0.013 & 0.047 \\
              & 2,000 & 0.003 & 0.003 & \textcolor{red}{0.073} \\
              & 5,000 & 0.000 & 0.013 & \textcolor{red}{0.073} \\
    \midrule
    MNAR      & 500 & 0.007 & 0.030 & -- \\
              & 1,000 & 0.003 & 0.040 & \textcolor{red}{0.067} \\
              & 2,000 & 0.010 & \textcolor{red}{0.060} & \textcolor{red}{0.077} \\
              & 5,000 & 0.027 & 0.020 & \textcolor{red}{0.057} \\
    \bottomrule
\end{tabular}

\end{table}

\subsection{Conservatism of the Combined Variance Estimator}
\label{app:variance_ratio}

Theorem~\ref{thm:bridge}(i) predicts that Rubin's combined variance $T$ over-estimates the variance of $\hat\mu$ under the internal null. We verify this across 24 calibration cells (linear-Gaussian and post-nonlinear DGPs; complete, MAR, and MNAR mechanisms; $n \in \{500, 1{,}000\}$; $|\Xb| \in \{2, 5\}$), generating 250 independent datasets per cell and running the general variant of PAIR-CI four times per dataset with fresh internal seeds (imputation draws, fold splits, placebo draws, forest seeds) while holding the data fixed. Table~\ref{tab:variance_ratio} compares the mean of $T$ with the Monte Carlo variance of $\hat\mu$ over independent datasets. The raw ratio exceeds 1 in 21 of 24 cells, reaching 5.6, and the false positive rate follows the overshoot as predicted: 0.004 at ratios near 3.5 against 0.072 in the single cell above nominal, whose raw ratio is about 0.75 and which is the 0.070 cell of Table~\ref{tab:calibration_by_dgp}. The three ratios below 1 are heavy-tail artifacts: on a robust scale (median of $T$ against the squared median absolute deviation of $\hat\mu$), $T$ over-covers in all 24 cells, by factors of 3.0--6.1. Replacing the Barnard--Rubin $t_\nu$ critical value by the standard normal changes rejection rates by at most 0.012, so the conservatism enters through the variance estimator, not the critical value. Four cells generated by the pre-correction MAR $\times$ $|\Xb| = 2$ code (Appendix~\ref{app:changelog}) are labelled by the mechanism they actually realized, one-column self-masking, and are retained because the measurement is valid for what was generated.

\begin{table}[h]
    \centering
    \caption{\textbf{Over-coverage of the combined variance estimator.} Per cell: mean of $T$ over all pipeline runs divided by the Monte Carlo variance of $\hat\mu$ over 250 independent datasets (raw; values below 1 in \textcolor{red}{red}); median of $T$ divided by the squared median absolute deviation of $\hat\mu$ over all runs (robust); and the false positive rate under the Barnard--Rubin $t_\nu$ and under $\mathcal{N}(0,1)$ critical values.}
    \label{tab:variance_ratio}
    \resizebox{\textwidth}{!}{\begin{tabular}{llcccccc}
    \toprule
    DGP & Mechanism & $n$ & $|\Xb|$ & $\bar T/\widehat{\mathrm{Var}}(\hat\mu)$ & robust & FPR ($t_\nu$) & FPR ($N(0,1)$) \\
    \midrule
    Linear Gaussian  & Complete       & 500 & 2 & 3.18 & 3.63 & 0.004 & 0.004 \\
                     & Complete       & 500 & 5 & 1.53 & 4.56 & 0.012 & 0.012 \\
                     & Complete       & 1,000 & 2 & 3.54 & 3.16 & 0.004 & 0.004 \\
                     & Complete       & 1,000 & 5 & \textcolor{red}{0.90} & 3.30 & 0.012 & 0.020 \\
                     & MAR            & 500 & 5 & 1.18 & 4.39 & 0.032 & 0.036 \\
                     & MAR            & 1,000 & 5 & \textcolor{red}{0.27} & 3.56 & 0.020 & 0.024 \\
                     & MNAR           & 500 & 2 & 3.00 & 3.97 & 0.008 & 0.008 \\
                     & MNAR           & 500 & 5 & 1.73 & 3.85 & 0.028 & 0.032 \\
                     & MNAR           & 1,000 & 2 & 3.63 & 4.64 & 0.004 & 0.004 \\
                     & MNAR           & 1,000 & 5 & \textcolor{red}{0.75} & 2.96 & \textcolor{red}{0.072} & 0.084 \\
                     & MNAR (1 col.)  & 500 & 2 & 4.55 & 3.67 & 0.004 & 0.004 \\
                     & MNAR (1 col.)  & 1,000 & 2 & 3.55 & 3.74 & 0.004 & 0.004 \\
    \midrule
    Post-nonlinear   & Complete       & 500 & 2 & 3.68 & 4.13 & 0.000 & 0.000 \\
                     & Complete       & 500 & 5 & 2.90 & 5.96 & 0.004 & 0.004 \\
                     & Complete       & 1,000 & 2 & 3.30 & 3.40 & 0.000 & 0.000 \\
                     & Complete       & 1,000 & 5 & 1.05 & 4.41 & 0.012 & 0.016 \\
                     & MAR            & 500 & 5 & 2.62 & 6.11 & 0.000 & 0.000 \\
                     & MAR            & 1,000 & 5 & 1.33 & 4.96 & 0.016 & 0.020 \\
                     & MNAR           & 500 & 2 & 3.31 & 4.56 & 0.012 & 0.012 \\
                     & MNAR           & 500 & 5 & 1.55 & 4.31 & 0.016 & 0.028 \\
                     & MNAR           & 1,000 & 2 & 3.67 & 4.39 & 0.000 & 0.004 \\
                     & MNAR           & 1,000 & 5 & 1.47 & 4.88 & 0.016 & 0.024 \\
                     & MNAR (1 col.)  & 500 & 2 & 5.63 & 5.39 & 0.004 & 0.004 \\
                     & MNAR (1 col.)  & 1,000 & 2 & 3.77 & 4.10 & 0.000 & 0.000 \\
    \bottomrule
\end{tabular}
}
\end{table}

\section{Assumptions in Practice}
\label{app:assumptions_practice}

\subsection{Untestable Assumptions on the Missingness Mechanism}
\label{app:untestable}

All missing-data methods rely on assumptions about the mechanism and functional form of missingness that cannot be verified from the observed data,\footnote{In practice, as \citet[567]{graham2009missing} notes, these distinctions are often blurry: ``The best way to think of all missing data is as a continuum between MAR and MNAR. Because all missingness is MNAR (i.e., not purely MAR), then whether it is MNAR or not should never be the issue.''} and PAIR-CI is no exception: its proven guarantee covers the internal null under the weak condition that the imputed datasets converge to a common limit (A1$'$) and, up to a measured caching residual, extends to the scientific null under MAR with a consistent imputer (A1). What distinguishes our approach is that the dependence on A1 is explicit and its violation measured rather than assumed (Appendix~\ref{app:imputer_quality}). Under self-masking MNAR, where A1 fails by construction, PAIR-CI remains calibrated in every benchmark cell because the imputation error is essentially uncorrelated with the candidate variable, so the product of residual correlations in Decomposition~\ref{dec:kappa} vanishes; false positive rates inflate when the distortion is systematic, as with mean or marginal imputation (Appendix~\ref{app:imputation_degradation}), or under adversarial DGPs where the imputer cannot approximate the true conditional (Appendix~\ref{app:adversarial_kappa}), where the estimated $\hat\kappa$ product predicts the inflation. Under MNAR, power as well as size rests on this empirical tier.

\subsection{Learner Robustness}
\label{app:learner_robustness}

Table~\ref{tab:assumptions_practice} maps each assumption of Section~\ref{sec:theory} to its implementation counterpart, a diagnostic, and the behavior when it fails. Two of them deserve emphasis. First, calibration requires learner stability (A2--A4), not learner accuracy: by Proposition~\ref{prop:calibration}, a weak learner costs power, not size. Table~\ref{tab:learner_robustness} demonstrates this asymmetry across 36,000 tests with PAIR-CI run with the default 100-tree forest, two deliberately weakened forests (25 trees of depth 3; 10 stumps of depth 1), a misspecified linear ridge learner on nonlinear DGPs, and fixed-$k$ nearest-neighbor regression, which sits outside the stability assumption, across two DGPs, three mechanisms, $n \in \{1{,}000, 2{,}000\}$, $|\Xb| = 5$, and signals $\{0, 0.3, 0.6\}$ with 200 replicates per cell. The ridge learner stays calibrated (false positive rate exactly 0 on complete data, 0.001 under MAR, with full power where the signal is linear), as exchangeability predicts, rising only under MNAR (0.081); the stump forest keeps a false positive rate of 0.022 while its power collapses to 0.07--0.12. Elevated rates concentrate in linear-Gaussian settings (shallow forest 0.078, fixed-$k$ nearest neighbors 0.108), the same shared-underfit signature we diagnose for GCM in Appendix~\ref{app:audit}, and remain a factor 2--6 below the 20--45\% failures of the impute-then-test pipelines. Calibration is therefore robust across the learner classes the assumptions cover, degradation outside them is gradual, and the binding cost of a weak learner is power. Second, imputation consistency (A1) is strong, and under self-masking MNAR it is false by construction; the imputer-accuracy and residual-correlation measurements of the next subsection make the distance between A1 and the data a measured quantity in each regime.

\begin{table}[h]
    \centering
    \caption{\textbf{Assumptions in practice.}}
    \label{tab:assumptions_practice}
    \small
    \begin{tabular}{p{2.6cm}p{3.3cm}p{3.3cm}p{3.3cm}}
        \toprule
        Assumption & Implementation counterpart & Diagnostic & Failure behavior \\
        \midrule
        (A1) imputation consistency & MICE with Bayesian ridge and posterior draws; test variables as auxiliary predictors & mean-imputation $R^2$, per-draw NRMSE, $\mathrm{corr}(\delta, Y)$, $\mathrm{corr}(\delta, Z)$ per cell (Appendix~\ref{app:imputer_quality}) & size and power for the scientific null rest on the empirical tier; inflation governed by the $\hat\kappa$ product (Appendix~\ref{app:adversarial_kappa}) \\
        (A1$'$) common-limit imputation & the $M$ chained-equation imputers share the model and the data & between-imputation variance $B$ small relative to $\bar W$ & Rubin pooling of divergent imputations; not observed in any cell \\
        (A2) bounded loss & cross-entropy with clipped probabilities; squared error on standardized $Z$ & --- & heavy-tailed loss differences inflate $\bar W$ (conservative direction) \\
        (A3) proper imputation & \texttt{sample\_posterior=True} & $B > 0$ across imputations & improper draws understate $B$ \\
        (A4) learner stability & random forests (general variant), ExtraTrees with early stopping (fast variant) & fold-to-fold variance of $\hat\mu_{km}$; the learner sweep of Table~\ref{tab:learner_robustness} & mild elevation on linear-Gaussian DGPs for unstable learners (fixed-$k$ NN 0.108); a factor 2--6 below impute-then-test failures \\
        (A5) learner consistency & forests with subsampling & power curves (Appendix~\ref{app:power_plots}) & a weak learner costs power, not size (stumps: FPR 0.022, power 0.07--0.12) \\
        (A6) faithfulness & standard for PC & --- & standard PC failure modes \\
        \bottomrule
    \end{tabular}
\end{table}

\begin{table}[h]
    \centering
    \caption{\textbf{Learner robustness.} False positive rate at signal $0$ by mechanism and by DGP, and power at signal $0.3$ by mechanism, for each learner; 36,000 tests in total. Ridge regression is misspecified on the post-nonlinear DGP by construction; fixed-$k$ nearest neighbors sits outside the stability assumption (A4). Values exceeding $\alpha = 0.05$ in \textcolor{red}{red}.}
    \label{tab:learner_robustness}
    \resizebox{\textwidth}{!}{\begin{tabular}{lcccccccc}
    \toprule
     & \multicolumn{3}{c}{FPR by mechanism} & \multicolumn{2}{c}{FPR by DGP} & \multicolumn{3}{c}{Power (0.3) by mechanism} \\
    \cmidrule(lr){2-4} \cmidrule(lr){5-6} \cmidrule(lr){7-9}
    Learner & Complete & MAR & MNAR & Lin.\ G. & Post-nl & Complete & MAR & MNAR \\
    \midrule
    RF, 100 trees, leaf 5 (default)      & 0.019 & 0.013 & \textcolor{red}{0.066} & \textcolor{red}{0.052} & 0.013 & 0.50 & 0.49 & 0.53 \\
    RF, 25 trees, depth 3                & 0.046 & 0.039 & \textcolor{red}{0.064} & \textcolor{red}{0.078} & 0.021 & 0.17 & 0.20 & 0.18 \\
    RF, 10 stumps (depth 1)              & 0.016 & 0.018 & 0.031 & 0.034 & 0.009 & 0.09 & 0.09 & 0.11 \\
    Ridge regression (linear)            & 0.000 & 0.001 & \textcolor{red}{0.081} & 0.028 & 0.028 & 1.00 & 1.00 & 0.96 \\
    $k$NN regression, $k{=}10$ (fixed)   & \textcolor{red}{0.066} & 0.049 & \textcolor{red}{0.106} & \textcolor{red}{0.108} & 0.039 & 0.61 & 0.63 & 0.63 \\
    \bottomrule
\end{tabular}
}
\end{table}

\subsection{Imputation Quality}
\label{app:imputer_quality}

Table~\ref{tab:imputer_quality} measures the imputer directly on the calibration grid, per DGP $\times$ mechanism $\times$ rate $\times$ $n$ cell (pooled over $|\Xb|$ and signal): the $R^2$ of the across-$M$ mean imputation on the missing entries, benchmarked against the missing subpopulation's own variance; the per-draw NRMSE against the marginal standard deviation, which exceeds 1 by design because proper imputation adds posterior noise, trading accuracy for distributional correctness; and the residual correlations $\mathrm{corr}(\delta, Y)$ and $\mathrm{corr}(\delta, Z)$ of the imputation error with the test variables, the direct empirical analogues of Definition~\ref{def:kappa_factors}'s factors. Under MAR, the mean imputation tracks the truth about as well as the tail-selected missing subpopulation allows (mechanism-level median $R^2$ of 0.19 against that benchmark; individual cells vary). Under self-masking MNAR the imputer is, as Proposition~\ref{prop:impossibility} says it must be, wrong in absolute terms, with mechanism-level median $R^2$ of $-2$ to $-3$ (per-DGP medians from $-1.5$ down to $-5.6$), that is, much further from the truth than the missing subpopulation's own mean, because the observed-data conditional differs from the population conditional. Yet PAIR-CI's false positive rate in those same cells is 0.000--0.070 (Table~\ref{tab:calibration_by_dgp}), and the DGP with the worst measured imputer (latent confounder under MNAR, $R^2 \approx -5.6$) has a false positive rate of exactly 0.000, together with power 0.75, 0.98, and 0.99 at signals 0.3, 0.6, and 1.0, so the test is conservative there rather than inert. That cell has the largest residual correlation in the grid, $\mathrm{corr}(\delta, Z) = -0.44$, with $\mathrm{corr}(\delta, Y) \approx 0.00$, which illustrates the product mechanism of Remark~\ref{rem:paired_robust}: bias requires both residual correlations to be non-zero. The paired design absorbs this contamination and remains calibrated under measurably failed imputation, whereas imputed KCI's false positive rate in the same cell is 0.320. Appendix~\ref{app:imputation_degradation} gives the complementary ordering under deliberately degraded imputers (posterior-draw MICE, mean, marginal).

\begin{table}[h]
    \centering
    \caption{\textbf{Imputer accuracy on the calibration grid.} Medians over replicates, $|\Xb|$, and signal. $R^2$: across-$M$ mean imputation on the missing entries against the missing subpopulation's variance. NRMSE: per posterior draw against the marginal standard deviation. Correlations: per-draw residual correlations of $\delta = \hatX_{\mathrm{miss}} - \Xb_{\mathrm{miss}}$ with $Y$ and $Z$.}
    \label{tab:imputer_quality}
    \resizebox{\textwidth}{!}{\begin{tabular}{lllcccccccccc}
    \toprule
     & & & \multicolumn{4}{c}{Mean-imputation $R^2$} & \multicolumn{4}{c}{Per-draw NRMSE} & & \\
    \cmidrule(lr){4-7} \cmidrule(lr){8-11}
    DGP & Mech. & Rate & 500 & 1{,}000 & 2{,}000 & 5{,}000 & 500 & 1{,}000 & 2{,}000 & 5{,}000 & $\mathrm{corr}(\delta,Y)$ & $\mathrm{corr}(\delta,Z)$ \\
    \midrule
    Lin.\ Gaussian   & MAR   & 10\% & 0.43 & 0.42 & 0.42 & 0.40 & 0.97 & 0.98 & 0.98 & 0.99 & +0.002 & -0.000 \\
                     &       & 30\% & 0.39 & 0.38 & 0.41 & 0.43 & 1.01 & 1.01 & 0.99 & 0.97 & +0.000 & -0.001 \\
                     &       & 50\% & 0.36 & 0.39 & 0.37 & 0.39 & 1.02 & 1.01 & 1.02 & 1.01 & +0.000 & -0.000 \\
                     & MNAR  & 10\% & -1.79 & -1.78 & -1.70 & -1.68 & 1.24 & 1.22 & 1.21 & 1.20 & +0.000 & +0.002 \\
                     &       & 30\% & -1.67 & -1.44 & -1.47 & -1.50 & 1.24 & 1.20 & 1.20 & 1.21 & +0.001 & +0.001 \\
                     &       & 50\% & -1.46 & -1.33 & -1.37 & -1.40 & 1.25 & 1.23 & 1.23 & 1.23 & +0.000 & -0.001 \\
    \midrule
    Post-nonlinear   & MAR   & 10\% & 0.28 & 0.29 & 0.29 & 0.30 & 1.08 & 1.07 & 1.07 & 1.07 & -0.002 & +0.002 \\
                     &       & 30\% & 0.29 & 0.29 & 0.30 & 0.32 & 1.08 & 1.08 & 1.07 & 1.06 & +0.000 & +0.001 \\
                     &       & 50\% & 0.28 & 0.29 & 0.31 & 0.31 & 1.10 & 1.08 & 1.08 & 1.07 & +0.000 & -0.000 \\
                     & MNAR  & 10\% & -2.43 & -2.35 & -2.37 & -2.27 & 1.34 & 1.33 & 1.34 & 1.33 & +0.001 & +0.002 \\
                     &       & 30\% & -1.97 & -1.96 & -1.80 & -1.92 & 1.31 & 1.30 & 1.27 & 1.30 & +0.001 & -0.003 \\
                     &       & 50\% & -1.84 & -1.76 & -1.71 & -1.63 & 1.34 & 1.32 & 1.31 & 1.29 & +0.001 & -0.001 \\
    \midrule
    Latent conf.     & MAR   & 10\% & -0.21 & -0.18 & -0.16 & -0.14 & 1.42 & 1.41 & 1.40 & 1.39 & -0.001 & +0.038 \\
                     &       & 30\% & -0.19 & -0.16 & -0.15 & -0.14 & 1.41 & 1.39 & 1.39 & 1.38 & -0.002 & +0.028 \\
                     &       & 50\% & -0.18 & -0.16 & -0.15 & -0.14 & 1.40 & 1.39 & 1.38 & 1.38 & -0.001 & +0.023 \\
                     & MNAR  & 10\% & -8.40 & -8.31 & -8.40 & -8.45 & 2.10 & 2.12 & 2.13 & 2.13 & -0.001 & -0.441 \\
                     &       & 30\% & -5.56 & -5.58 & -5.57 & -5.55 & 1.85 & 1.85 & 1.85 & 1.85 & -0.001 & -0.443 \\
                     &       & 50\% & -4.18 & -4.16 & -4.12 & -4.11 & 1.74 & 1.73 & 1.73 & 1.72 & -0.001 & -0.449 \\
    \bottomrule
\end{tabular}
}
\end{table}

\section{Operating Points and Design Choices}
\label{app:design_extra}

\subsection{Precision--Recall Frontiers and Sparsity-Matched Comparison}
\label{app:pr_frontier}

PAIR-CI operates at a sparse, high-precision point (Section~\ref{sec:exp_graph}). To test whether its graph-level gains are driven by conservative pruning, we sweep each $\alpha$-tunable baseline (test-wise deletion, complete-case analysis, FZ-vote, FZ-Rubin) over $\alpha \in \{0.001, 0.01, 0.05, 0.1, 0.15, 0.2, 0.3\}$ at $p = 10$ and on ALARM, and over $\{0.01, 0.05, 0.1, 0.2\}$ on HAILFINDER, recording median precision, recall, F1, and output edges at each $\alpha$; MVPC has no $\alpha$ and sits at its native operating point. Figure~\ref{fig:pr_frontier} draws the resulting frontiers and Table~\ref{tab:sparsity_matched} compares each baseline at the $\alpha$ whose median edge count is closest to PAIR-CI's at $\alpha = 0.05$. On ALARM with nonlinear edges at 20\% MNAR, PAIR-CI outputs a median 33.5 edges with precision, recall, and F1 of 0.94, 0.67, and 0.80; tuning each baseline to that sparsity (33--35 edges) gives F1 0.68 (test-wise deletion), 0.68 (FZ-Rubin), and 0.49 (complete-case), a 12-point gap at identical edge counts, and no swept $\alpha$ reaches PAIR-CI's F1 (best 0.68; MVPC 0.65 at its native operating point). On HAILFINDER the best swept F1 is 0.54 (test-wise deletion, 71 edges) and MVPC's is 0.57 (63 edges), against PAIR-CI's 0.59 at 42 edges; no tunable baseline reaches PAIR-CI's sparsity, since test-wise deletion bottoms out at 71 edges and complete-case analysis collapses to the empty graph. At $p = 10$ the frontiers largely coincide at the sparse end (PAIR-CI precision 1.00 and recall 0.56 against test-wise deletion's 1.00 and 0.57 under MNAR).

Table~\ref{tab:bootstrap_contrasts} reports paired, cluster-bootstrap 95\% intervals for the headline SHD contrasts (PAIR-CI minus baseline; negative favors PAIR-CI), pooled over the MAR, MNAR, and mixed mechanisms, resampling graphs on the Erd\H{o}s--R\'enyi grids and replicates on the network topologies (10{,}000 resamples). The $p = 10$ advantage over test-wise deletion is significant but small ($-0.74$, 95\% CI $[-1.04, -0.42]$, linear; $-0.75$, $[-1.24, -0.27]$, nonlinear); the substantive gaps are at $p \geq 20$ with nonlinear edges and on the networks, and MVPC is ahead of PAIR-CI on linear structure at $p \geq 20$ and on ALARM with linear edges.

\begin{figure}[h]
    \centering
    \includegraphics[width=\textwidth]{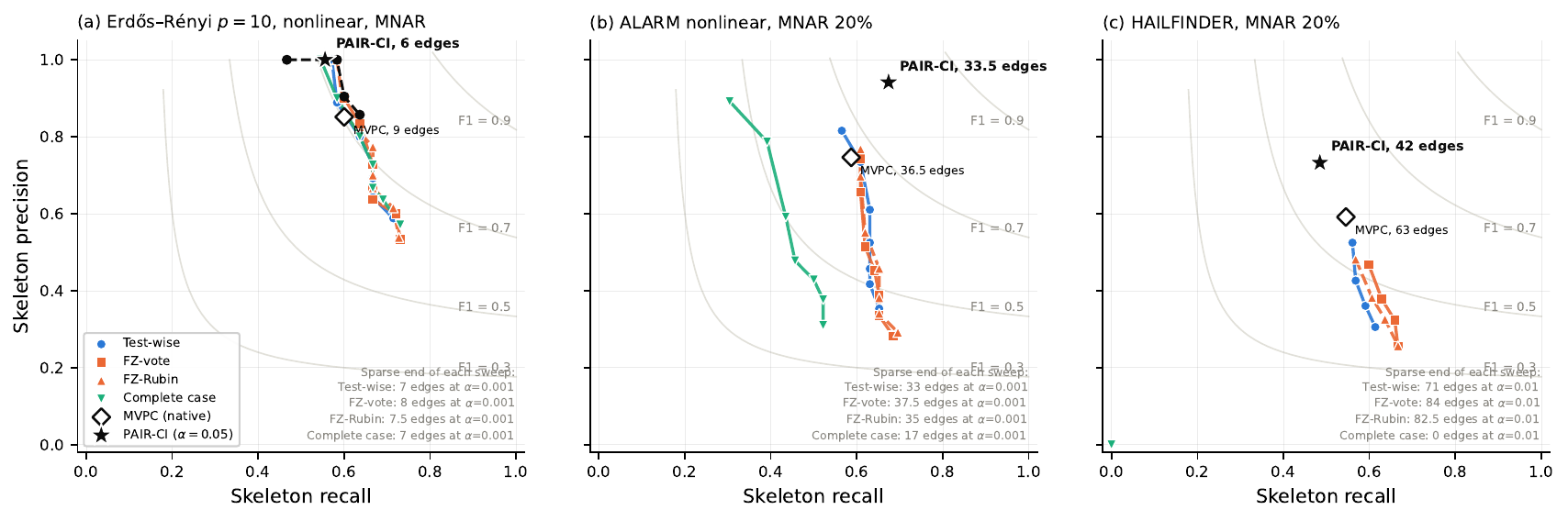}
    \caption{\textbf{Precision--recall operating points.} Median skeleton precision against recall for the $\alpha$-tunable baselines along their $\alpha$ sweeps, PAIR-CI at $\alpha = 0.05$ (and along its own sweep at $p = 10$), and MVPC at its native operating point; gray curves are iso-F1 contours. (a) Erd\H{o}s--R\'enyi $p = 10$, nonlinear edges, MNAR; (b) ALARM, nonlinear edges, 20\% MNAR; (c) HAILFINDER, 20\% MNAR.}
    \label{fig:pr_frontier}
\end{figure}

\begin{table}[h]
    \centering
    \caption{\textbf{Sparsity-matched comparison at 20\% MNAR.} Each tunable baseline at the swept $\alpha$ whose median output edge count is closest to PAIR-CI's, with its best swept F1 and the $\alpha$ attaining it; MVPC at its native operating point.}
    \label{tab:sparsity_matched}
    \resizebox{\textwidth}{!}{\begin{tabular}{llccccc|c}
    \toprule
     & & \multicolumn{5}{c}{At PAIR-CI's sparsity} & Best swept \\
    \cmidrule(lr){3-7}
    Benchmark & Method & $\alpha$ & Edges & Precision & Recall & F1 & F1 ($\alpha$) \\
    \midrule
    ALARM (nonlinear), MNAR 20\%     & PAIR-CI ($\alpha{=}0.05$)  & 0.05 & 33.5 & 0.94 & 0.67 & 0.80 & -- \\
                                     & MVPC (native)              & 0.05 & 36.5 & 0.75 & 0.59 & 0.65 & native \\
                                     & Test-wise                  & 0.001 & 33 & 0.82 & 0.57 & 0.68 & 0.68 (0.001) \\
                                     & Complete case              & 0.05 & 35 & 0.59 & 0.43 & 0.49 & 0.52 (0.01) \\
                                     & FZ-vote                    & 0.001 & 37.5 & 0.74 & 0.61 & 0.67 & 0.67 (0.001) \\
                                     & FZ-Rubin                   & 0.001 & 35 & 0.77 & 0.61 & 0.68 & 0.68 (0.001) \\
    \midrule
    ALARM (linear), MNAR 20\%        & PAIR-CI ($\alpha{=}0.05$)  & 0.05 & 44.5 & 0.91 & 0.88 & 0.89 & -- \\
                                     & MVPC (native)              & 0.05 & 44 & 0.95 & 0.91 & 0.93 & native \\
                                     & Test-wise                  & 0.001 & 45 & 0.91 & 0.88 & 0.89 & 0.89 (0.001) \\
                                     & Complete case              & 0.05 & 45.5 & 0.78 & 0.76 & 0.78 & 0.78 (0.05) \\
                                     & FZ-vote                    & 0.001 & 46.5 & 0.89 & 0.91 & 0.90 & 0.90 (0.001) \\
                                     & FZ-Rubin                   & 0.001 & 46 & 0.91 & 0.89 & 0.90 & 0.90 (0.001) \\
    \midrule
    HAILFINDER, MNAR 20\%            & PAIR-CI ($\alpha{=}0.05$)  & 0.05 & 42.0 & 0.73 & 0.48 & 0.59 & -- \\
                                     & MVPC (native)              & 0.05 & 63 & 0.59 & 0.55 & 0.57 & native \\
                                     & Test-wise                  & 0.01 & 71 & 0.52 & 0.56 & 0.54 & 0.54 (0.01) \\
                                     & Complete case              & 0.01 & 0 & 0.00 & 0.00 & 0.00 & 0.00 (0.01) \\
                                     & FZ-vote                    & 0.01 & 84 & 0.47 & 0.60 & 0.53 & 0.53 (0.01) \\
                                     & FZ-Rubin                   & 0.01 & 82.5 & 0.48 & 0.57 & 0.53 & 0.53 (0.01) \\
    \bottomrule
\end{tabular}
}
\end{table}

\begin{table}[h]
    \centering
    \caption{\textbf{Paired bootstrap contrasts.} Mean paired difference in total SHD, PAIR-CI minus baseline (negative favors PAIR-CI), with percentile 95\% intervals from 10{,}000 cluster resamples (graphs on the Erd\H{o}s--R\'enyi grids; replicates on the networks), pooled over MAR, MNAR, and mixed mechanisms. Intervals covering zero in gray. $^\ast$MVPC on HAILFINDER is paired by replicate index only (datasets matched in distribution); the same holds for FZ-Rubin on HAILFINDER.}
    \label{tab:bootstrap_contrasts}
    \resizebox{\textwidth}{!}{\begin{tabular}{llccccc}
    \toprule
    Benchmark & Metric & Test-wise & Complete case & FZ-vote & FZ-Rubin & MVPC \\
    \midrule
    ER linear p=10             & SHD & -0.74 [-1.04, -0.42] & -0.87 [-1.36, -0.47] & -1.05 [-1.56, -0.59] & -0.84 [-1.24, -0.47] & \textcolor{gray}{+0.10 [-0.19, +0.38]} \\
    \midrule
    ER linear p=20             & SHD & -1.45 [-1.80, -1.09] & -3.41 [-4.65, -2.33] & -2.21 [-2.53, -1.84] & -1.88 [-2.21, -1.51] & +2.38 [+1.73, +3.04] \\
    \midrule
    ER linear p=30             & SHD & -4.12 [-4.71, -3.52] & -7.99 [-9.19, -6.92] & -6.22 [-7.21, -5.25] & -5.83 [-6.84, -4.86] & +5.16 [+4.37, +5.90] \\
    \midrule
    ER nonlinear p=10          & SHD & -0.75 [-1.24, -0.27] & -0.76 [-1.10, -0.39] & -1.71 [-2.19, -1.22] & -1.26 [-1.74, -0.79] & \textcolor{gray}{-0.11 [-0.51, +0.30]} \\
    \midrule
    ER nonlinear p=20          & SHD & -4.71 [-5.25, -4.13] & -4.70 [-5.13, -4.28] & -8.40 [-8.83, -7.95] & -7.01 [-7.56, -6.41] & -0.53 [-0.96, -0.02] \\
    \midrule
    ER nonlinear p=30          & SHD & -9.28 [-9.85, -8.70] & -9.26 [-9.97, -8.53] & -16.23 [-16.82, -15.62] & -13.84 [-14.57, -13.05] & -1.04 [-1.63, -0.44] \\
    \midrule
    ALARM (linear) 20\%        & SHD & -8.70 [-10.10, -7.33] & -13.37 [-14.88, -11.73] & -9.92 [-11.17, -8.72] & -9.55 [-10.78, -8.20] & +4.47 [+3.03, +5.85] \\
    \midrule
    ALARM (nonlinear) 20\%     & SHD & -17.38 [-19.00, -15.80] & -17.68 [-19.32, -15.90] & -24.83 [-26.63, -22.92] & -21.95 [-23.67, -20.32] & -8.12 [-9.45, -6.90] \\
    \midrule
    HAILFINDER 20\%            & SHD & -43.13 [-44.75, -41.47] & -18.68 [-20.70, -16.58] & -60.82 [-62.58, -59.05] & -56.55 [-60.30, -53.23] & -15.05 [-17.48, -12.83]$^\ast$ \\
    \bottomrule
\end{tabular}
}
\end{table}

\subsection{Sensitivity to Design Choices}
\label{app:sensitivity_summary}

Table~\ref{tab:sensitivity_summary} consolidates the sensitivity analyses. Increasing $M$ from 3 to 10 moves power at signal 0.3 by at most 6 percentage points (the test is learner-bound, not pooling-bound; Appendix~\ref{app:m_sensitivity}); $K \in \{3, 5, 10\}$ leaves median graph-recovery SHD unchanged at 10 ($p = 10$, nonlinear edges, MAR), and the graph-recovery grid varying $M$, $K$, the number of trees, and the leaf size keeps SHD within $\pm 2$ across all 11 configurations (Appendix~\ref{app:sensitivity}); random forests, LightGBM (default and tuned), and XGBoost are calibrated standalone with comparable power, random forests best or tied at weak signals, whereas ExtraTrees runs slightly hot under MNAR (Appendix~\ref{app:learner_comparison}); and the $\alpha$ sweep over $\{0.01, \ldots, 0.2\}$ moves the operating point smoothly with no rank reversals. Nothing depends on a narrow hyperparameter window.

\begin{table}[h]
    \centering
    \caption{\textbf{Consolidated sensitivity summary.} Each row varies one design choice with the rest held at the defaults ($M = 5$, $K = 5$ or $10$, 100-tree random forests, $\alpha = 0.05$).}
    \label{tab:sensitivity_summary}
    \resizebox{\textwidth}{!}{\begin{tabular}{lllll}
    \toprule
    Factor & Levels & Metric & Values & Detail \\
    \midrule
    $M$ (imputations) & 3 / 5 / 10 / 20 & power at 0.3, Complete & 0.26 / 0.27 / 0.26 / 0.24 & App.~\ref{app:m_sensitivity} \\
    $M$ (imputations) & 3 / 5 / 10 / 20 & power at 0.3, MAR & 0.25 / 0.29 / 0.29 / 0.32 & App.~\ref{app:m_sensitivity} \\
    $M$ (imputations) & 3 / 5 / 10 / 20 & power at 0.3, MNAR & 0.34 / 0.33 / 0.36 / 0.35 & App.~\ref{app:m_sensitivity} \\
    $K$ (folds) & 3 / 5 / 10 & median SHD, $p{=}10$ nonlinear MAR & 10.0 / 10.0 / 10.0 & App.~\ref{app:sensitivity} \\
    Joint grid ($M, K$, trees, leaf) & 11 configurations & median SHD range & 10.0--11.0 & App.~\ref{app:sensitivity} \\
    Learner & RF / ExtraTrees / LightGBM / LightGBM (tuned) / XGBoost & FPR, MNAR & 0.06 / 0.09 / 0.01 / 0.01 / 0.01 & App.~\ref{app:learner_comparison} \\
     &  & power at 0.3, MNAR & 0.35 / 0.45 / 0.23 / 0.21 / 0.23 &  \\
    $\alpha$ & 0.01 / 0.05 / 0.1 / 0.15 / 0.2 & median F1, $p{=}10$ MAR & 0.62 / 0.71 / 0.71 / 0.72 / 0.71 & App.~\ref{app:pr_frontier} \\
    $\alpha$ & 0.01 / 0.05 / 0.1 / 0.15 / 0.2 & median F1, $p{=}10$ MNAR & 0.62 / 0.71 / 0.72 / 0.72 / 0.71 & App.~\ref{app:pr_frontier} \\
    \bottomrule
\end{tabular}
}
\end{table}

\subsection{Missingness in the Test Variables}
\label{app:testvar}

The standalone design of Section~\ref{sec:exp_ci} confines missingness to the conditioning set. Inside PC, the test variables $Z$ and $Y$ are themselves incomplete in some queries, and PAIR-CI imputes their missing entries from the same cached MICE model rather than deleting those rows (Section~\ref{sec:computation}). Table~\ref{tab:testvar_missingness} evaluates this rule against the alternative of scoring only rows with observed labels, with 30\% missingness placed on $Z$, on $Y$, on both, or on $\Xb$ and $Z$ jointly, under MAR and MNAR, on the linear-Gaussian and post-nonlinear DGPs ($n = 1{,}000$, $|\Xb| = 5$, 250 replicates per cell). The deployed rule keeps the false positive rate between 0.000 and 0.068 across the 16 null cells, the largest value arising for linear-Gaussian MNAR on $\Xb$ and $Z$ jointly, and has higher power than the labels-observed rule (0.35 against 0.28 at signal 0.3).

\begin{table}[h]
    \centering
    \caption{\textbf{Missingness on the test variables.} False positive rate (signal $0$) and power (signal $0.3$) with 30\% missingness on the indicated variables, under the deployed imputed-labels rule and the labels-observed alternative. Values exceeding $\alpha = 0.05$ in \textcolor{red}{red}.}
    \label{tab:testvar_missingness}
    \resizebox{\textwidth}{!}{\begin{tabular}{lllcccc}
    \toprule
     & & & \multicolumn{2}{c}{Imputed labels (deployed)} & \multicolumn{2}{c}{Labels observed} \\
    \cmidrule(lr){4-5} \cmidrule(lr){6-7}
    DGP & Mechanism & Missing on & FPR & Power (0.3) & FPR & Power (0.3) \\
    \midrule
    Linear Gaussian  & MAR   & $Z$            & 0.036 & 0.53 & 0.024 & 0.43 \\
                     &       & $Y$            & 0.024 & 0.50 & 0.024 & 0.50 \\
                     &       & $Y$ and $Z$    & 0.020 & 0.52 & 0.012 & 0.38 \\
                     &       & $\Xb$ and $Z$  & 0.032 & 0.48 & 0.008 & 0.41 \\
                     & MNAR  & $Z$            & 0.036 & 0.54 & 0.016 & 0.42 \\
                     &       & $Y$            & 0.032 & 0.52 & 0.012 & 0.39 \\
                     &       & $Y$ and $Z$    & 0.012 & 0.39 & 0.008 & 0.29 \\
                     &       & $\Xb$ and $Z$  & \textcolor{red}{0.068} & 0.53 & 0.020 & 0.44 \\
    \midrule
    Post-nonlinear   & MAR   & $Z$            & 0.000 & 0.23 & 0.000 & 0.12 \\
                     &       & $Y$            & 0.000 & 0.21 & 0.004 & 0.20 \\
                     &       & $Y$ and $Z$    & 0.000 & 0.17 & 0.004 & 0.10 \\
                     &       & $\Xb$ and $Z$  & 0.004 & 0.22 & 0.004 & 0.14 \\
                     & MNAR  & $Z$            & 0.004 & 0.19 & 0.004 & 0.18 \\
                     &       & $Y$            & 0.008 & 0.13 & 0.004 & 0.12 \\
                     &       & $Y$ and $Z$    & 0.000 & 0.12 & 0.000 & 0.07 \\
                     &       & $\Xb$ and $Z$  & 0.020 & 0.36 & 0.024 & 0.23 \\
    \bottomrule
\end{tabular}
}
\end{table}

\section{Worked Example for Proposition~\ref{prop:impossibility}}
\label{app:worked_example}

We derive condition~(b) of Proposition~\ref{prop:impossibility} from primitives in a fully specified model, so that the induced partial correlation is a closed-form quantity rather than an empirical observation. The imputer in the example regresses the incomplete variable on the test variables, which is the auxiliary-predictor design of the impute-then-test pipelines in Section~\ref{sec:exp_ci}, so the example is aimed at those pipelines rather than at a per-query imputer.

\paragraph{Model.} Let $X \sim \mathcal{N}(0,1)$, $Y = 0.8X + \varepsilon_Y$ and $Z = 0.8X + \varepsilon_Z$ with $\varepsilon_Y, \varepsilon_Z \sim \mathcal{N}(0, 0.36)$ independent, so that $Z \indep Y \mid X$. Entries of $X$ are missing under self-masking probit selection, $P(X \text{ missing} \mid X = x) = \Phi(s(x - t))$ with steepness $s = 5$ and threshold $t$ the $(1 - \rho_{\mathrm{miss}})$-quantile of $X$, so that the marginal missingness rate is approximately $\rho_{\mathrm{miss}}$ (0.304 at $\rho_{\mathrm{miss}} = 0.3$). Missing entries are imputed by a least-squares regression of $X$ on $(1, Y, Z)$ fitted on the observed cases, the population limit of the Bayesian-ridge conditional mean at negligible penalty, either as the conditional mean or with posterior-draw noise added; $\tilde X$ denotes the completed column, which equals $X$ on the observed rows.

\paragraph{Derivation.} All steps are in closed form; the released script evaluates the selection moments by quadrature. (i)~Selection acts on $X$ only, so the first two moments of $X$ on the observed and missing subpopulations follow from the inverse-Mills identity: for $W \sim \mathcal{N}(0,1)$ weighted by $\Phi(a + bW)$, the weighted mean is $(b/\sqrt{1 + b^2})\,\varphi(c)/\Phi(c)$ with $c = a/\sqrt{1 + b^2}$, and the weighted second moment follows likewise. (ii)~Because selection acts on $X$ only, all $(X, Y, Z)$ subpopulation moments follow from those of $X$. (iii)~The imputer's coefficients are the observed-subpopulation least-squares coefficients of $X$ on $(1, Y, Z)$. (iv)~The pooled second moments of $\tilde X$ over the observed rows and the missing rows (where $\tilde X$ is the fitted value, plus noise for the draw variant) give $\mathrm{corr}(Z, Y \mid \tilde X)$ by the standard partial-correlation formula.

\paragraph{Result.} At $\rho_{\mathrm{miss}} = 0.3$, $\mathrm{corr}(Z, Y \mid \tilde X) = -0.0675$ for conditional-mean imputation and $+0.0730$ with posterior-draw noise (Monte Carlo over 20 independent replicates of $n = 2 \times 10^5$: $-0.0681 \pm 0.0005$ and $+0.0715 \pm 0.0006$). The sign difference is itself instructive: conditional-mean imputation makes the completed column, on the missing rows, exactly a positive linear index of $(Y, Z)$, so conditioning on it over-adjusts and induces a negative partial correlation, whereas draw noise degrades the proxy for $X$ and under-adjusts, leaving part of the positive $X$-induced association. Both variants violate condition~(b), and because the observed rows contribute $\tilde X = X$ exactly, the whole effect comes from the missing rows and its magnitude grows with the missingness rate ($-0.011$, $-0.068$, $-0.159$ at 10\%, 30\%, and 50\% for the conditional-mean variant). At $n = 2{,}000$ the conditional-mean variant gives a Fisher-$Z$ statistic of about 3.0 and a two-sided rejection probability of about 0.86 at $\alpha = 0.05$, tending to one as $n$ grows. The released code includes the script that evaluates the closed form and checks it against simulation.

\section{Changes to the Released Code and Results}
\label{app:changelog}

While preparing the mechanism statement of Section~\ref{sec:exp_ci} we audited the missingness generator of the standalone grid and found a defect in one stratum: at $|\Xb| = 2$ under MAR, no fully observed driver was available and the code degenerated to one-column self-masking, that is, MNAR mislabelled as MAR. We corrected the generator so that at $|\Xb| = 2$ one column is incomplete and the other serves as the MAR driver (cells with $|\Xb| \geq 5$ are bit-identical under the corrected code, and the graph experiments, which use external drivers, were never affected), re-ran the affected cells (100 replicates per cell; 50 for KCI), and merged the results into the released result files. The corrected MAR column strengthens the paper's diagnosis: under MAR with a proper imputer, Rubin-pooled Fisher's $Z$ is near-nominal (0.058, previously 0.113), so the calibration failure we target concentrates where Proposition~\ref{prop:impossibility} places it; PAIR-CI's MAR figure is unchanged (0.021). Because the corrected generator consumes the random-number stream differently in that stratum, a full single-seed regeneration of the calibration grid is Monte-Carlo-equivalent to, but not bit-identical with, the released files downstream of the affected configurations. The reproduction script includes the experiments introduced in this revision (MVPC, the baseline $\alpha$ sweeps, the GCM/KCI audit, the imputer-quality, learner-robustness, variance-decomposition, and test-variable-missingness experiments). MVPC's correction step draws from an unseeded random-number generator in its authors' implementation, so MVPC's outputs are reproducible in distribution but not bit for bit. The paired bootstrap intervals of Section~\ref{sec:exp_graph} were recomputed for this version and differ from those quoted during the discussion period by at most 0.02 at either endpoint.

Three corrections to the theory of the submitted version: Theorem~\ref{thm:bridge} builds on Theorems~1 and~2 of \citet{robins2000inference}, not the ``Theorem~4.1'' the submitted version cited; the size result is stated as a bound, $\leq \alpha$, rather than as exact size; and PC consistency (Corollary~\ref{cor:pc_consistency}) is stated for a sequence of levels $\alpha_n \to 0$, since at a fixed level spurious edges persist with non-vanishing probability; Corollary~\ref{cor:skeleton_fixed_alpha} (Appendix~\ref{app:skeleton_fixed_alpha}) states what a fixed level delivers for the skeleton.

\section{Notation}
\label{app:notation}

\begin{table}[h]
    \centering
    \caption{\textbf{Notation.} Symbols in order of first use. Where a symbol carries a second meaning inherited from the submitted version, both are listed.}
    \label{tab:notation}
    \small
    \begin{tabular}{>{\raggedright\arraybackslash}p{3.6cm}>{\raggedright\arraybackslash}p{9.6cm}}
        \toprule
        Symbol & Meaning \\
        \midrule
        $Z$, $Y$, $\Xb$ & response, candidate variable, and conditioning set of the test $H_0\colon Z \indep Y \mid \Xb$ \\
        $R$ & response indicator, $R_{ij} = 1$ iff $X_{ij}$ is observed; $\Xb = (\Xb_{\mathrm{obs}}, \Xb_{\mathrm{miss}})$ \\
        $\hatX$, $\hat{D}^{(m)}$ & completed conditioning set; $m$-th of $M$ multiply imputed datasets \\
        $\varepsilon$; $\varepsilon_Y$, $\varepsilon_Z$ & lower bound on the total-variation misspecification in Proposition~\ref{prop:impossibility}; noise terms in Appendix~\ref{app:worked_example} \\
        $\delta = \Xb_{\mathrm{miss}} - \hatX_{\mathrm{miss}}$ & imputation residual \\
        $U$, $U_k$ & placebo for $Y$ obtained by conditional permutation (fold-specific) \\
        $k_{\mathrm{nn}}$, $d = |\Xb|$, $\varepsilon_n$ & neighborhood size and dimension of the local permutation; total-variation distance between the joint law of the placebo sample and the exactly exchangeable one \\
        $F_k$, $T_{km}$, $n_k$ & held-out fold $k$, training set for fold $k$ of imputation $m$, fold size \\
        $\hat g^{ZY}_{km}$, $\hat g^{ZU}_{km}$ & full and partial models trained on $T_{km}$; $g^{\mathrm{fu}}_\psi$, $g^{\mathrm{pl}}_\psi$ their population counterparts at imputation parameter $\psi$ \\
        $\ell$, $L_\ell$ & loss (cross-entropy for $Z$ with $\leq 20$ unique values, squared error otherwise) and its bound (A2) \\
        $d_{ikm}$, $\hat\mu_{km}$, $\hat\mu_m$, $\hat\mu$, $\mu_n$ & individual, fold-level, per-imputation, and overall loss differences; $\E[\hat\mu_m]$ \\
        $\hat\sigma^2_{k,m}$, $\hat\sigma^2_{m,\mathrm{in}}$, $V_m^{\mathrm{Bayle}}$ & within-fold variance, its fold average, and the within-imputation variance estimate $\hat\sigma^2_{m,\mathrm{in}}/n$ \\
        $\bar W$, $B$, $T$ & within-imputation variance, between-imputation variance, combined variance $\bar W + (1 + 1/M)B$ (unrelated to the observation $W$ below) \\
        $\tCI$, $\nu$, $\nu_M$, $\hat\nu_{\mathrm{obs}}$, $\nu_{\mathrm{com}}$, $\hat\gamma$ & test statistic; Barnard--Rubin, Rubin, observed-data and complete-data degrees of freedom; estimated fraction of missing information \\
        $t_{\alpha,\nu}$, $z_\alpha$ & $t$ and standard normal upper $\alpha$-quantiles \\
        $W$, $W_{\mathrm{obs}}$ & an observation and its observed part \\
        $\bar h_m$, $\bar h^*$, $\sigma^2_m$, $\sigma^2_*$ & per-imputation and limiting influence functions and their variances \\
        $\Sigma_M$, $\Sigma_\infty$, $\Sigma_{\mathrm{Rubin}}$ & limiting variance of $\sqrt{n}(\hat\mu - \mu_n)$ at fixed $M$ and as $M \to \infty$; probability limit of $nT$ \\
        $\psi$, $\psi^*$, $S_{\mathrm{mis}}$, $\mathfrak{F}$, $\mathcal{U}$ & imputation parameter and its limit, imputation score, $\sigma(\Xb_{\mathrm{obs}}, \hatX, R)$, the estimating function $\bar h^*(W) - \beta$ \\
        $\dot\lambda = \kappa_{\mathrm{score}} + \dot\lambda_{\mathrm{train}}$, $\Delta_{\mathrm{cache}}$ & sensitivity of the estimating function to $\psi$ (score and learner-drift channels); the score channel under caching \\
        $\sigma^2_{\mathrm{obs}}$, $\sigma^2_{\mathrm{draw}}$, $\eta_m$, $v$ & sampling and draw-to-draw components of $\mathrm{Var}(\hat\mu_m)$; draw noise of imputation $m$; $\mathrm{Var}(\bar h^*(W) \mid R, W_{\mathrm{obs}})$ \\
        $\kappa_{\mathrm{imp}}$, $\kappa_Y$, $\kappa_Z$, $\kappa$ & imputation-error magnitude, residual correlations, and residual bias of the loss difference (Definition~\ref{def:kappa_factors}) \\
        $\Delta^*$, $G_{\mathrm{CPDAG}}$, $\hat G$ & Bayes-risk gap in Proposition~\ref{prop:consistency}; true and estimated CPDAG \\
        $\pi_j(v)$, $\rho$, $q_{1-\rho}$ & masking probability, target per-column missingness rate, driver quantile (Equation~\ref{eq:mechanism}); $\rho_{\mathrm{miss}}$ in Appendix~\ref{app:worked_example} \\
        \bottomrule
    \end{tabular}
\end{table}

\clearpage
\section*{NeurIPS Paper Checklist}

\begin{enumerate}

\item {\bf Claims}
    \item[] Question: Do the main claims made in the abstract and introduction accurately reflect the paper's contributions and scope?
    \item[] Answer: \answerYes{}
    \item[] Justification: The abstract's headline claims are substantiated in the paper: (i) when misspecified imputation induces spurious conditional dependence, any consistent CI test rejects a true null with probability approaching 1 (Proposition~\ref{prop:impossibility}); (ii) the combined variance estimator gives size at most $\alpha$ in the many-imputation limit, conservative by a derived, sign-definite mechanism (Theorem~\ref{thm:bridge}; at fixed $M$ the bound is examined empirically, Appendix~\ref{app:variance_ratio}); (iii) imputation-based CI tests have false positive rates of 28--45\% under MNAR, whereas PAIR-CI's pooled rate is at most 3.6\% and exceeds 5\% only in cells with $|\Xb| \geq 5$ (Tables~\ref{tab:calibration} and~\ref{tab:calibration_full}); and (iv) on the nonlinear networks PAIR-CI reduces SHD by 18\% (ALARM) and 11\% (HAILFINDER) relative to MVPC and by 31--48\% relative to test-wise deletion and the Fisher-$Z$ pipelines (Table~\ref{tab:realworld}). The abstract and Section~\ref{sec:intro} also state that a well-specified parametric pipeline is preferable when relationships are linear (Tables~\ref{tab:scaling} and~\ref{tab:realworld}), and Section~\ref{sec:theory} separates what is proved from what is measured.

\item {\bf Limitations}
    \item[] Question: Does the paper discuss the limitations of the work performed by the authors?
    \item[] Answer: \answerYes{}
    \item[] Justification: Section~\ref{sec:theory} (``Scope, and assumptions in practice'') states that the size bound is proved only for the internal null in the many-imputation limit, that the fixed-$M$ bound and the MNAR extension rest on measurement, and that PAIR-CI exceeds the nominal level in some cells with $|\Xb| \geq 5$ (Table~\ref{tab:calibration_full}). Section~\ref{sec:discussion} discusses the power cost relative to well-specified parametric tests, including under-rejection at low signal and small $n$, attributed to Rubin's combined variance overstating the variance of $\hat\mu$ (Theorem~\ref{thm:bridge}(i); Appendix~\ref{app:variance_ratio}); it recommends MVPC under linear conditions and restricts the guarantees and experiments to PC under causal sufficiency. Appendix~\ref{app:assumptions_practice} maps each assumption to a diagnostic and a failure behavior (Table~\ref{tab:assumptions_practice}) and reports learner robustness; Appendices~\ref{app:imputation_degradation} and~\ref{app:adversarial_kappa} show that crude imputers (mean, marginal) inflate false positive rates to 29--43\% and that an adversarial nonlinear hub defeats the default linear imputer (87--100\%). Appendix~\ref{app:skeleton_fixed_alpha} notes that the PC corollaries do not cover the simplified search used in the experiments, and Appendix~\ref{app:robustness} reports computational cost (about 100 minutes per replicate at $p = 56$).

\item {\bf Theory assumptions and proofs}
    \item[] Question: For each theoretical result, does the paper provide the full set of assumptions and a complete (and correct) proof?
    \item[] Answer: \answerYes{}
    \item[] Justification: Each result is stated with its assumptions: (A1)--(A6) and (A1$'$) in Section~\ref{sec:theory}, and Assumptions~\ref{ass:cross_imp}--\ref{ass:regularity} and conditions (N1)--(N2) in Appendix~\ref{app:bridge}. Proofs of Propositions~\ref{prop:impossibility}--\ref{prop:consistency} and Corollaries~\ref{cor:pc_consistency} and~\ref{cor:skeleton_fixed_alpha} are in Appendix~\ref{app:proofs}; Lemma~\ref{lem:sensitivity} and Theorem~\ref{thm:bridge} are proved in Appendix~\ref{app:bridge}, which also gives the full statement of the theorem (sketched in Section~\ref{sec:theory}); and Lemma~\ref{lem:within_ral} restates results of \citet{bayle2020cross} under (A2) and (A4). The paper states the limits of these proofs. At fixed $M$, the size bound rests on Barnard--Rubin domination, which is not proved and is examined empirically (Appendix~\ref{app:variance_ratio}). The proof of Theorem~\ref{thm:bridge} applies Theorems~1--2 of \citet{robins2000inference} to a learner-dependent estimating function and to chained-equation imputers as an extension of their argument, with their conditions S1--S4 verified in outline (Remark~\ref{rem:verified}). Corollary~\ref{cor:pc_consistency} is conditional on an $\alpha_n \to 0$ oracle property; conditions (a)--(b) of Corollary~\ref{cor:skeleton_fixed_alpha} are hypotheses; and Decomposition~\ref{dec:kappa} is a qualitative first-order expansion, not a theorem.

    \item {\bf Experimental result reproducibility}
    \item[] Question: Does the paper fully disclose all the information needed to reproduce the main experimental results of the paper to the extent that it affects the main claims and/or conclusions of the paper (regardless of whether the code and data are provided or not)?
    \item[] Answer: \answerYes{}
    \item[] Justification: Section~\ref{sec:experiments} and Appendices~\ref{app:implementation} and~\ref{app:robustness} specify the DGPs, the missingness mechanism (Equation~\ref{eq:mechanism}), sample sizes, replicate counts (100 per standalone cell, 50 for KCI; 10 graphs $\times$ 20 missingness draws per graph-recovery condition; 20 replicates per network cell), hyperparameters (Table~\ref{tab:defaults}) and metrics. Appendix~\ref{app:audit} gives the baseline configurations, and Algorithm~\ref{alg:pc} (Appendix~\ref{app:implementation}) the simplified PC search shared by PAIR-CI and the in-house baselines; MVPC runs the PC-stable search of \texttt{causal-learn} with its authors' correction. The replication code is a SLURM-orchestrated pipeline (\texttt{reproduce\_all.sh}) with one script per experiment and fixed seeds (42 by default). Appendix~\ref{app:changelog} states where reproduction is not bit-exact: MVPC's correction step uses an unseeded random-number generator in its authors' implementation, and the corrected missingness generator makes a single-seed regeneration of the calibration grid Monte-Carlo-equivalent but not bit-identical to the released files.

\item {\bf Open access to data and code}
    \item[] Question: Does the paper provide open access to the data and code, with sufficient instructions to faithfully reproduce the main experimental results, as described in supplemental material?
    \item[] Answer: \answerYes{}
    \item[] Justification: The code is included in the supplementary material as the replication package \texttt{pair-ci\_replication/}. It comprises a master pipeline (\texttt{reproduce\_all.sh}), per-experiment scripts in \texttt{experiments/}, the core library in \texttt{src/}, analysis scripts in \texttt{analysis/}, a pinned \texttt{environment.yml}, and a README documenting setup and execution. The scripts are also available at \url{https://github.com/tsrobinson/pair-ci} (Section~\ref{sec:experiments}). The synthetic DGPs are generated by \texttt{src/dgps.py}, the ALARM and HAILFINDER structures are encoded in the experiment scripts, and the only external data, the Sachs observations, are included as \texttt{data/sachs\_data.csv}. Appendix~\ref{app:changelog} lists the experiments introduced in this revision, which the reproduction script covers.

\item {\bf Experimental setting/details}
    \item[] Question: Does the paper specify all the training and test details (e.g., data splits, hyperparameters, how they were chosen, type of optimizer) necessary to understand the results?
    \item[] Answer: \answerYes{}
    \item[] Justification: Section~\ref{sec:experiments} describes the experimental setup (DGPs, missingness mechanisms, sample sizes, replications, baseline methods, and evaluation metrics). Appendix~\ref{app:implementation} specifies implementation defaults for both PAIR-CI variants (general: random forests with $K=10$; fast: ExtraTrees with $K=5$ and early stopping), the MICE imputer with posterior draws, the neighborhood size of the local permutation, feature bagging, loss selection, and the nonlinear edge functions, with defaults in Table~\ref{tab:defaults}; Appendix~\ref{app:audit} gives the baseline configurations. Appendix~\ref{app:sensitivity} explains how the defaults were chosen, and Appendices~\ref{app:m_sensitivity}, \ref{app:learner_comparison} and~\ref{app:sensitivity_summary} report sensitivity to $M$, $K$, the learner, and $\alpha$. The replication code contains the exact configuration of every reported run.

\item {\bf Experiment statistical significance}
    \item[] Question: Does the paper report error bars suitably and correctly defined or other appropriate information about the statistical significance of the experiments?
    \item[] Answer: \answerYes{}
    \item[] Justification: The headline graph-recovery contrasts carry paired, cluster-bootstrap 95\% percentile intervals (10{,}000 resamples, resampling graphs on the Erd\H{o}s--R\'enyi grids and replicates on the networks, pooled over MAR, MNAR, and mixed), reported in Table~\ref{tab:bootstrap_contrasts} (Appendix~\ref{app:pr_frontier}) and referenced in Section~\ref{sec:exp_graph}. The adversarial calibration cells carry exact Clopper--Pearson 95\% intervals (Appendix~\ref{app:adversarial_kappa}, Figure~\ref{fig:kappa_fpr_trace}). Table~\ref{tab:calibration} reports rejection rates with their replicate counts (100 per cell, 50 for KCI) and states their binomial Monte Carlo standard errors, at most 0.023 for any cell and 0.003 for PAIR-CI's pooled rates, so the headline bound of 3.6\% (pooled over DGPs, $n$ and $|\Xb|$; Section~\ref{sec:exp_ci}) is precise to about $\pm 0.006$; its disaggregations are in Tables~\ref{tab:calibration_by_dgp} and~\ref{tab:calibration_full}. Tables~\ref{tab:scaling} and~\ref{tab:realworld} report medians over 10 graphs $\times$ 20 missingness draws and over 20 replicates, respectively, without error bars; the bootstrap intervals above cover their headline contrasts.

\item {\bf Experiments compute resources}
    \item[] Question: For each experiment, does the paper provide sufficient information on the computer resources (type of compute workers, memory, time of execution) needed to reproduce the experiments?
    \item[] Answer: \answerYes{}
    \item[] Justification: All experiments ran on CPU on the authors' departmental SLURM cluster (Acknowledgments), with 16 CPU cores and 4~GB of memory per job (8~GB for $p = 30$ and HAILFINDER) and no GPU, as documented in the replication README and the SLURM scripts. Appendix~\ref{app:robustness} reports runtimes per replicate (about 3 minutes at $p = 10$, 5 and 10 minutes at $p = 20$ and $30$, and 100 minutes at $p = 56$, at a stable 2.5\,s per CI test), and Appendices~\ref{app:sensitivity} and~\ref{app:learner_comparison} report runtimes by hyperparameter and learner. The README gives per-experiment runtimes and totals of about 2{,}550 CPU-hours and 5 days of wall time for the submitted pipeline. Compute for the experiments added in revision, and for preliminary experiments not reported in the paper, is not documented.

\item {\bf Code of ethics}
    \item[] Question: Does the research conducted in the paper conform, in every respect, with the NeurIPS Code of Ethics \url{https://neurips.cc/public/EthicsGuidelines}?
    \item[] Answer: \answerYes{}
    \item[] Justification: This research is methodological in nature and uses only public benchmark datasets (Sachs protein-signaling, ALARM, HAILFINDER network topologies), with no human subjects, personal data, or sensitive applications. The authors have reviewed the NeurIPS Code of Ethics and identified no conflicts.

\item {\bf Broader impacts}
    \item[] Question: Does the paper discuss both potential positive societal impacts and negative societal impacts of the work performed?
    \item[] Answer: \answerYes{}
    \item[] Justification: The paper identifies common conditions under which applying CI tests to incomplete data can lead causal discovery methods to infer spurious relationships---compromising the reliability of downstream scientific conclusions---and proposes PAIR-CI as a calibrated remedy that reduces this risk. As a contribution to causal discovery methodology with no direct pathway to malicious use, the primary residual risk is overreliance on automated causal discovery without domain validation, which we address explicitly in Section~\ref{sec:discussion}.

\item {\bf Safeguards}
    \item[] Question: Does the paper describe safeguards that have been put in place for responsible release of data or models that have a high risk for misuse (e.g., pre-trained language models, image generators, or scraped datasets)?
    \item[] Answer: \answerNA{}
    \item[] Justification: The paper releases neither trained models nor scraped datasets; rather, it provides code for implementing a statistical hypothesis test and a causal discovery wrapper. These artifacts do not introduce misuse risks beyond the general considerations discussed in the broader impacts statement.

\item {\bf Licenses for existing assets}
    \item[] Question: Are the creators or original owners of assets (e.g., code, data, models), used in the paper, properly credited and are the license and terms of use explicitly mentioned and properly respected?
    \item[] Answer: \answerYes{}
    \item[] Justification: Creators and licences are stated in the Software paragraph of Appendix~\ref{app:implementation} and in the \texttt{NOTICE} file of the code. We use the published ALARM \citep{beinlich1989alarm} and HAILFINDER \citep{abramson1996hailfinder} graph structures, from which all data are simulated, and the Sachs data \citep{sachs2005causal} as distributed by the bnlearn repository \citep{scutari2010learning}; our code downloads these data rather than redistributing them, because bnlearn's site licence (CC BY-SA) does not state terms for its data files. Software: \texttt{causal-learn} v0.1.3.7 \citep{zheng2024causallearn} (MIT) for KCI \citep{zhang2011kernel} and MVPC \citep{tu2019causal}; scikit-learn \citep{pedregosa2011scikit}, numpy, scipy and pandas (BSD-3); LightGBM (MIT) and XGBoost (Apache-2.0), with versions pinned in \texttt{environment.yml}. The FZ-Rubin baseline ports the R package \texttt{micd} \citep{witte2022multiple} (GPL-3), and that file is released under GPL-3. The linear NOTEARS solver \citep{zheng2018dags} is modified from the Apache-2.0 reference code, whose licence and our changes are included. The MissDAG baseline \citep{gao2022missdag} adapts its E-step from the authors' public code, which specifies no licence, so that module is not redistributed.

\item {\bf New assets}
    \item[] Question: Are new assets introduced in the paper well documented and is the documentation provided alongside the assets?
    \item[] Answer: \answerYes{}
    \item[] Justification: The paper releases new code: the general and fast implementations of PAIR-CI (\texttt{src/ci\_oracle.py} and \texttt{src/ci\_oracle\_fast.py}), the PC integration with a CI oracle (\texttt{src/pc\_integration.py}), data-generating utilities (\texttt{src/dgps.py}), and the full replication pipeline. The supplementary material includes a README documenting setup, an \texttt{environment.yml} specifying dependencies, a SLURM-orchestrated master pipeline (\texttt{reproduce\_all.sh}), and per-experiment SLURM submission scripts in \texttt{slurm/} for rerunning individual experiments. The code is released under the MIT License, except for one GPL-3 file (the \texttt{micd} port) and one Apache-2.0 file (the modified NOTEARS solver), as listed in \texttt{NOTICE}; the MissDAG-derived module and the Sachs data are not redistributed, and a script downloads the data.

\item {\bf Crowdsourcing and research with human subjects}
    \item[] Question: For crowdsourcing experiments and research with human subjects, does the paper include the full text of instructions given to participants and screenshots, if applicable, as well as details about compensation (if any)?
    \item[] Answer: \answerNA{}
    \item[] Justification: The paper does not involve crowdsourcing or research with human subjects.

\item {\bf Institutional review board (IRB) approvals or equivalent for research with human subjects}
    \item[] Question: Does the paper describe potential risks incurred by study participants, whether such risks were disclosed to the subjects, and whether Institutional Review Board (IRB) approvals (or an equivalent approval/review based on the requirements of your country or institution) were obtained?
    \item[] Answer: \answerNA{}
    \item[] Justification: The paper does not involve human subjects research.

\item {\bf Declaration of LLM usage}
    \item[] Question: Does the paper describe the usage of LLMs if it is an important, original, or non-standard component of the core methods in this research? Note that if the LLM is used only for writing, editing, or formatting purposes and does \emph{not} impact the core methodology, scientific rigor, or originality of the research, declaration is not required.
    \item[] Answer: \answerNA{}
    \item[] Justification: LLMs are not part of the paper's methodology. The proposed CI test (PAIR-CI), the bridge theorem, the PC integration, and the experiments involve no LLM components.

\end{enumerate}

\end{document}